\documentclass[lettersize,journal]{IEEEtran}
\pdfoutput=1 
\usepackage{amsmath,amsfonts,dsfont}
\usepackage{algorithmic}
\usepackage{algorithm}
\usepackage{array}
\usepackage[caption=false,font=normalsize,labelfont=sf,textfont=sf]{subfig}
\usepackage{textcomp}
\usepackage{amssymb}
\usepackage{stfloats}
\usepackage{url}
\usepackage{verbatim}
\usepackage{graphicx, xcolor}
\usepackage{cite}
\usepackage{mathtools}
\mathtoolsset{showonlyrefs}
\usepackage{soul}
\usepackage{comment}
\newtheorem{rem}{Remark}
\newtheorem{theorem}{Theorem}

\newtheorem{prop}{Proposition}
\newtheorem{definition}{Definition}[section]
\newtheorem{lemma}{Lemma}

\newtheorem{corollary}{Corollary}[theorem]

\newenvironment{proof}[1]{\medskip\par\noindent{\bf Proof:\,}\,#1}{{\mbox{\,$\blacksquare$}\par}}

\newcommand{\agesubs}[4]{{\Delta}(\left[\begin{smallmatrix} #1 & #2 \\ #3 & #4 \end{smallmatrix}\right])}
\newcommand{\agenonsubs}[4]{{\Delta}^\dagger(\left[\begin{smallmatrix} #1 & #2 \\ #3 & #4 \end{smallmatrix}\right])}
\newcommand{\ageongraph}[4]{{\Delta}_{\Vec{G}}(\left[\begin{smallmatrix} #1 & #2 \\ #3 & #4 \end{smallmatrix}\right])}
\newcommand{\agesubswo}[4]{\bar{\Delta}(\left[\begin{smallmatrix} #1 & #2 \\ #3 & #4 \end{smallmatrix}\right])}
\newcommand{\agenonsubswo}[4]{\bar{\Delta}^\dagger(\left[\begin{smallmatrix} #1 & #2 \\ #3 & #4 \end{smallmatrix}\right])}
\newcommand{\ageongraphwo}[4]{\bar{\Delta}_{\Vec{G}}(\left[\begin{smallmatrix} #1 & #2 \\ #3 & #4 \end{smallmatrix}\right])}

\begin{document}

\title{Age of Coded Updates In Gossip Networks Under Memory and Memoryless Schemes}

\author{Erkan Bayram, Melih Bastopcu, Mohamed-Ali Belabbas, Tamer Ba\c{s}ar~\IEEEmembership{Life Fellow,~IEEE}
        % <-this % stops a space
%\thanks{This paper was produced by the IEEE Publication Technology Group. They are in Piscataway, NJ.}% <-this % stops a space
\thanks{
% Manuscript received April 19, 2021; revised August 16, 2021. Research of the authors was supported in part by the ARO MURI Grant AG285 and in part by ARL-DCIST Grant AK868.
{A part of this paper, \cite{bayram2024age}, is presented at the  Asilomar Conference on Signals, Systems, and Computers, Pacific Grove, California, October 27-30, 2024.}}
%\thanks{The authors are with the Coordinated Science Laboratory, University of Illinois Urbana-Champaign, Urbana, IL, 61801. E-mail: \emph{(ebayram2,bastopcu,belabbas,basar1)@illinois.edu}}
\thanks{E. Bayram, M.-A. Belabbas, and T. Başar are with the Coordinated Science Laboratory, University of Illinois Urbana-Champaign, Urbana, IL, 61801. E-mail: \emph{(ebayram2,belabbas,basar1)@illinois.edu}. {M. Bastopcu is with the Department of Electrical and Electronics Engineering at Bilkent University, Ankara, Turkey, 06800. E-mail: \emph{bastopcu@bilkent.edu.tr}. Research of UIUC authors was supported in part by ARO MURI Grant AG285, NSF-CCF 2106358, ARO W911NF-24-1-0105 and AFOSR FA9550-20-1-0333.} }
}

% \markboth{IEEE Transactions on Communications, ~Vol.~??, No.~?, May~2024
% }%
% {Bayram \MakeLowercase{\textit{et al.}}: Age of Coded Updates In Gossip Networks Under Memory and Memoryless Schemes}

\IEEEpubid{0000--0000/00\$00.00~\copyright~2024 IEEE}

\maketitle

\begin{abstract}
We consider an information update system on a gossip network, where a source node encodes information into $n$ total keys such that any subset of at least $k+1$ keys can fully reconstruct the original information. This encoding process follows the principles of a $k$-out-of-$n$ threshold system. The encoded updates are then disseminated across the network through peer-to-peer communication. We have two different types of nodes in a network: subscriber nodes, which receive a unique key from the source node for every status update instantaneously, and nonsubscriber nodes, which receive a unique key for an update only if the node is selected by the source, and this selection is renewed for each update. For the message structure between nodes, we consider two different schemes: a memory scheme (in which the nodes keep the source's current and previous encrypted messages) and a memoryless scheme (in which the nodes are allowed to only keep the source's current message). We measure the \emph{timeliness} of information updates by using a recent performance metric, the version age of information. We present explicit formulas for the time average AoI in a scalable homogeneous network as functions of the network parameters under a memoryless scheme. Additionally, we provide strict lower and upper bounds for the time average AoI under a memory scheme. 
\end{abstract}


\section{Introduction}
\IEEEPARstart{I}{n} certain communication networks, the dissemination of a single piece of information is insufficient for full message reconstruction~\cite{shah2009gossip,bayram2023vector}; that is, a node in the networks needs to collect a multitude of messages or observations generated at the same time to construct meaningful information. This scenario arises in applications such as cryptography, error correction codes, and multi-sensor measurements, which will be explored in this paper. We call a system in which any subset of at least $k+1$ out of $n$ messages can fully reconstruct the original information a {\em $k$-out-of-$n$ threshold system}. 

% \mbnote{Is it known as ``$k$-out-of-$n$ system" or do we denote as ``$k$-out-of-$n$ system"? }

In this work, we consider a gossip network, where the source disseminates coded updates using a $k$-out-of-$n$ threshold system over peer-to-peer connections. We consider two types of nodes in the network: {\em subscriber nodes}, which receive a unique key from the source node instantaneously for every status update, and {\em nonsubscriber nodes}, which receive a unique key from the source node only if they are selected for a particular update. The selection of nonsubscriber nodes that receive a key is done independently for each status update. Upon receiving an update, the nodes start to share their local messages with their neighboring nodes to decrypt the source's update. The nodes that collects $k+1$ different messages of the same update can decode the source's update. For the communication structure between nodes, we consider two different settings: a {\em memory scheme}---where the nodes are able to keep the keys received from the source--- and a {\em memoryless scheme}---where the nodes are allowed to keep only keys associated with the source's current status. For both of these settings, we study the information freshness achieved by the receivers, which informally refers to how up-to-date the information at the receiver is compared to the most recent update from the source. Such systems find applications in various domains including cryptography, robotics, and communication systems. 

For example, in cryptography, the information source can apply Shamir's secret sharing method~\cite{shamir1979share} or Blakley's method~\cite{blakley1979safeguarding} on the information to put it into $n$ distinct keys such that any subset of $n$ keys with $k+1$ cardinality would be sufficient to decode the encrypted message. Such a scheme is known as $(k,n)$-Threshold Signature Scheme (TSS)~\cite{hwang2003practical}. 

Furthermore, using $k$-out-of-$n$ threshold systems can enhance error correction performance in data transmission. Integrating multiple sensor observations minimizes information errors. For instance, in robotics, simultaneous localization and mapping (SLAM) fuse measurements from cameras, LiDAR, or inertial measurement units to estimate the robot's pose; another example is the determination of the relative position of a target object using the Time Difference of Arrival (TDoA) technique~\cite{gustafsson2003positioning}, for which one must have at least three distinct Time of Flight (ToF) measurements, while more measurements increase the accuracy of detection. 
Another common example of $k$-out-of-$n$ threshold systems that is used in communication and distributed storage is $(n,k+1)$-MDS error correction codes~\cite{buyukates2020timely}. For example, in Reed–Solomon error correcting algorithm~\cite{wicker1999reed}, which is an example of $(n,k+1)$-MDS coding, while any $k+1$ codeword is sufficient to decode the message (less than $k+1$ keys cannot decode the actual information), having more redundant codewords (keys) allows for more error corrections in the message.  
\IEEEpubidadjcol

% For example, consider a sensor network measuring the Time of Flight (ToF) of a signal coming from a target object, where gathering at least three simultaneously generated observations is crucial for detecting the relative position of the target object using the Time Difference of Arrival (TDoA) technique~\cite{gustafsson2003positioning}. As another example, consider a peer-to-peer communication network with unsecured communication channels; the information can be decomposed into $n$ distinct pieces and spread throughout the network. A node in the network can only reconstruct the information if it collects $k+1$ distinct pieces of the actual information. This collaborative approach enhances the reliability and security of information exchange in dynamic and distributed environments~\cite{varadharajan1991analysis}. Such a scheme is known as $(k,n)$-Threshold Signature Scheme (TSS)~\cite{hwang2003practical}. For example, the information source can apply Shamir's secret sharing method~\cite{shamir1979share} or Blakley's method~\cite{blakley1979safeguarding} on the information to put it into $n$ distinct keys such that any subset of $n$ keys with $k+1$ cardinality would be sufficient to decode the encrypted message. 

% width=0.90\columnwidth

% Motivated by these applications, in this work, we consider an information source that generates updates and then encrypts them by using $(k,n)$-TSS.
% Moreover, the inclusion of both subscriber and nonsubscriber nodes in our analysis provides a perspective on the consequences of communication faults between the source and the nodes. 

Motivated by these applications, we consider in this work an information source that generates updates and then encodes them by using a $k$-out-of-$n$ threshold systems. For the sake of simplicity, we consider the source using $(k,n)$-TSS, but it is worth noting that our results are applicable to {\em any $k$-out-of-$n$ threshold system}, some examples are discussed above. 

% In this work, we consider two types of nodes in a network: {\em subscriber nodes}, which receive a unique key from the source node instantaneously for every status update, and {\em nonsubscriber nodes}, which receive a unique key from the source node only if they are selected for a particular update. The selection of nonsubscriber nodes that receive a key is done independently for each status update.

% Upon receiving an update, the nodes start to share their local messages with their neighboring nodes to decrypt the source's update. The nodes that collects $k+1$ different messages of the same update can decode the source's update. 

% For the communication structure between nodes, we consider two different settings: a memory scheme where the nodes are able to keep the keys that receive from the source and a memoryless scheme where the nodes are allowed to keep only the source's current message if it gets a key from the source. For both of these settings, we study the information freshness achieved by the receivers as a result of applying $(k,n)$-TSS.
% {\em for different number of subscriber and nonsubscriber nodes.} 

In order to measure the freshness of information in communication networks, age of information (AoI) has been introduced as a new performance metric~\cite{Kaul12a}. Since then, AoI has been considered in queueing networks \cite{Soysal2021}, energy harvesting problems \cite{Bacinoglu15, Arafa2020}, caching systems \cite{Bastopcu2021}, remote estimation \cite{Sun2020}, distributed computation systems \cite{buyukates2020timely}, and RF-powered communication systems \cite{Elmagid2020}. A more detailed review of the literature on AoI can be found in \cite{Yates2021}. The traditional age metric {is defined as the time that has passed since the most recent status update packet at the receiver was generated. Thus, the traditional age metric} increases linearly over time until the receiver gets a new status update from the source. However, if the information at the source does not change frequently, although the receiver may not get updates from the source for a long time, it may still have the most up-to-date information prevailing at the source. The traditional age metric is said to be {\em source-agnostic},  meaning that it does not consider the information change rate at the source. Considering this problem, new \textit{source-aware} freshness metrics such as the version age of information {(VAoI)}~\cite{yates2021age, Eryilmaz2021}, the binary freshness~\cite{Brewington00}, and age of incorrect information~\cite{MaatoukAoII} have been considered recently which can achieve \textit{semantic-oriented} communication~\cite{Uysal2022semtantic}.  

In the earlier works on AoI, age has been studied for simple communication networks where the information flows through directly from the source, or other serially connected nodes (as in multi-hop multi-cast networks \cite{Zhong17a}). The development of the stochastic hybrid system (SHS) approach in \cite{Yates2019} paves a new way to calculate the average age in arbitrarily connected networks. In particular, reference \cite{yates2021age} considers a setting where in addition to source sending updates to the receiver nodes, nodes also share their local updates through \textit{a gossiping mechanism} to improve their freshness further. Inspired by  \cite{yates2021age}, reference \cite{Buyukates2022} improves the AoI scaling by using the idea of clustered networks and \cite{Bastopcu2021gossip} studies scaling of binary freshness metric in gossip networks. In the aforementioned works in \cite{yates2021age, Buyukates2022, Bastopcu2021gossip}, age scaling has been considered in disconnected, ring, and fully-connected symmetric gossip networks. Different from these symmetric gossip networks, age scaling has been studied in a grid network in \cite{mitra2023ageaware}. Reference \cite{delfani2023version} optimizes age in a gossip network where updates are provided through an energy harvesting sensor. In gossiping, as the information exchanges happen through comparison of time-stamp of information, the gossip networks are vulnerable to adversarial attacks. The \textit{timestomping}, i.e., changing time-stamp of the updates to brand the outdated information as fresh, has been studied in gossip networks in \cite{Kaswan2024timestomping}. The reliability of information sources \cite{kaswan2023choosing}, the binary dynamic information dissemination \cite{Bastopcu2024gossip}, information mutation and spread of misinformation \cite{kaswan2023information} have been studied in gossip networks. More recent advances of AoI in gossip networks can be found in \cite{kaswan2023age}.            
\begin{figure}[t]
	\centering  	\includegraphics[scale=0.65]{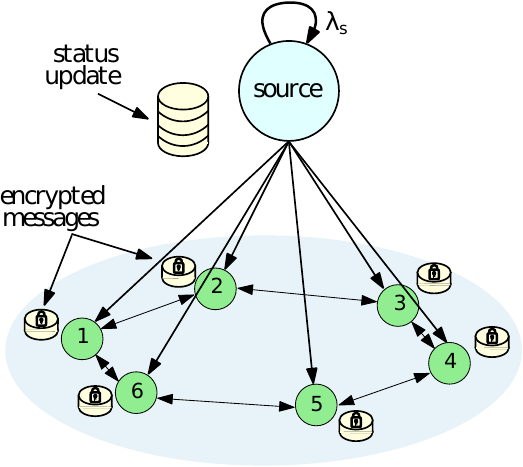}
	\caption{A gossip network consisting of a source and $n=6$ receiver nodes. Source encrypts status updates and sends them to nodes.}
	\label{fig:main_model}
 \vspace{-0.5cm}
\end{figure}

% We study two different communication settings where ($i$) the nodes have memory, in which case they can hold the current and also the previous keys received from the source, and ($ii$) the nodes do not have memory, in which case they can only hold the keys from the most current update. 
% and analyze how the subscription of a node to the source for a limited number of key (total $n$) affects the freshness of information for other nodes in the network,
% In this work, we consider the time average of $k$-keys version age for nodes under memory and memoryless schemes. 

In all these aforementioned works, the source sends its information to gossip nodes without using any encryption. In this work, we consider the version age of information in a gossip network, where the source disseminates coded updates using $(k,n)$-TSS encryption (or any $k$-out-of-$n$ threshold system). The main contributions of our work are as follows: 
\begin{itemize}
    \item We measure the freshness of the information through $k$-keys VAoI on subscriber and nonsubscriber nodes for both memory and memoryless schemes.
    \item For nodes with memory, we derive the closed-form average age expression for an arbitrary \emph{non-homogeneous} $(k,n)$- TSS feasible network in full subscription where {\em every node} receives a key from the source once a status update is generated. 
    % \item We have introduced (a) total key subscription where only subscriber nodes get a unique key from the source (b) partial key subscription where source send keys to subscribers node and some randomly selected nonsubscriber nodes. 
    \item In total key subscription where {\em only subscriber nodes} get a unique key from the source, we derive a closed-form expression for the average $k$-keys version age under both memory and memoryless schemes
    \item In partial key subscription where the source sends keys to {\em subscriber nodes and some randomly selected nonsubscriber nodes}, we provide tight upper and lower bounds for the average $k$-keys version age under the memory scheme and we derive a closed-form expression for the memoryless scheme.
    \item We provide an asymptotic bound on the average $k$-keys version age as the network size and total key numbers proportionally grow for the scalability of the network.
\end{itemize}

Furthermore, we show that assigning a node as a subscriber reduces its $k$-keys version age, while significantly increasing the version age for the remaining nonsubscriber nodes under the memoryless scheme. To emphasize this, a closed-form expression {is derived} to quantify the impact of assigning a node as a subscriber in terms of the overall age of the network. 

Finally, we present extensive numerical simulations to validate our theoretical findings. In the numerical results, we show that the time average of the $k$-keys version age for a node in both schemes decreases as the edge activation rate increases, or as the number of keys required to decode the information decreases, or as the number of gossip pairs in the network increases. We show that a memory scheme yields a lower time average of $k$-keys version age compared to a memoryless scheme. However, the difference between the two schemes diminishes with infrequent source updates, frequent gossip between nodes, or a decrease in $k$ for a fixed number of keys $n$. Finally, this result motivates us to provide a more detailed discussion on the value of the memory, where we provide an analytical result on how to choose gossip rates between nodes so that the difference between two schemes can be made arbitrarily small. 

% Additionally, for the memory scheme, we provide simulation results showing that as the gossip rate increases, the number of subscribers becomes insignificant for the average age over the network.  

% The balance of paper is organized as follows: In Sec.~\ref{sec:system_model}, we introduce the system model, the metric that we measure for information freshness and . In Section~\ref{sec:prelim}, we provide preliminaries on the order statistic with its notation. Then, we have done the age analysis under {\em memory} and {\em memoryless} In Secion~\ref{sec:w_memory} and~\ref{sec:wo_memory},respectively. In Section~\ref{sec:discuss}, we have numerical results and discussion. In Sec.~\ref{sec:value}, we discuss the value of memory for informaiton freshenss. The paper is conclused Sec.\ref{sec:conl}.

The remainder of the paper is organized as follows: In Sec.~\ref{sec:system_model}, we introduce the system model and the metric used to measure information freshness. In Sec.~\ref{sec:prelim}, we provide preliminaries on order statistics, along with relevant notation. The age analysis under memory and memoryless schemes is presented in Sec~\ref{sec:w_memory} and Sec.~\ref{sec:wo_memory}, respectively. Sec.~\ref{sec:discuss} includes numerical results and discussions. In Sec.~\ref{sec:memory_value}, we evaluate the value of memory in maintaining information freshness. Finally, the paper concludes in Sec.~\ref{sec:concl}.

\section{System Model and Metric}\label{sec:system_model}

We consider an information updating system consisting of a single source, which is labeled as node $0$, and $m$ receiver nodes. We denote by $\Vec{G}=(V,\Vec{E})$ the directed graph with node set $V$ and edge set $\Vec{E}$. We let  $\Vec{G}$ represent the communication network according to which nodes exchange information. We denote the set of nodes except the source node by ${V}^*$. Let $\mathcal{N}^+_j=\{ i\in V^* , e_{ij} \in \Vec{E} \}$ be the set of nodes with out-bound connections to node $j$; denote its cardinality by $n_j$. If there is a directed edge $e_{ij}\in\vec{E}$, that is, $i \in \mathcal{N}^+_j$; we call node $j$ {\em gossiping neighbor} of node $i$. 

The information at the source is updated at times distributed according to a Poisson counter, denoted by $N_0(t)$, with rate $\lambda_s$. We refer to the time interval between $\ell$th and $(\ell+1)$th information updates (messages) as the {\em $\ell$th version cycle} and denote it by $U^\ell$. Each update is stamped by the current value of the process $N_0(t)=\ell$ and the time of the  $\ell$th update is labeled $\tau_\ell$ once it is generated. The stamp $\ell$ is called {\em version-stamp} of the information. Let $T=\{\tau_\ell\}_{\ell=1}^\infty$ be the monotonically increasing sequence of times when the status updates occur at the source node $0$, with $\tau_0:=0$. 

We call a communication network {\em $(k,n)$-TSS feasible} if the node $0$ has out-bound connections to all other nodes and the smallest in-degree of the receiver nodes is greater than $k$. We illustrate a $(2,6)$-TSS feasible network in Fig.~\ref{fig:main_model}. We assume that the source instantaneously encrypts the information update by using $(k,n)$-TSS (or any other $k$-out-of-$n$ threshold system) once it is generated. To be more precise, we assume that the source puts the information update into $n$ distinct keys, out of which $k+1$ are sufficient to decode the status update. To build upon our previous work~\cite{bayram2024age}, in this study, we consider two types of receiver node: subscriber nodes to the source and nonsubscriber nodes. Let $\mathcal{S}(\subseteq V^*)$ and $\mathcal{S}^\dagger(\subseteq V^*)$ be the sets of subscriber nodes and nonsubscriber nodes, respectively. The source sends a unique key to each subscriber node in the set $\mathcal{S}$ at $\tau_\ell$ for all $\ell$. Then, the source node sends the remaining $n-|\mathcal{S}|$ keys to nodes which are selected uniformly at random  from the set $\mathcal{S}^\dagger$ at time $\tau_\ell$. We denote the set of nodes selected at $\tau_\ell$ as $\mathcal{M}^\ell(\subseteq \mathcal{S}^\dagger)$. The source nodes selects the set $\mathcal{M}^\ell$ independently for each $\ell$. The set of nodes that receive a key directly from the source at $\tau_\ell$, denoted by $\mathcal{K}^\ell$, is defined as $\mathcal{K}^\ell:=\!\mathcal{S}\!\cup\!\mathcal{M}^\ell$ (the union of subscriber nodes and randomly selected nonsubscriber nodes at $\tau_\ell$).
% Then, by definition, we have $|\mathcal{S}| \leq n$.

Each node wishes that its knowledge of the source is as timely as possible. The timeliness is measured for an arbitrary node $j$ by the difference between the latest version of the message at the source node, $N_0(t)$, and the latest version of the message which can be {\em decrypted} at node $j$, denoted by $N^k_j(t)$. This metric has been introduced as {\em version age of information} in~\cite{yates2021age, Eryilmaz2021}. We call it  {\em $k$-keys version age of node $\!j$} at time $\!t\!$ and denote it as \begin{equation}%\label{eqn:defn_process}
    A_j(k,t) := N_0(t) - N^k_j(t).
\end{equation}

Recall that in the $(k,n)$-TSS, a node needs to have  $k+1$ keys with the version stamp $\ell$ in order to decrypt the information at the source generated at $\tau_\ell$. Therefore, we consider a $(k,n)$-TSS feasible network, in which, nodes are allowed to communicate and share {\em only} the keys that are received from the source with their {\em gossiping neighbor}. We do not allow the nodes to share the keys that they gather from the other nodes ensuring that any set of messeage on a single channel is not sufficient to decode the status update at any time. The edge $e_{ij}$, connecting node $i$ to node $j$, is activated at times distributed according to the Poisson counter $N^{ij}(t)$, which has a rate $\lambda_{ij}$. All counters are pairwise independent. Once the edge $e_{ij}$ is activated, node $i$ sends a message to node $j$, instantaneously. This process occurs under two distinct schemes: {\em with memory scheme} and {\em memoryless scheme}.

\begin{figure}[t]
    \centering
    \includegraphics[width=1\linewidth]{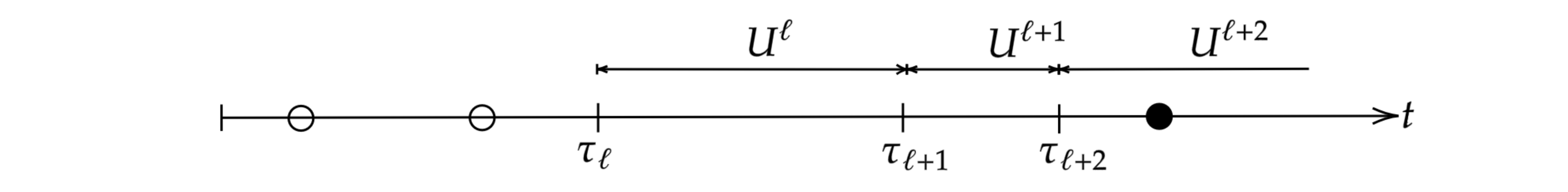}
    \vspace{-0.55cm}
     \caption{Sample timeline of the source update and the edge $e_{ij}$ activation. The last activation of $e_{ij}$ is marked by ($\bullet)$ and the previous activations of $e_{ij}$ are marked by ($\circ$).}\label{fig:timeline_for_schemes}
     \vspace{-0.55cm}
\end{figure}

In the memory scheme, nodes can store (and send) the keys of the previous updates. For example, if the edge $e_{ij}$ is activated at time $t$, node $i$ sends node $j$ {\em all} the keys that the source has sent to node $i$ since the last activation of $N^{ij}(t)$ before time $t$. For the illustration in Fig.~\ref{fig:timeline_for_schemes}, node $i$ sends the set of keys with the versions $\{\ell,\ell+1,\ell+2\}$ to node $j$ in the memory scheme. Note that this can be implemented by finite memory in a finite node network with probability $1$~\cite{bayram2024age}. In the memoryless scheme, nodes have no memory and only store the latest key obtained from the source. If the edge $e_{ij}$ is activated at time $t$, node $i$ sends node $j$ {\em only} the key that belongs to the most up-to-date information prevailing at the source at time $t$. It is worth noting that if a node does not receive a key from the source for the most up-to-date information, for example, a nonsubscriber node that is not selected to the set $\mathcal{M}^\ell$ for the latest version $\ell$, the node does not transmit any message in the memoryless scheme during the update cycle $U^\ell$. Referring again to the illustration in Fig.~\ref{fig:timeline_for_schemes}, node $i$ in this case sends only the key with the version $\{\ell+2\}$ to node $j$ (under the assumption of $i\in\mathcal{M}^{\ell+2}$).

% \mbnote{Erkan, now considering our definitions of subscribers and nonsubscribers, what if the node is nonsubscriber and did not get the current message from the source. In that case, nonsubscribers do not have any message to share in memoryless scheme, right?}
\begin{comment}
\begin{figure} \centering
\begin{center}
    \begin{subfigure}[b]{\linewidth}
    \centering
        \includegraphics[width=0.5\linewidth]{figures/path_w_memory_v3.eps}
        \label{fig:a}
        \caption{}
    \end{subfigure} %

    \begin{subfigure}[b]{\linewidth}   
        \centering

        \includegraphics[width=0.5\linewidth]{figures/path_wo_memory_v3.eps}
        \label{fig:b}  
            \caption{}

    \end{subfigure} 
    \vspace{-0.4cm}
\caption{Sample path of the $k$-keys version age (a) $A(k,t)$ for a node with memory and (b) $\bar{A}(k,t)$ for a node without memory.}
\label{fig:path_w_memory}
    \end{center}
    \vspace{-0.7cm}
\end{figure}
\end{comment}
\begin{figure}[!t]
\centering
\subfloat[]{\includegraphics[width=0.85\linewidth]{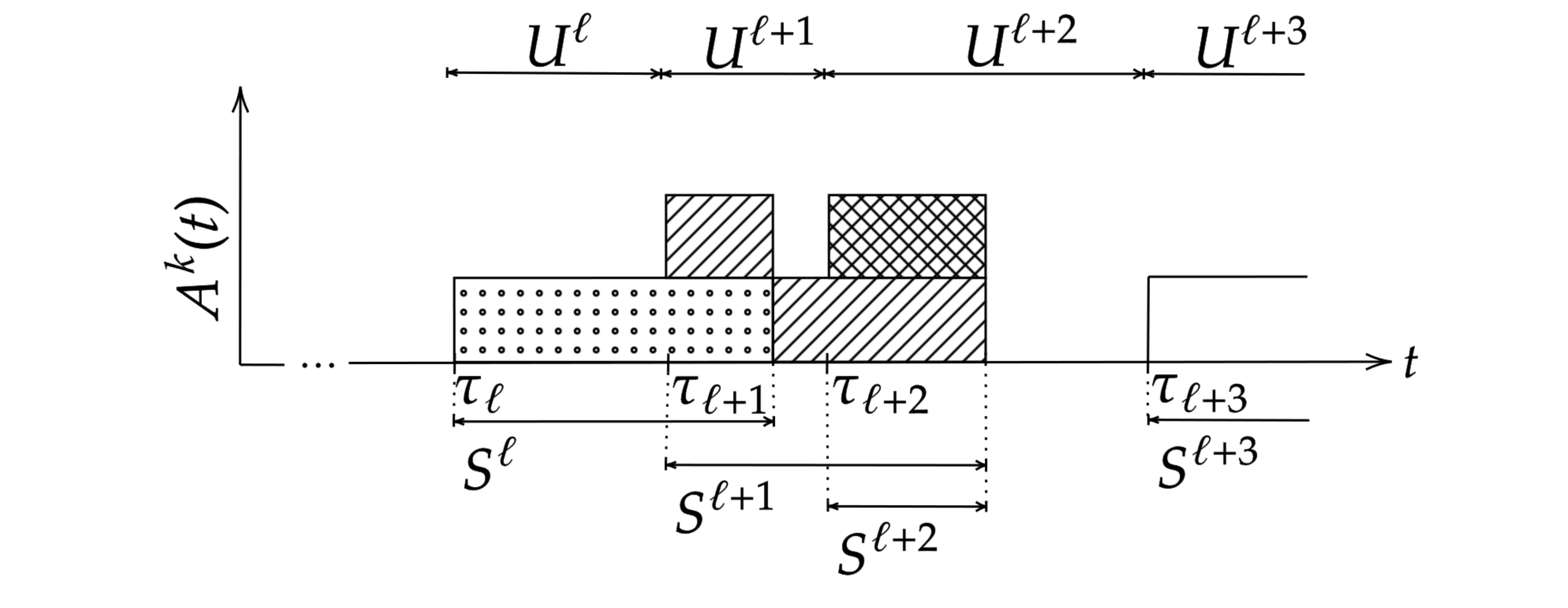}%
\label{fig:a}}
\hfil
\subfloat[]{\includegraphics[width=0.85\linewidth]{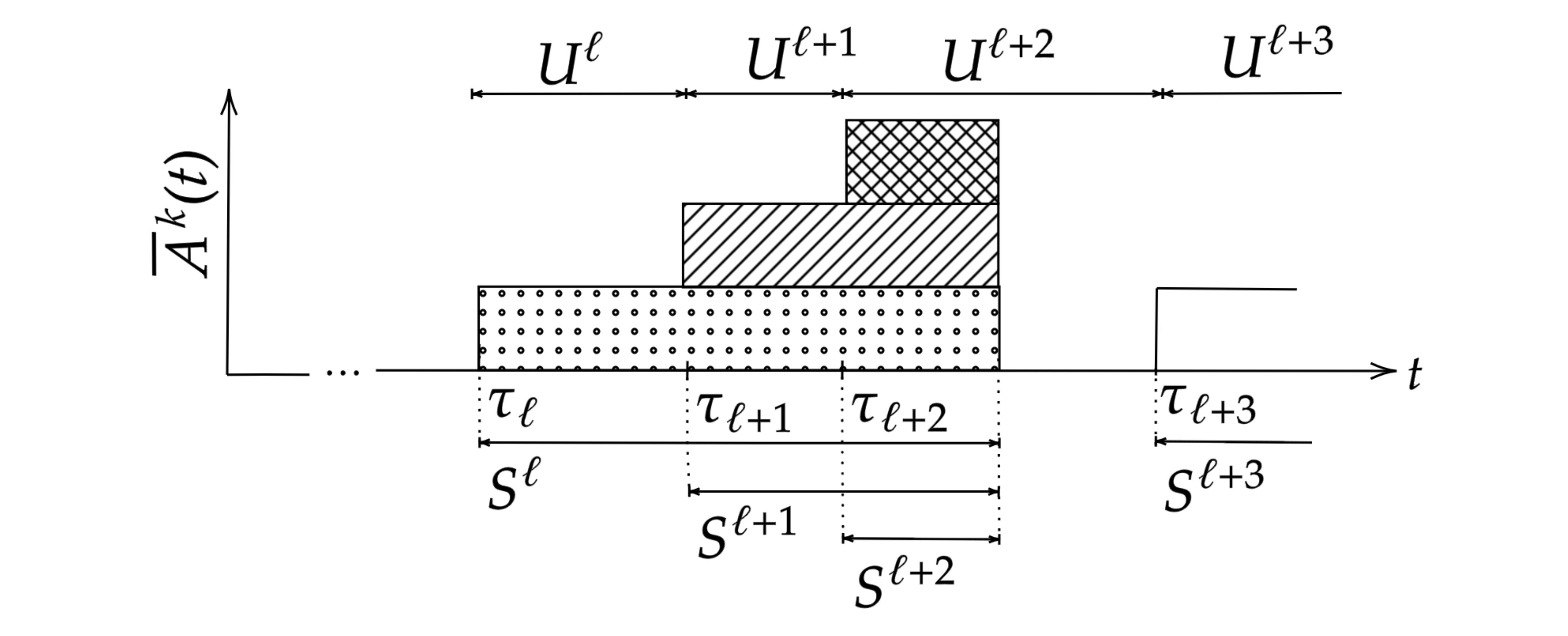}%
\label{fig:b}}
    \vspace{-0.2cm}
\caption{Sample path of the $k$-keys version age (a) $A(k,t)$ for a node with memory and (b) $\bar{A}(k,t)$ for a node without memory.}
\label{fig:path_w_memory}
\vspace{-0.5cm}
\end{figure}

% \begin{figure}[t]
% 	\begin{center}
% 		\subfigure[]{%
% 			\includegraphics[width=0.75\linewidth]{figures/path_w_memory_v3.png}}
% 		\subfigure[]{%
% 			\includegraphics[width=0.75\linewidth]{figures/path_wo_memory_v3.png}}
% 	\end{center}
%  \vspace{-5mm}
% 	\caption{Sample path of the $k$-keys version age (a) $A^k(t)$ for a node with memory and (b) $\bar{A}^k(t)$ for a node without memory.}
% 	\label{fig:path_w_memory}
% 	\vspace{-6mm}
% \end{figure} 

Fig.~\ref{fig:path_w_memory}(a) and Fig.~\ref{fig:path_w_memory}(b) depict the sample path of the $k$-keys version age process ${A}(k,t)$ (resp.~ $\bar{A}(k,t)$) for a node with memory (resp.~without memory). It is worth noting that we associate the notation $\bar{\cdot}$ with the memoryless scheme. It is assumed that edge activations and source updates occur at the same time in both schemes in the figures. In the memory scheme, we define {\em the service time of information with version $\ell$ to an arbitrary node $j$}, denoted by $S^\ell_j$, as the duration between $\tau_\ell$ and the time when node $j$ can decrypt the information with version $\ell$, as shown in Fig.~\ref{fig:path_w_memory}(a). In the memoryless scheme, a node can miss an information update with version $\ell$, unlike in the memory scheme, if it cannot get $k$ more distinct keys before the next update arrives at $\tau_{\ell+1}$. Thus, for a node without memory, we define $S^\ell_j$ as the duration between $\tau_\ell$ and the earliest time when the node can decrypt information with a version of at least $\ell$. In Fig.~\ref{fig:path_w_memory}(b), the node misses the information with versions $\ell$ and $\ell\!+\!1$. It can only decode the information with version $\ell\!+\!2$. %Therefore, 
Hence, the service times $\!S^\ell\!$ and $\!S^{\ell+\!1\!}\!$ end at the same time as the service time $\!S^{\ell+2}\!$. %ends. 

% If we have a fully connected directed graph $\vec{G}$ on $n+1$ nodes (including the source node) with $\lambda_{ij}=\frac{\lambda_e}{(n-1)}$ for all edges $e_{ij}$ in $\vec{E}$, we call the graph $\vec{G}$ {\em scalable homogeneous network} and $\lambda_e$ is the {\em gossip rate}. In the second part of this work, we extend our investigation to scalable homogeneous networks, introducing two approaches alongside the previously discussed memoryless and memory schemes: namely, {\em subscription} and {\em random selection}. The source node transmits information updates through distinct mechanisms. The first approach involves sending updates to uniformly selected $m$ nodes among $n$ nodes. It is called {\em $m$ out of $n$ selection}. The second approach involves consistently providing the latest status updates to a specific (pre-selected) subset of nodes with cardinality, called $r$ subscriber to the source and sending updates to uniformly selected $m-r$ nodes among $n-r$ nodes that are not a subscriber to the source.. We analyze these two different approaches in both memory and memoryless schemes

In the memory scheme, let $\Delta_j(k,t)$ be the total $k$-keys version age of node $j$, defined as the integrated $k$-keys version age of node $j$, $A_j(k,\tau)$, until time $t$ which is given by:
\begin{align}\label{eqn:defn_time_av}
    \Delta_j(k,t) := \int_{0}^{t} A_j(k,\tau) d\tau.
\end{align}
Then, the time average of $k$-keys version age process of node $j$ is defined as follows:
\begin{align}\label{eqn:defn_time_av_2}
    \Delta_j(k) :=\lim_{t\to \infty} \frac{\Delta_j(k,t)}{t}  = \lim_{t\to \infty} \frac{1}{t} \int_{0}^{t} A_j(k,\tau) d\tau.
\end{align}
We interchangeably call $\Delta_j(k)$ {\em the version age of $k$-keys for node $j$}. If nodes in the network have no memory, we denote {\em the version age of $k$-keys for node $j$} by $\bar{\Delta}_j(k)$, which can be similarly defined by following~\eqref{eqn:defn_time_av} and~\eqref{eqn:defn_time_av_2} for $\bar{A}^k_j(k,\tau)$.

If we have a fully connected directed graph $\vec{G}$ on $m+1$ nodes (including the source node) with $\lambda_{ij}=\frac{\lambda_e}{(m-1)}$ for all edges $e_{ij}$ in $\vec{E}$, we call the graph $\vec{G}$ {\em scalable homogeneous network} (SHN) on $m$ nodes and $\lambda_e$ is the {\em gossip rate}.  We consider a scalable homogeneous network on $m$ nodes with $|\mathcal{S}|=s$ subscriber nodes on the graph $\vec{G}$ to implement $(k,n)$-TSS in {\em with memory} and {\em without memory} schemes. 

On a{n} SHN, the processes $A_j(k,t)$ are statistically identical for the subscriber nodes, $j \in \mathcal{S}$. Thus, we denote them by ${A}(k,t)$ for subscriber nodes. Similarly, the processes $A_j(k,t)$ are statistically identical for %nodes 
the nonsubscriber nodes, $j \in \mathcal{S}^\dagger$. Thus, we denote them by ${A}^\dagger(k,t)$ for nonsubscriber nodes. To emphasize the number of subscribers $s$ and the total number of nodes $m$, we denote the version age of $k$-keys for any subscriber node $j\!\in\!\mathcal{S}$ by $\agesubs{k}{m}{s}{n}$ and for any nonsubscriber node $j\!\in\!\mathcal{S}^\dagger$ by $\agenonsubs{k}{n}{s}{m}$. We define the average 
version age of $k$-keys over the graph $\Vec{G}$, denoted by $\ageongraph{k}{n}{s}{m}$, as follows:
\begin{align}\label{eqn:age_on_graph}
     \ageongraph{k}{n}{s}{m} := \frac{|\mathcal{S}|}{|V^*|} \agesubs{k}{n}{s}{m} +\frac{|\mathcal{S}^\dagger|}{|V^*|}  \agenonsubs{k}{n}{s}{m}. %(= \frac{r}{n} \agesubs{k}{n}{r}{m} +\frac{m-r}{m}  \agenonsubs{k}{n}{r}{m} )
\end{align}
If nodes in the network have no memory, we denote the version age of $k$-keys for a subscriber node by $\agesubswo{k}{n}{s}{m}$ and for a nonsubscriber node by $\agenonsubswo{k}{n}{s}{m}$.

We have three different network types based on the number of subscriber nodes $s$, the number of total keys $n$ and the network size $m$, as {summarized in Table~\ref{tab:types}}.

% If all nodes in a network are subscriber nodes, i.e., $s=n$ and $m=n$, we call the network {\em in full subscription}. If there is at least one nonsubscriber node, and every generated key is sent to a subscriber node, i.e., $s=n$ and $n<m$, we call the network {\em in total key subscription}. If there is at least one nonsubscriber node, and some of the generated keys are sent to randomly selected nonsubscriber nodes, we call the network {\em in partial key subscription} in which case, we have $s<n<m$.

\begin{table}[t]
  \renewcommand{\arraystretch}{1.3}
  \vspace{-4mm}
    \caption{Network Types Based on $\{s,n,m\}$}  \vspace{-3mm}
  \label{tab:types}
  \centering
  % Some packages, such as MDW tools, offer better commands for making tables
  % than the plain LaTeX2e tabular which is used here.
  \begin{tabular}{|c||c|}
    \hline
    Type   & Relation \\
    \hline
    Full Subscription & $m=n=s$\\
    \hline
    Total Key Subscription & $m>n=s$\\
    \hline
    Partial Key Subscription & $m > n >s$\\
    \hline
  \end{tabular}
  \vspace{-5mm}
\end{table}
\section{Preliminaries on Order Statistics}\label{sec:prelim}
%\subsection{Order Statistics}
We first focus on the order statistics of a set of random variables. Consider a set of random variables $\mathcal{Y}=\{Y_i\}_{i=1}^n$. We denote the $k$th smallest variable in the set $\mathcal{Y}$ by $Y_{(k:n)}$. We call $Y_{(k:n)}$ the $k$th order statistic of $n$-samples ($k=1,2,\cdots,n$) in the set $\mathcal{Y}$. For a set of $i.i.d.$ random variables $\{ Y_i\}_{i=1}^n$that are exponentially distributed with mean $\frac{1}{\lambda}$, the expectation of the $k$th order statistic $Y_{(k:n)}$ is given by \cite{david2004order}:
\begin{equation}\label{eqn:exp_mean}
   \mathbb{E}[Y _{(k:n)}] = \frac{1}{\lambda} ( H_n - H_{n-k} ), \mbox{ where } H_n= \sum_{j=1}^n \frac{1}{j}.\\[-0.5em]
   % \\    Var[Y_{(k:n)}] =\frac{1}{\lambda^2} ( G_n - G_{n-k} ) 
\end{equation}
% where $H_n= \sum_{j=1}^n \frac{1}{j}$. % and $G_n= \sum_{j=1}^n \frac{1}{j^2}$. 
% Let $\mathcal{N}^+_j=\{ i\in V , e_{ij} \in \Vec{E} \}$ be the set of nodes with out-bound connections to node $j$; denote its cardinality by $n_j$. 

Let $X_{ij}$ be the times between successive activations of the edge $e_{ij}$. Then, $X_{ij}$ is an exponential random variable with mean $1/\lambda_{ij}$. Let $\mathcal{X}_j$ be the set of random variables $X_{ij}$, $\forall i \in \mathcal{N}^+_j$. We denote the $k$th order statistic of the set $\mathcal{X}_j$ by $\mathcal{X}_{(k:n_j)}$. Let $\mathcal{X}_j(\ell)$ be the set of random variables $X_{ij}$, $\forall i \in \mathcal{N}^+_j \cap \mathcal{K}^\ell$. We denote the $k$th order statistic of the set $\mathcal{X}_j(\ell)$ by $\mathcal{X}_{(k:\gamma_j^\ell)}(\ell)$ where $\gamma_j^\ell=|\mathcal{N}^+_j \cap \mathcal{K}^\ell|$.
\section{Age Analysis for Nodes with Memory}\label{sec:w_memory}
In this section, we consider the case when nodes in the network have memory. We first analyze a $(k,n)$-TSS feasible network in full-subscription with nonhomogenoues edge activation rates. Then, we analyze a scalable homogeneous network in total key $(s=n)$ and partial key $(s<n)$ subscriptions.

\subsection{Arbitrary $(k,n)$-TSS Feasible Network in Full Subscription}
In this subsection, we consider a $(k,n)$-TSS feasible network $\Vec{G}$ in full subscription. We first compute the service time of the information to a node with memory. Then, we provide the closed-form expression for the version age of $k$-keys for any node $j$, ${\Delta}_j(k)$, in $\vec{G}$.
\begin{lemma}\label{lem:service}
If nodes have memory in full subscription, then the service time of the information with version $\ell$ to node $j$ is the $k$th order statistic of the set of exponential random variables $\mathcal{X}_j$.   
\end{lemma}
See~\cite[Lemma 1]{bayram2024age} for the proof of Lemma~\ref{lem:service}. We are now in a position to state the main theorem.

\begin{theorem}\label{thm:hetero_w_memory_age}
Let $\vec{G}$ be a $(k,n)$-TSS feasible network. Consider an arbitrary node $j$ in $\Vec{G}$. The version age of $k$-keys for node $j$ with {\em memory} in full subscription is:
\begin{equation}
    \Delta_j(k) = \frac{\mathbb{E}[\mathcal{X}_{(k:{n}_j)}]}{\mathbb{E}[U]} \mbox{ w.p. } 1.
\end{equation}
where $U$ is the interarrival time for the source update. 
\end{theorem}
% \ebnote{Now, we have full proof in conf. version so we can refer to it}
See~\cite[Thm. 1]{bayram2024age} for the proof of Theorem~\ref{thm:hetero_w_memory_age}. 
\vspace{-3mm}
\subsection{Scalable Homogeneous Network with $s$ Subscriber}

In this subsection, we consider a scalable homogeneous network with $m+1$ nodes. We first analyze the total key subscription case $(s=n)$. Then, we focus on the partial key subscription case $(s<n)$.

\paragraph{Total Key Subscription $(s=n)$} In this network type, each generated key (total $n$) is sent to a subscriber node in the set $\mathcal{S}$. Therefore, there is no random selection among the nonsubscriber nodes to send a key at $\tau_\ell$, that is, $\mathcal{K}^\ell=\mathcal{S}$ where the cardinality of each is $n(=s)$. We have the following corollary to Theorem~\ref{thm:hetero_w_memory_age} for an SHN with $n$ subscribers.
 
\begin{corollary}\label{cor:w_memory_age}
For a{n} SHN in total key subscription, a node with {\em memory} has the following version age of $k$-keys. If the node is {\em {a} subscriber node}, then:
\begin{align*}
 \agesubs{k}{n}{n}{m} = \frac{(m-1)\lambda_s}{\lambda_e}\left(\sum_{i=n-k}^{n-1} \frac{1}{i}\right) \mbox{w.p.~} 1.
\end{align*}
If the node is {\em {a} nonsubscriber node}, then:
\begin{align*}
 \agenonsubs{k}{n}{n}{m} = \frac{(m-1)\lambda_s}{\lambda_e} \left( \sum_{i=n-k}^{n} \frac{1}{i} \right) \mbox{w.p.~} 1.
\end{align*}
%where $X$ and $U$ are is exponential random variable with mean $n/\lambda_e$ and $\lambda_s$, respectively.
\end{corollary} 
\begin{proof}
Without loss of generality, consider node $j$ in {an} SHN. In the total key subscription case, the set $\mathcal{K}^\ell$ is equal to the set $\mathcal{S}$ for any $\ell$. On the one hand, if node $j$ is a subscriber node, it receives a key from the source at $\tau_\ell$ for any $\ell$. Thus, it requires $k$ additional keys with the same version stamp to decode the status update. It can receive the remaining keys only from the nodes in the set $\mathcal{S} \setminus \{ j\}$, that is, from the set of nodes that receives unique keys (subscribers) excluding node $j$ itself. Therefore, we can assume, without loss of generality, that $\mathcal{N}^+_j$ is equal to $\mathcal{S}\setminus \{ j\}$, which implies $n_j=n-1$. Then, from Theorem~\ref{thm:hetero_w_memory_age}, we have:
\begin{align}
    \Delta_{j}(k) =  \frac{ \mathbb{E}[\mathcal{X}_{(k:n_j)}] }{\mathbb{E}[U]} =  \frac{ \mathbb{E}[\mathcal{X}_{(k:n-1)}] }{\mathbb{E}[U]}\mbox{ w.p. } 1.
\end{align}
On the other hand, if node $j$ is a nonsubscriber node, it requires $k+1$ {\em additional keys} to decode the status update and it can receive keys only from the set $\mathcal{S}$ for any status update. Thus, we consider that $\mathcal{N}_j^+$ is equal to $\mathcal{S}$, which implies $n_j=n$. Then, from Theorem~\ref{thm:hetero_w_memory_age}, we have:
\begin{align}
     \Delta_{j}(k) =   \frac{ \mathbb{E}[\mathcal{X}_{(k+1:n_j)}] }{\mathbb{E}[U]}  =  \frac{ \mathbb{E}[\mathcal{X}_{(k+1:n)}] }{\mathbb{E}[U]}\mbox{ w.p. } 1.
\end{align}
In a{n} SHN, the set $\mathcal{X}_j$ is the set of $i.i.d.$ exponential random variables with rate $\frac{\lambda_e}{(m-1)}$. From~\eqref{eqn:exp_mean}, we have $\mathbb{E}[\mathcal{X}_{(k:n-1)}]=\frac{(m-1)}{\lambda_e} ( H_{n-1} - H_{n-1-k} )$ and  $\mathbb{E}[\mathcal{X}_{(k+1:n)}]=\frac{(m-1)}{\lambda_e} ( H_{n} - H_{n-1-k} )$. This completes the proof.
\end{proof}
% ( H_{n-1} - H_{n-1-k} )
% .  \frac{(m-1)\lambda_s}{\lambda_e} ( H_{n} - H_{n-1-k} ). 
% RECALL: mention about lower bound 

From Corollary~\ref{cor:w_memory_age}, it is clear that for fixed $k$ and $n$, the version age of $k$-keys on a{n} SHN in total key subscription increases linearly as the total number of nodes $m$ grows. Therefore, if the total number of keys $n$ also grows linearly with the total number of nodes in the network, that is, $n=\lfloor \alpha m \rfloor$ for some rate $\alpha \in (0,1]$, then we obtain the following scalability result.

\begin{corollary}\label{cor:w_memory_scale}
The average version age of $k$-keys over a{n} SHN in total key subscription with a countable memory is given by:
\vspace{-4mm}
\begin{align}
    \lim_{m \to \infty} \ageongraph{k}{n}{n}{m} = \frac{\lambda_s (k+1-\alpha)}{\alpha\lambda_e},
\end{align}
where $n=\lfloor \alpha m \rfloor$ for $\alpha \in (0,1]$.
\end{corollary}
% We have assumed that $n=\lfloor \alpha m \rfloor$. By definition of total key subscription in Table~\ref{tab:types}, we also note that $\alpha \in (0,1]$ 
\begin{proof} For $\alpha=1$, we have total key subscription case $n=m$. The result trivially follows from Theorem~\ref{thm:hetero_w_memory_age}. For $\alpha\in(0,1)$, we have the following from Corollary~\ref{cor:w_memory_age} and~\eqref{eqn:age_on_graph}:
\begin{align*}
    \ageongraph{k}{n}{n}{m} =& \frac{\lfloor \alpha m \rfloor}{m}\left(\sum_{i=\lfloor \alpha m \rfloor-k}^{\lfloor \alpha m \rfloor-1} \frac{\lambda_s(m-1)}{\lambda_e i}\right) \\[-2pt] &+ \frac{m - \lfloor \alpha m \rfloor}{m}\left(\sum_{i=\lfloor \alpha m \rfloor-k}^{\lfloor \alpha m \rfloor} \frac{\lambda_s(m-1)}{\lambda_e i}\right),
\end{align*}
where $n=\lfloor \alpha m \rfloor$. Then, we have the following as $m \to \infty$: 
\begin{align}
    \lim_{m \to \infty} \ageongraph{k}{n}{n}{m} &= \alpha \left(\sum_{i=1}^{k} \frac{\lambda_s}{\lambda_e \alpha }\right) + (1 - \alpha)\left(\sum_{i=1}^{k+1} \frac{\lambda_s}{\lambda_e  \alpha}\right) \\
    &= k \frac{\lambda_s}{\lambda_e } + \frac{(1 - \alpha)}{\alpha}  (k+1) \frac{\lambda_s}{\lambda_e}.\end{align}\end{proof}

% If the node is subscriber, 
% \begin{align*}
%     \agesubs{k}{n}{n}{m} &= \frac{k\lambda_s}{\alpha\lambda_e} \mbox{ w.p. } 1.
% \end{align*}
% If the node is non-subcriber,
% \begin{align*}
%     \agenonsubs{k}{n}{n}{m}= \frac{(k+1)\lambda_s}{\alpha\lambda_e} \mbox{ w.p. } 1.
% \end{align*}
% where $n=\alpha m$ for $\alpha \in [0,1]$. 

% \begin{rem}
% The smallest version age for scalability is achieved when the total number of keys $n$ grows at the same rate as the total number of nodes $m$, that is, $\alpha = 1$.   
% \end{rem}

\paragraph{Partial Key Subscription $(s<n)$}
In this network type, the source sends keys to $s$ subscriber nodes in the set $\mathcal{S}$ and then, randomly chooses $n-s$ nonsubscriber nodes to which it sends a unique key for the version $\ell$ at time $\tau_\ell$, denoted by $\mathcal{M}^\ell$. To be more precise, the source sends a unique key to each node in the set $\mathcal{K}^\ell(=\mathcal{M}^\ell \cup \mathcal{S})$ at time $\tau_\ell$. 

We have a complex system to analyze due to the random selection of nodes receiving a key from the source for a status update, combined with the memory capabilities of the nodes. We provide an example later in this section to illustrate the complexity of the age analysis for partial key subscription. We first provide Theorem~\ref{thm:w_memory_r_subs} that provides an upper and lower bounds for the average $k$-keys age over the graph $\Vec{G}$ in the partial key subscription in memory scheme. 

\begin{theorem}\label{thm:w_memory_r_subs}
In a{n} SHN, the average $k$-keys age over the graph $\Vec{G}$, $\ageongraph{k}{n}{s}{m}$, satisfies the following bounds:
\begin{align}\label{eqn:symbol_low_up}
    \agesubs{k}{m}{m}{m} \leq \ageongraph{k}{n}{s}{m} \leq \agenonsubs{k}{n}{n}{m}. 
\end{align}\end{theorem}

{In other words}, the statement says that the average $k$-keys age over the graph $\Vec{G}$ is lower bounded by the average $k$-keys age for a node with memory in full subscription $(n=m)$, $\agesubs{k}{m}{m}{m}$, and is upper bounded by the average $k$-keys age for a nonsubscriber node with memory in total key subscription $(n=s)$, $\agenonsubs{k}{n}{n}{m}$. We derive closed-form expressions for the upper and lower bounds given in~\eqref{eqn:symbol_low_up} as a corollary to Theorem~\ref{thm:w_memory_r_subs} and Corollary~\ref{cor:w_memory_age}. 
\begin{corollary}\label{cor:thm_bounds}
In a{n} SHN, the average $k$-keys age over the graph $\Vec{G}$ satisfies the following bounds:
\begin{equation}\label{eqn:thm_bounds}
\Tilde{\lambda} \left( \sum_{i=m-k}^{m-1} \frac{1}{i} \right) \leq \ageongraph{k}{n}{s}{m}\nonumber \leq \Tilde{\lambda} \left( \sum_{i=n-k}^{n} \frac{1}{i} \right) \mbox{ w.p. } 1.    
\end{equation}
where $\Tilde{\lambda}=\frac{(m-1)\lambda_s}{\lambda_e}$. 
\end{corollary}
% Before it, we provide auxiliary lemmas to prove Theorem~\ref{thm:w_memory_r_subs}. 

Later in this section, we discuss how strict the bounds given in Theorem~\ref{thm:w_memory_r_subs} are. First, we need the notion of {early stopping} to prove Theorem~\ref{thm:w_memory_r_subs}. We denote the time when node $j$ decodes the status update $\ell$ by $\tau^*_{\ell,j}$. 
\begin{definition}\label{Defn:Early_stop}
For node $j$, if $\tau^*_{\ell_1,j}>\tau^*_{\ell_2,j}$ for $\ell_1 < \ell_2$, then the version $\ell_2$ 
{\em early stops} the version $\ell_1$ at node $j$.    
\end{definition}
%Paraphrasing,
{In other words,} Definition~\ref{Defn:Early_stop} says that if node $j$ decodes a status update with a higher version earlier than {the} one with a lower version, the higher version is said to early stop the lower version at node $j$.

Consider Figs.~\ref{fig:path_subset}(a) and {\ref{fig:path_subset}(b) where we} %Fig.~\ref{fig:path_subset}(b), they 
depict the sample path of the $k$-keys version age process $A^\dagger(k,t)$ for a nonsubscriber node in total key subscription $(s=n)$ and partial key subscription $(s<n)$ cases, respectively. It is assumed that edge activations in the network and the source updates occur at the same time in both given sample paths. For simplicity, we remove the subscript $j$ from the notation in Fig.~\ref{fig:path_subset}.

\begin{figure}[!t]
\centering
\subfloat[]{\includegraphics[width=0.9\linewidth]{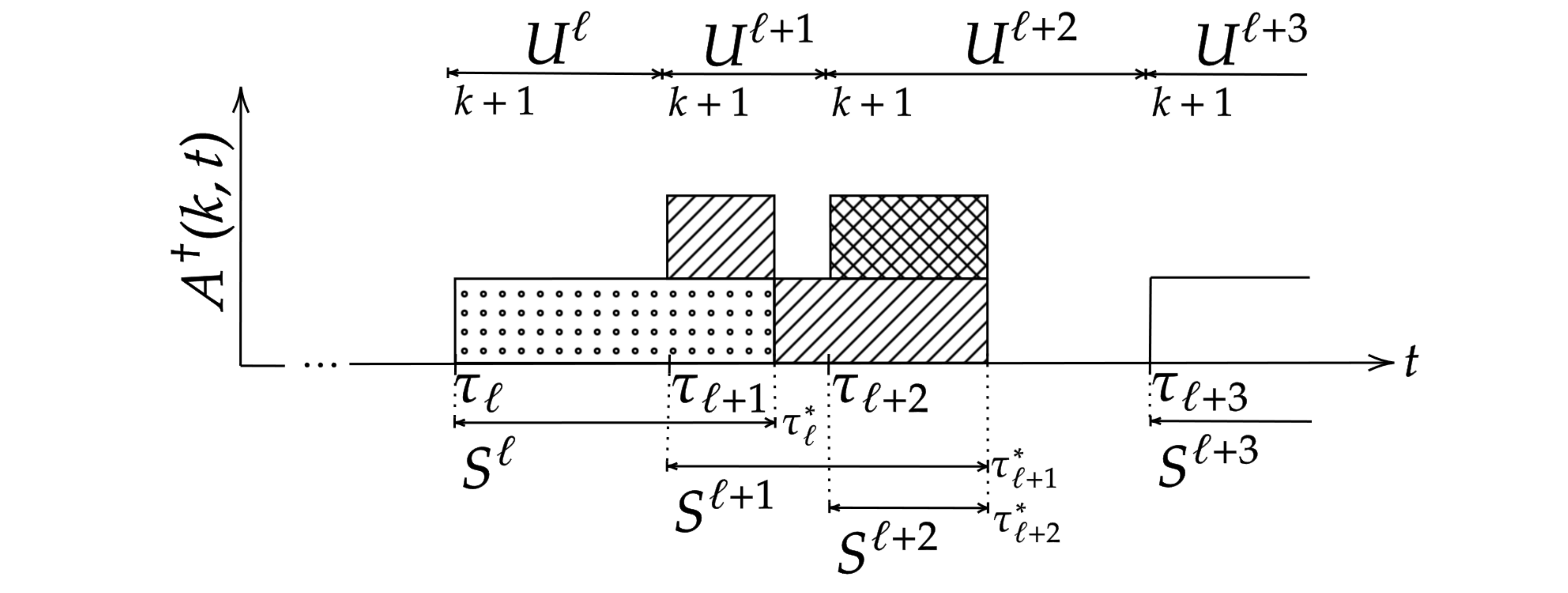}
        \label{fig:a_sub}}
\hfil
\subfloat[]{\includegraphics[width=0.90\linewidth]{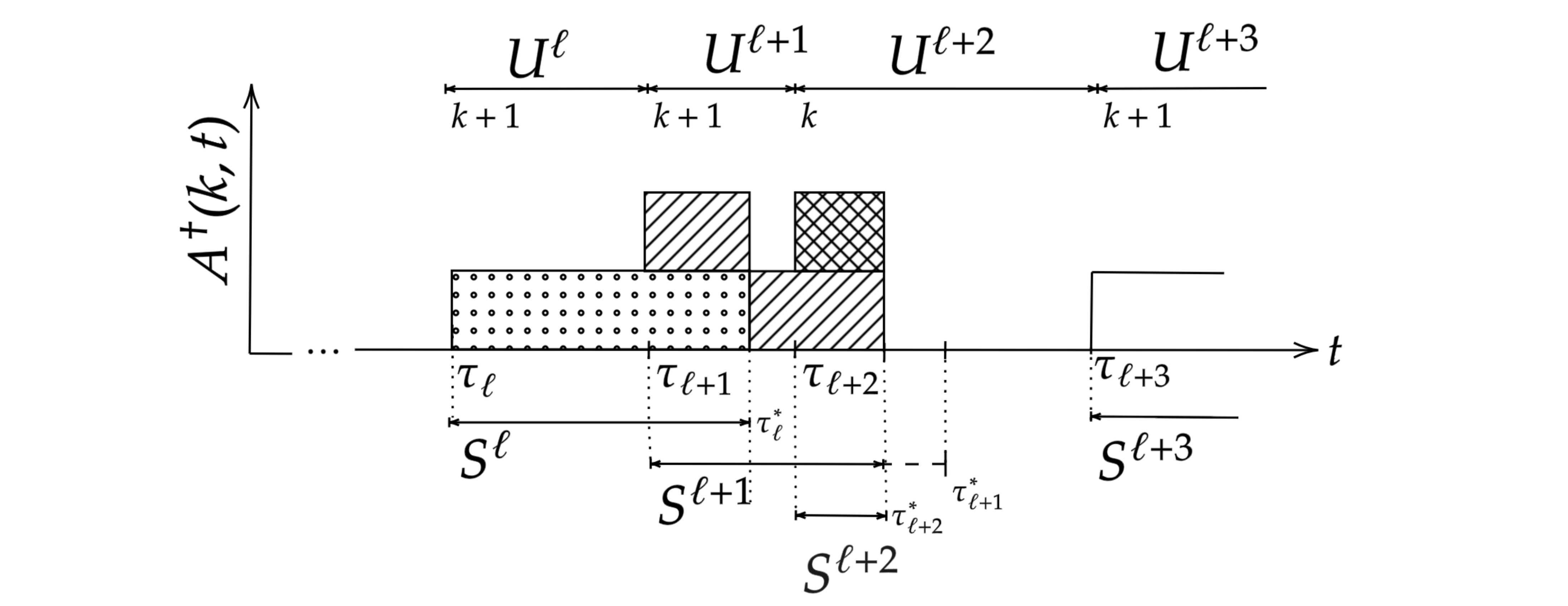}
        \label{fig:b_sub}}
    \vspace{-0.35cm}
\caption{Sample path of the $k$-keys version age $A^\dagger(k,t)$ for a nonsubscriber node with a memory in (a) total key subscription $(s=n)$, (b) partial key subscription $(s<n)$.}
\label{fig:path_subset}
\vspace{-0.5cm}
\end{figure}

On the one hand, one can easily see that {\em no} version can early stop another one at a node with memory in total key subscription. First, we have $\mathcal{S}^\dagger \cap \mathcal{K}^\ell = \emptyset $ for any $\ell$ and  $\mathcal{K}^\ell = \mathcal{S}$ by construction in  total key subscription. Furthermore, the existence of memory implies once node $j$ communicates with a subscriber node in the set $\mathcal{S}(=\mathcal{K}^{\ell_1}=\mathcal{K}^{\ell_2})$, it receives {\em all} keys for the previous status updates as well. These two show that node $j$ cannot collect $k+1$ keys for $\ell_2$ earlier than it collects $k+1$ for $\ell_1$.

For example, in the given sample path of the $k$-keys version age process in Fig.~\ref{fig:path_subset}(a), the node decodes the status update with the version stamp $\ell+2$ at $\tau^*_{\ell+2}$. It means that there are $k+1$ distinct edge activations since $\tau_{\ell+2}$ with the nodes that received keys from the source, $\mathcal{K}^{\ell+2}$. We also know that these messages also include the keys for the version $\ell_1$ due to memory scheme. Therefore, the node decodes the status update with the version $\ell+2$ and $\ell+1$ at the same time, that is, $\tau^*_{\ell+1}=\tau^*_{\ell+2}$. Clearly, time  $\tau^*_{\ell+2}$ cannot be less than $\tau^*_{\ell+1}$.

% the service time $\mathcal{S}^\ell$ for any status update to a node in $\mathcal{S}^\dagger$ is equal to $\mathcal{X}_{(k+1:n)}$ as shown in the proof of Corollary~\ref{cor:w_memory_age} above.

On the other hand, a version can {\em early stop} another one at a node in partial key subscription. In this case, the sets $\mathcal{K}^{\ell+2}(\supset \mathcal{S})$ and $\mathcal{K}^{\ell+1}(\supset \mathcal{S})$ do not have to be equal to each other. Therefore, having $k+1$ interaction with the nodes in $\mathcal{K}^{\ell+2}$ does not necessarily imply having $k$ (or $k+1$) distinct interactions with the nodes in $\mathcal{K}^{\ell+1}$.

For example, in the given sample path of the $k$-keys version age process in Fig.~\ref{fig:path_subset}(b), the node needs $k$ more keys to decode the status update $\ell+2$ since it is in the set $\mathcal{K}^{\ell+2}$. However, it needs $k+1$ keys to decode the status update $\ell+1$ since it is not in the set $\mathcal{K}^{\ell+1}$. Thus, we can state the event of version $\ell+2$ early stops version $\ell+1$ for the given case: 
$$
\{  \mathcal{X}_{(k:n-1)}(\ell+2)  \leq  \mathcal{X}_{(k+1:n)}(\ell+1) \}.
$$
% \mbnote{Is this the only possibility? What if the node is not in the sets of $\mathcal{K}^{\ell+2}$ and also $\mathcal{K}^{\ell+1}$? Since $\mathcal{K}^{\ell+2} \neq \mathcal{K}^{\ell+1}$, it is still possible to have  $\mathcal{X}_{(k+1:n)}(\ell+2)  \leq  \mathcal{X}_{(k+1:n)}(\ell+1)$ due to the realizations of the different activations. If it is true, should we mention that as well as a footnote maybe? }
In the given sample paths Fig.~\ref{fig:path_subset}, the node decodes the status update $\ell\!+\!2$ earlier than the status update $\ell+1$, i.e., $\tau^*_{\ell+1}\!>\!\tau^*_{\ell+2}$. It means that the node collects $k+1$ distinct keys required to decode the version $\ell\!+\!2$ earlier than $k+1$ distinct keys required to decode the version $\ell+1$. As a result, at time $\tau^*_{\ell+2}$, the status update with the version $\ell+2$ is decoded and version age reduces down to $0$ as shown in Fig.~\ref{fig:path_subset}(b). Therefore, the service time $\mathcal{S}^{\ell+1}(=\tau^*_{\ell+2}-\tau_\ell)$ is less than $\mathcal{X}_{(k+1:n)} =\tau^*_{\ell+1}-\tau_{\ell+1}$. 

We note that the given example is not the only event where early stopping occurs. For instance, consider a node that belongs to neither $\mathcal{K}^{\ell+2} $ nor $ \mathcal{K}^{\ell+1}$. Since $\mathcal{K}^{\ell+2} \neq \mathcal{K}^{\ell+1}$, it is still possible to have early stopping due to the realizations of different activations, such as the event $\mathcal{X}_{(k+1:n)}(\ell+2) \leq \mathcal{X}_{(k+1:n)}(\ell+1)$. Other scenarios where early stopping occurs can also be identified.

By construction, we have $\mathcal{X}_j(\ell+1) \cap \mathcal{X}_j(\ell+2) \neq \emptyset$, that is, the provided ordered sets includes the same edge activations, making them dependent. Additionally, the cycle length $U^\ell$ is a random variable and the source randomly selects key receiver node sets $\mathcal{K}^\ell$ for each status update. This shows that it is challenging to track which version is early stopped by another version, making it difficult to determine a closed-form age expression for the partial key subscription case. Therefore, we provide lower and upper bounds in the given theorem.

Now, we can state the first Lemma to prove Theorem~\ref{thm:w_memory_r_subs}:

\begin{lemma}\label{lem:no_earylstop}
{Assuming that} node $j$ is {a} nonsubscriber node with a countable memory in {the} {\em partial key} subscription case in a{n} SHN, the service time of version $\ell$ to node $j$ is less than or equal to the $(k+1)$th order statistic of the set of exponential random variables, $\mathcal{X}_{(k+1:n)}$. 
% Let $\ell_2$ and $\ell_1$ be two version stamp such that $\ell_2 > \ell_1$. 
% Assume node $j$ is nonsubsriber node with a countable memory in total key subscription case for a SHN, then there exists {\em no} version that early stops another version at node $j$ and ..  the service time of a key 
\end{lemma}
\begin{proof}
For simplicity, we remove the subscript $j$ from the notation. 
% The node set $\mathcal{M}^\ell$ constructed independently and uniformly selection from $\mathcal{S}^\dagger$ for each $\ell$, that is,
First, we assume that there is no version that early stops version $\ell$ at node $j$. Node $j$ receives a unique key from the source with probability $\frac{n-s}{m-s}$. If $j \in \mathcal{M}^\ell$, then, the service time of $\ell$ to node $j$ is equal to $\mathcal{X}_{(k:n-1)}$, otherwise it is $\mathcal{X}_{(k+1:n)}$ from Lemma~\ref{lem:service}. One can easily see that $\mathcal{X}_{(k:n-1)}$ is less than $\mathcal{X}_{(k+1:n)}$. This completes the first part of the proof.

We assume that there exists a version $\tilde{\ell}$ such that it early stops version $\ell$ at node $j$, that is, we have $\tau_{\tilde{\ell}} < \tau_\ell$ and $\tilde{\ell} > \ell$.  If $j \in \mathcal{M}^{\tilde{\ell}}$, node $j$ has at least $k$ distinct interactions with the nodes in $\mathcal{K}^{\tilde{\ell}}$ after $\tau_{\tilde{\ell}}$ so that it decode{s} the status update $\tilde{\ell}$. From Lemma~\ref{lem:service}, we know that the service time of version $\tilde{\ell}$ is $\mathcal{X}_{(k:n-1)}$. In addition to that, {\em early stopping} assumption {implies} that node $j$ has fewer than $k+1$ distinct interactions with the nodes in $\mathcal{K}^{{\ell}}$ after $\tau_{{\ell}}$. Then, we have that the service time of version ${\ell}$ is equal to $\tau^*_{\tilde{\ell}}-\tau_\ell$ instead of $\tau^*_{\ell}-\tau_\ell$ because of early stopping. This implies that the service time of version ${\ell}$, $\mathcal{S}^\ell$, is less than $\mathcal{X}_{(k+1:n)}(=\tau^*_{\ell}-\tau_\ell)$. We can conclude that the service time of version ${\ell}$ is upper bounded by the service  $\mathcal{X}_{(k+1:n)}$. If $j \notin \mathcal{M}^{\tilde{\ell}}$, node $j$ has at least $k+1$ distinct interactions with the nodes in $\mathcal{K}^{\tilde{\ell}}$ after $\tau_{\tilde{\ell}}$ so that it decodes the status update $\tilde{\ell}$. This observation leads us to the conclusion that the remaining part of the proof follows a similar pattern to the one discussed earlier. This completes the proof.\end{proof}

\begin{lemma}\label{lem:upper_bound}
In a{n} SHN, the version age of $k$-keys for a nonsubscriber node in partial key subscription, $\agenonsubs{k}{n}{s}{m}$, is less than or equal to that in total key subscription, $\agenonsubs{k}{n}{n}{m}$. To be more precise, we have the following for $s<n$:
\begin{align}
    \agenonsubs{k}{n}{s}{m} \leq \agenonsubs{k}{n}{n}{m}.
\end{align}
\end{lemma}
\begin{proof} 
% Let $T:=\{\tau_\ell\}_{\ell=0}^{\infty}$ be the monotonically increasing sequence of times when status updates occur at the source node $0$, with $\tau_0\!:=\!0$. 
Let $T^1\!=\!\{\tau_{\ell_a}\}_{a=0}^{\infty}$ be a subsequence of $T$ such that $A(k,\tau_{\ell_a})\!=\!1$. Let $L_a$ be the time elapsed between two consecutive arrivals of the subsequence $T^1$. Let $R_a$ be the version age of $k$-keys, $A(k,t)$ integrated over the duration $[\tau_{\ell_a},\tau_{\ell_{a+1}})$ in a node. 
It is worth noting that two random variables $ S^i_j$ and $ S^{i+1}_j$ (service time of a status update to a node) are not independent (consider early stopping) in memory scheme if $\ell_a\!\leq\!i\!\leq\!\ell_{a+1}-1$ for any $a \in \mathbb{N}$, but they are identically distributed. By construction of the sequence $T^1$, a pair of $((L_a,R_a),(L_b,R_b))$ for any $a \neq b$ is $i.i.d.$.  Then, from~\cite[Thm. 6]{gallager1997discrete},  we have:
\begin{align}\label{eqn:}
   \agenonsubs{k}{n}{s}{m} = \frac{\mathbb{E}[R_a]}{\mathbb{E}[L_a]} \!= \frac{\mathbb{E}[\sum_{i=\ell_a}^{\ell_{a+1}-1} S^i_j  ]}{\mathbb{E}[ \sum_{i=\ell_a}^{\ell_{a+1}-1} U^i]}  \mbox{ w.p. }  1.
\end{align}
Consider the upper bound that we constructed for the service time for a nonsubsciber node in partial key subscription in Lemma~\ref{lem:upper_bound}; then, we have: 
\begin{align*}
 % \frac{\mathbb{E}[ \sum_{i=\ell_a}^{\ell_{a+1}-1} S^i_j]}{ \mathbb{E}[\sum_{i=\ell_a}^{\ell_{a+1}-1} U^i]} &\leq \frac{\mathbb{E}[\sum_{i=\ell_a}^{\ell_{a+1}-1} \mathcal{X}^i_{(k+1:n)}  ]}{ \mathbb{E}[\sum_{i=\ell_a}^{\ell_{a+1}-1} U^i]} \\ &= \! \frac{\mathbb{E}[{\ell_{a+1}\!-\!\ell_a}]  \mathbb{E}[\mathcal{X}_{(k+1:n)}] }{ \mathbb{E}[{\ell_{a+1}-\ell_a}] \mathbb{E}[U]} \\ &= \frac{\mathbb{E}[\mathcal{X}_{(k+1:n)}]}{{E}\mathbb[U]} =\agenonsubs{k}{n}{n}{m}, \quad\mbox{ w.p. } 1.
\agenonsubs{k}{n}{s}{m}&= \frac{\mathbb{E}[R_a]}{\mathbb{E}[L_a]} = \frac{\mathbb{E}[ \sum_{i=\ell_a}^{\ell_{a+1}-1} S^i_j]}{ \mathbb{E}[\sum_{i=\ell_a}^{\ell_{a+1}-1} U^i]}\\ &\leq \frac{\mathbb{E}[\sum_{i=\ell_a}^{\ell_{a+1}-1} \mathcal{X}^i_{(k+1:n)}  ]}{ \mathbb{E}[\sum_{i=\ell_a}^{\ell_{a+1}-1} U^i]} \!=\! \frac{\mathbb{E}[{\ell_{a+1}\!-\!\ell_a}]  \mathbb{E}[\mathcal{X}_{(k+1:n)}] }{ \mathbb{E}[{\ell_{a+1}-\ell_a}] \mathbb{E}[U]} \\ &= \frac{\mathbb{E}[\mathcal{X}_{(k+1:n)}]}{{E}\mathbb[U]} =\agenonsubs{k}{n}{n}{m}, \quad\mbox{ w.p. } 1.
\end{align*}\end{proof}

% Thus, its shows that the version age of $k$-keys for a nonsubscriber node in a partial subscription $(s<n)$ is less than or equal to that in a total key subscription $(s=n)$.
Now, we can give the proof of Theorem~\ref{thm:w_memory_r_subs}. 
% , then, we construct the upper bound. We consider a SHN on $\!m\!+\!1\!$ nodes with $\!s\!$ subscriber to implement $\!(k,\!n)\!$-TSS. 
\begin{proof}[Proof of Theorem~\ref{thm:w_memory_r_subs}] \textit{Upper Bound.} First, we construct the upper bound. From the definition in~\eqref{eqn:age_on_graph}, the average $k$-keys age over the graph $\vec{G}$, $\ageongraph{k}{n}{s}{m}$ is a convex combination of $\agesubs{k}{n}{s}{m}$ and $\agenonsubs{k}{n}{s}{m}$. Subscriber nodes always need $k$ additional keys to decode the status update while nonsubscriber nodes need $k$ additional keys with probability $\frac{n-s}{m-s}$ or $k\!+\!1$ additional keys with probability $\frac{m-n}{m-s}$ to decode the status update. Then, one can easily see that we have $\agesubs{k}{n}{s}{m}\leq\agenonsubs{k}{n}{s}{m}$ and it follows:
\begin{align}\label{eqn:upper_convex}
  \agesubs{k}{n}{s}{m}  \leq \ageongraph{k}{n}{s}{m} \leq \agenonsubs{k}{n}{s}{m}. 
\end{align}
From Lemma~\ref{lem:upper_bound} and~\eqref{eqn:upper_convex}, we have the following upper bound:
% for $\ageongraph{k}{n}{s}{m}$:
\begin{align}\label{eqn:upper_n_subs_s_subs}
  \ageongraph{k}{n}{s}{m} \leq \agenonsubs{k}{n}{s}{m}  \leq \agenonsubs{k}{n}{n}{m}. 
\end{align}
\textit{Lower Bound.} We consider that $\mathcal{S}=V^*$, that is, any node in the network gets a unique key (equivalently, full subscription $n=m$). Furthermore, any node on the network always needs $k$ more keys with the version stamp $\ell$ after $\tau_\ell$ in this case. Then, we have the following lower bound for $\ageongraph{k}{n}{s}{m}$:
% By definition, it holds that $\ageongraph{k}{m}{m}{m}=\agesubs{k}{m}{m}{m}$.
\begin{align}\label{eqn:lower_full_subs}
   \agesubs{k}{m}{m}{m} = \ageongraph{k}{m}{m}{m} \leq  \ageongraph{k}{n}{s}{m}, 
\end{align}
where the inequality in (\ref{eqn:lower_full_subs}) directly follows from Cor.~\ref{cor:w_memory_age}. When we combine~\eqref{eqn:upper_n_subs_s_subs} and~\eqref{eqn:lower_full_subs}, we have obtained~\eqref{eqn:symbol_low_up}, which completes the proof.\end{proof}

% \begin{rem}
% We note that the bounds provided in Corollary~\ref{cor:w_memory_scale} can be tight depending on the value of $\alpha$. The smallest scalability result is obtained when the total number of keys $n$ grows at the same rate as the total number of nodes $m$, that is, for $\alpha=1$.       
% \end{rem}

Next, we analyze the tightness of the bounds presented in Theorem~\ref{thm:w_memory_r_subs}
{as follows}. Let $U_B$ and $L_B$ denote the upper and lower bounds, respectively, as given in Corollary~\ref{eqn:thm_bounds}. For finite values of $k,$ $m$, $n$ in a{n} SHN with a countable memory, one can easily see that we have the following:
% as a result of basic algebraic manipulation on Cor.~\ref{eqn:thm_bounds}:
% \begin{corollary}
% For finite $k,m,n$ in a SHN with a countable memory, we have:
\begin{align*}
     k\frac{\lambda_s}{\lambda_e} \leq  L_B  \leq  \ageongraph{k}{n}{s}{m} \leq U_B \leq \frac{(m-1)(k+1)\lambda_s}{(n-k)\lambda_e}.
\end{align*} 
% \end{corollary}

We use {\em the normalized relative gap between the upper and lower bounds}, $\frac{U_B-L_B}{L_B}$,
to measure the tightness of provided bounds:
\begin{align}\label{eqn:relative_gap}
      \frac{U_B-L_B}{L_B} \leq  \frac{(m-n)k + m - 1 -k + k^2}{(n-k) k }.
\end{align}

From Corollary~\ref{cor:thm_bounds} and~\eqref{eqn:relative_gap}, it is clear that for fixed $k$ and $n$, the version age of $k$-keys over a{n} SHN and the normalized relative gap between the upper and lower bounds increases linearly as the total number of nodes $m$ grows. Therefore, we consider that the total number of keys $n$ grows linearly with the total number of nodes on the network, that is, $\!n\!=\!\!\lfloor\! \alpha m \!\rfloor$ for some rate $\alpha\in (0,1]$, to obtain the following scalability result.
\begin{corollary}
The average version age of $k$-keys over an SHN in partial key subscription under memory scheme obeys:
\begin{align}
    k\frac{\lambda_s}{\lambda_e} \leq \lim_{m \to \infty} \ageongraph{k}{n}{s}{m} \leq \frac{(k+1)\lambda_s}{\alpha\lambda_e},
\end{align}
where $n=\lfloor \alpha m \rfloor$ for $\alpha \in (0,1]$. 
\end{corollary}
Similarly, we have the following scalability results for the normalized relative gap between the upper and lower bounds:
% For scalability results, we consider that $n=\lfloor \alpha m \rfloor$ as $m$ countable grows:
\begin{align*}
      \frac{U_B-L_B}{L_B} \leq  \frac{k-\alpha k + 1}{ \alpha k} ,\mbox{ when } m \to \infty.
\end{align*}
% \st{Paraphrasing the statement, it says}
{To paraphrase, the statement indicates} that we have tighter bounds as the number of keys required to decode the status update, $k$, increases and the growth rate of the number of total keys, $\alpha$ with respect to the network size increases for a countably large SHN.

\vspace{-5mm}
\section{Age Analysis for Nodes without Memory}\label{sec:wo_memory}

In this section, we consider nodes in the network without memory. Similar to the previous section, we first analyze a $(k,n)$-TSS feasible network in  full subscription with nonhomogenoues edge activation rates. Then, we analyze a scalable homogeneous network  in total $(s=n)$ and partial $(s<n)$ key subscriptions with the total number of nodes $m$.

\subsection{Arbitrary $(k,n)$-TSS Feasible Network in Full Subscription}
In this subsection, we provide a closed-form expression for the version age of $k$-keys for any node without memory, $\bar{\Delta}_j(k)$, in a $(k,n)$-TSS feasible network. 

\begin{theorem}\label{thm:hetero_wo_memory_age}
Let $\vec{G}$ be a $(k,n)$-TSS feasible network. The version age of $k$-keys for node $j$ {\em without memory} is:
\begin{align}\label{eqn:hetero_wo_memory_age}
    \bar{\Delta}_j(k) = \frac{ \mathbb{E}[\min(\mathcal{X}_{(k:{n}_j)},U)]}{Pr(\mathcal{X}_{(k:{n_j})}\!\leq\!U) \mathbb{E}[U]} \mbox{ w.p. } 1, 
\end{align}
where $U$ is the interarrival times for the source update.
\end{theorem}
% One can easily compute $\Bar{\Delta}_j(k)$ and $Pr(\mathcal{X}_{(k:{n_j})} \leq U)$ by using~\eqref{eqn:non-iid_cdf}.
% Instead {of} providing a proof here, we can first focus on the statement of 
See \cite[Thm.2]{bayram2024age} for the proof of Theorem~\ref{thm:hetero_wo_memory_age}. We focus here on the statement of Theorem~\ref{thm:hetero_wo_memory_age}. In the memoryless scheme, a node needs to get at least $k\!+1\!$ different keys with version $\ell$ before $\tau_{\ell+1}$ so that it can decrypt the status update generated at $\tau_\ell$. To be more precise, a node can decode the status update at $\tau_\ell$ if the event $E=\{\mathcal{X}_{(k:n_j)}\!\leq\!U^\ell\}$ happens, otherwise, it misses. Then, Theorem~\ref{thm:hetero_wo_memory_age} says that if the event $E$ happens in an update cycle, the accumulated age in the update cycle is $\mathcal{X}_{(k:n_j)}$, otherwise, it is $U^\ell$. The expected time of the generation of two consecutive updates that can be decrypted at node $j$ is equal to $Pr(\mathcal{X}_{(k:n_j)} \leq U)\mathbb{E}[U]$ from \cite[Lem. 2]{bayram2024age}. 
\subsection{Scalable Homogeneous Network with $s$ Subscriber}
In this subsection, we consider a scalable homogeneous network with $m+1$ nodes without memory. Unlike the previous section, we have provided closed-form expressions for both total key and partial key subscription cases. Therefore, we analyze both scenarios together.
\begin{theorem}\label{thm:wo_memory_subset}
For an SHN, the version age of $(k,n)$-TSS for an individual node without memory is the following. If a node is a nonsubscriber node, then; 
\begin{equation}\label{eqn:wo_memory_nonsubs}
    \agenonsubswo{k}{n}{s}{m}\!=\!\frac{ \alpha_0  \mathbb{E}[\min(\mathcal{X}_{(k:n\!-\!1)},U)] \!+\!  \alpha_1 \mathbb{E}[\min(\mathcal{X}_{(k\!+\!1:n)},U)] }{  \alpha_0 Pr( E_0 )\mathbb{E}[U] + \alpha_1  Pr( E_1)\mathbb{E}[U]  },
\end{equation}
with probability $1$ where $\alpha_0 = \frac{n-s}{m-s} $ and $\alpha_1 = 1 - \alpha_0$  and the events $E_0=\{\mathcal{X}_{(k:n-1)}\leq U \}$ and $E_1=\{\mathcal{X}_{(k+1:n)} \leq U\}$ for $s\leq n$. If a node is {a} subscriber node, then;
\begin{align}\label{eqn:wo_memory_subset_subs}\\[-1.5em]
    \agesubswo{k}{n}{s}{m}=  \frac{ \mathbb{E}[\min(\mathcal{X}_{(k:n-1)},U)]  }{   Pr( E_0 )\mathbb{E}[U]    }  \mbox{\quad w.p. } 1. \\[-2em]
\end{align}
\end{theorem}
We have an closed-form expressions for both $\agenonsubswo{k}{n}{s}{m}$  and $\agesubswo{k}{n}{s}{m}$, which will be provided following stating the proof of Theorem~\ref{thm:wo_memory_subset}. {For that} we need the following lemma.
\begin{lemma}\label{lem:markov}
Let $\bar{A}^k_j[\ell]=\Bar{A}_j(k,\tau_\ell)$ be the version age of information at $\tau_\ell$ for the node $j$. The sequence $\bar{A}^k_j[\ell]$ is homogeneous success-run with the rate $\mu_j\!=\!Pr(\mathcal{X}^j_{(k:n_j)}\!\leq\!U)$.
\end{lemma}

See \cite[Lem. 2]{bayram2024age} for the proof of Lemma~\ref{lem:markov}. 
% As a result of Lemma~\ref{lem:markov}, the random variable $\bar{A}_j[\ell]$ has truncated geometric distribution at $\ell+1$ with success rate $\mu_j(=Pr( \mathcal{X}_{(k:n_j)}\leq U) )$. As $\ell$ goes to infinity, $\mathbb{E}[\bar{A}_j[\ell]]$ goes to $1/\mu_j$ for node $j$. 
Now, we can prove Theorem~\ref{thm:wo_memory_subset}. 

\begin{proof}[Proof of Theorem~\ref{thm:wo_memory_subset}] We first prove the theorem for a nonsubscriber node~\eqref{eqn:wo_memory_nonsubs}. Consider a nonsubscriber node. Let $\bar{A}^\dagger[\ell]$ be the version age of information at $\tau_\ell$ for a nonsubscriber node. Let $E_{\mathcal{K}^\ell}$ be the event that the node is in the set $\mathcal{K}^\ell$. On the one hand, once the event $E_{K^\ell}$ happens at $\tau_\ell$, the node gets a unique key from the source, and thus it needs to get $k$ more keys with the version stamp $\ell$ before $\tau_{\ell+1}$ to decode the status update $\ell$. We know that $\mathcal{K}^\ell$ is fixed in $[\tau_\ell,\tau_{\ell+1})$. Thus, equivalently, we can assume as if we had a network such that the set $\mathcal{N}^+$ for the node is equal to the set $\mathcal{K}^\ell$ during $[\tau_\ell,\tau_{\ell+1})$. It is worth noting that the nodes in the set $V^* \setminus \mathcal{K}^\ell$ do not have the key with version-stamp $\ell$. Therefore, the nodes in the set $V^* \setminus \mathcal{K}^\ell$ do not spread any messeage in $[\tau_\ell,\tau_{\ell+1})$. Then, if the event $E_{\mathcal{K}^\ell}$ happens, we have $n_j=n-1$, otherwise, $n_j=n$.

Let $T^1\!=\!\{\tau_{\ell_a}\}_{a=0}^{\infty}$ be a subsequence of $T$ such that $\bar{A}_j(k,\tau_{\ell_a})\!=\!1$. Let $R_a$ be version age of $k$-keys, $\bar{A}^\dagger(k,t)$ integrated over the duration $[\tau_{\ell_a},\tau_{\ell_{a+1}})$ in a node. Then, 
\begin{align*}%\label{eqn:R_a_integ}
        R_a =& \int^{\tau_{\ell_{a+1}-1}}_{\tau_{\ell_a}} \!\bar{A}^\dagger(k,t) d t \\[-2pt]
        =& \sum^{\ell_{a+1}-1}_{i=\ell_{a}} \bar{A}^\dagger[i]   \mathds{1}_{E_{\mathcal{K}^i}} ( \min( \mathcal{X}^{i}_{(k:n-1)},U^{i} ) ) \\[-2pt] &+  \sum^{\ell_{a+1}-1}_{i=\ell_{a}} \bar{A}^\dagger[i]  \mathds{1}_{E^c_{\mathcal{K}^i}} ( \min( \mathcal{X}^{i}_{(k+1:n)},U^{i} ) ). 
\end{align*}
The pair of random variables $\mathcal{X}^{i}_{(k+1:n)}$ and $\mathcal{X}^{i+1}_{(k+1:n)}$ are identically distributed. Thus, we have $\mathbb{E}[\mathcal{X}_{(k+1:n)}]=\mathbb{E}[\mathcal{X}^i_{(k+1:n)}]$ for any $i$. Clearly, the random process $\bar{A}^\dagger[i]$ and the event $E_{\mathcal{K}^i}$ are independent. Then, we have:
% \begin{align}\label{eqn:y_l_computation_3}
%         {\mathbb{E}[R_a]} &=\mathbb{E}\left[ \sum^{\ell_{a+1}-1}_{i=\ell_{a}} \bar{A}^\dagger[i] \left(  \mathds{1}_{E_{\mathcal{K}^i}} \min( \mathcal{X}^{i}_{(k:n-1)}U^{i} ) +  \mathds{1}_{E^c_{\mathcal{K}^i}} \min( \mathcal{X}^{i}_{(k+1:n)},U^{i} ) \right)\! \right] 
% \end{align}
%  A pair of random variables $\mathcal{X}^{i}_{(k+1:n)}$ and $\mathcal{X}^{i+1}_{(k+1:n)}$ are identically distributed. Thus, we have $\mathbb{E}[\mathcal{X}_{(k+1:n)}]=\mathbb{E}[\mathcal{X}^i_{(k+1:n)}]$ for any $i$. Similarly, we have  $\mathbb{E}[\mathcal{X}_{(k:n-1)}]=\mathbb{E}[\mathcal{X}^i_{(k:n-1)}]$. Then, we have:
\begin{align*}
        {\mathbb{E}[R_a]} =& \alpha_0 \mathbb{E}[\min( \mathcal{X}_{(k:n-1)},U )] \mathbb{E}\left[ \sum^{\ell_{a+1}-1}_{i=\ell_{a}} \bar{A}^\dagger[i]  \right] \\[-2pt] &+ \alpha_1 \mathbb{E}[ \min( \mathcal{X}_{(k+1:n)},U ) ]  \mathbb{E}\left[ \sum^{\ell_{a+1}-1}_{i=\ell_{a}} \bar{A}^\dagger[i]  \right], 
\end{align*}
where $\alpha_0=Pr(E_{\mathcal{K}^i})=\frac{n-s}{m-s}$ for any $i$ and $\alpha_1 =  1- \alpha_0$. Let $M_a:=\bar{A}^\dagger[\ell_{a+1}-1]$ be the number of status updates that the node has missed in $[\tau_{\ell_a},\tau_{\ell_{a+1}})$, that is, $M_a:=\ell_{a+1}-\ell_{a}$. Then, we have:
\begin{align}\label{eqn:closed_a_Ra}
        {\mathbb{E}[R_a]} =& \alpha_0 \mathbb{E}[\min( \mathcal{X}_{(k:n-1)},U )] \mathbb{E}\left[ \frac{M_a(M_a+1)}{2} \right]  \\[-2pt] &+ \alpha_1 \mathbb{E}[ \min( \mathcal{X}_{(k+1:n)},U ) ] \mathbb{E}\left[ \frac{M_a(M_a+1)}{2} \right]. 
\end{align}
As a corollary of Lemma~\ref{lem:markov}, the sequence $\bar{A}^\dagger[i]$ is homogeneous success-run with the rate $Pr(\{\mathcal{X}_{(k:n-1)}\leq U \})$ if the event $E_{\mathcal{K}^i}$ happens at $\tau_i$. On the other hand, if the event $E_{\mathcal{K}^i}$ does not happen at $\tau_i$, similarly, as a corollary of Lemma~\ref{lem:markov}, the sequence $\bar{A}^\dagger[i]$ is homogeneous success-run with the rate $Pr(\{\mathcal{X}_{(k+1:n)}\leq U \})$. Thus, by combining both of these cases, the sequence $\bar{A}^\dagger[\ell]$ is homogeneous success-run chain with the following rate: 
$$
\mu = \alpha_0 Pr(\{\mathcal{X}_{(k:n-1)}\leq U \}) + \alpha_1 Pr(\{\mathcal{X}_{(k+1:n)}\leq U \}).
$$
This implies that $M_a$ has a geometric distribution for sufficiently large $a \in \mathbb{N}$ with rate $\mu$.
% Then, we have the following:
% \begin{align}\label{eqn:closed_M_a}
%          \mathbb{E}\left[ \frac{M_a(M_a+1)}{2} \right] = \frac{2-\mu}{2\mu^2}\!+\!\frac{1}{2\mu} = \frac{1}{\mu^2}.
% \end{align}
% When we plug~\eqref{eqn:closed_M_a} into $\mathbb{E}[M_a]$, we have the following:
% \begin{align}\label{eqn:closed_R_a}
%        {\mathbb{E}[R_a]} &= \frac{  \alpha_0 \mathbb{E}[\min( \mathcal{X}_{(k:n-1)},U )] + \alpha_1 \mathbb{E}[ \min( \mathcal{X}_{(k+1:n)},U ) ]}{\mu^2}.
% \end{align}
Let $L_a$ be the elapsed time between two consecutive successful arrivals of the subsequence $T^1$. Then, we have:
\begin{equation}\label{eqn:closed_L_a}
    \mathbb{E}[L_a] = \mathbb{E}\left[ \sum^{\ell_{a+1}-1}_{i=\ell_{a}} U^i \right] = \mathbb{E}[U]\mathbb{E}[ M_a ] = \frac{\mathbb{E}[U]}{\mu}.
\end{equation}
When we consider $M_a \sim \mbox{geo}(\mu)$ for~\eqref{eqn:closed_a_Ra} and~\eqref{eqn:closed_L_a}, we have:
 \begin{align}\label{eqn:R_a_over_L_a}
        \frac{\mathbb{E}[R_a]}{\mathbb{E}[L_a]} &= \frac{ \alpha_0 \mathbb{E}[\min( \mathcal{X}_{(k:n-1)},U )] + \alpha_1 \mathbb{E}[ \min( \mathcal{X}_{(k+1:n)},U ) ] }{ \mu\mathbb{E}[U]}.
\end{align}
By {the} definition of the subsequence $T^1$, a pair of $(\!L_a,\!R_a\!)$ and $(L_b,R_b)$ for any $a \neq b$ is $i.i.d.$. From~\cite[Thm. 6]{gallager1997discrete}, we find the time average $\agenonsubswo{k}{n}{s}{m}$, thus {completing} the proof for~\eqref{eqn:wo_memory_nonsubs}.

Now, if the node is a subscriber node, then one can easily see that we have $Pr(E_{\mathcal{K}^i})=1$ for all $i$ (equivalently, $\alpha_0=1$) and $\mu=Pr(\{\mathcal{X}_{(k:n-1)}\leq U \})$. When we plug these into~\eqref{eqn:R_a_over_L_a}, we have~\eqref{eqn:wo_memory_subset_subs}, which completes the proof.\end{proof}

We need the following proposition to provide a closed-form expressions for $\agenonsubswo{k}{n}{s}{m}$ in~\eqref{eqn:wo_memory_nonsubs} and $\agesubswo{k}{n}{s}{m}$~\eqref{eqn:wo_memory_subset_subs}.

\begin{prop}\label{prop:closed_from_for_order}
Consider a set of $i.i.d.$ exponentially distributed random variables $\mathcal{X}$ with cardinality $\Tilde{n}-1$ and mean $\frac{m-1}{\lambda_e}$ and exponentially distributed random variable $U$ with mean $\frac{1}{\lambda_s}$. Then, we have the following closed-form expressions for $\tilde{k}<\tilde{n}$:
\begin{align}\label{eqn:prop_min} \\[-2em] 
  \mathbb{E}[\min(\mathcal{X}_{\tilde{k}:\Tilde{n}-1},U)]=  \sum\limits_{j=1}^{\tilde{k}} \frac{\mathcal{A}(\tilde{n},m,j-1)}{\lambda_e \mathcal{B}(\tilde{n},m,j) + \lambda_s},  
\end{align}
where $\!Pr(\mathcal{X}_{\tilde{k}:\Tilde{n}\!-\!1} \!\leq \!U)\!\!=\! {   \mathcal{A}(\tilde{n},m,\Tilde{k}) }$ with $\!\mathcal{A}(n,m,j)\!\!:=\!\!\! \prod\limits_{i=1}^{j}\!\!\!\frac{\lambda_e \mathcal{B}(n,m,i)}{\lambda_e \mathcal{B}(n,m,i)\! + \!\lambda_s }$ for $j\!\!\geq\!\! 1$, $\mathcal{A}(n,m,0)\!\!:=\!\!\!1$, and $\!\mathcal{B}(n,m,i)\!\!:=\!\!\frac{n\!-\!i}{m\!-\!1}$.
\end{prop}
\vspace{-3mm}
\begin{proof} 
Let $Y^{\Tilde{k}}$ be $\min\{\mathcal{X}_{(\Tilde{k}:\tilde{n}-1)},U\}$. Let $\Tilde{X}_{i:\tilde{n}-1}$ be the difference between {the} $i$th and $(i- 1)$th order statistics of the set $\mathcal{X}$ for $i\geq 1$ and $\Tilde{X}_{1:\tilde{n}-1} = \mathcal{X}_{(1:\tilde{n}-1)}$. Then, we have,
\begin{align*} 
  \\[-2em] 
    \mathcal{X}_{(\tilde{k}:\tilde{n}-1)} = \sum_{j=1}^{\tilde{k}} \Tilde{X}_{j:\tilde{n}-1}
\end{align*}
by construction. Let $\Tilde{Y}_i = \min\{\Tilde{X}_{i:\tilde{n}-1},U\}$. One can see that, from the memoryless property, the random variable (r.v.) $\Tilde{X}_{i:\tilde{n}-1}$ corresponds to the minimum of a set of $(\tilde{n} - i)$ $i.i.d.$ r.v. from an exponential distribution with mean $\frac{(m-1)}{\lambda_e}$. Thus, the r.v. $\Tilde{X}_{i:\tilde{n}-1}$ is also an exponentially distributed r.v., the parameter of which is equal to the sum of the parameters of $(\tilde{n}-i)$ $i.i.d.$ r.v., and precisely its mean is equal to $\frac{(m-1)}{\lambda_e (\tilde{n}-i)}$. This implies that the r.v. $\Tilde{Y}_i$ is the minimum of two independent exponentially distributed r.v.. From~\cite[Prob. 9.4.1]{yates2014probability}, $\Tilde{Y}_i$ is an exponentially distributed r.v. with mean 
\begin{align}\label{eqn:exp_tilde_y_i}  \\[-1.7em]
    \mathbb{E}[\Tilde{Y}_i] = \frac{1}{\lambda_e \frac{(\tilde{n}-i)}{(m-1)} + \lambda_s},
\end{align}
and  $Pr(U > \Tilde{X}_{i:\tilde{n}-1})$ for some $0<i<\tilde{n}$, can be found as:
\begin{align}\label{eqn:prob_ul} \\[-2em] 
    Pr(U > \Tilde{X}_{i:\tilde{n}-1})= \frac{\lambda_e \frac{(\tilde{n}-i)}{(m-1)}}{\lambda_e \frac{(\tilde{n}-i)}{(m-1)} + \lambda_s }.
\end{align}
From the total law of expectation and memoryless property of $U$, we have the following:
\begin{align*}
    \mathbb{E}[Y^{\tilde{k}}] =& Pr( U \leq \Tilde{X}_{1:\tilde{n}-1} ) \mathbb{E}[\Tilde{Y}_1] \\
                    &+ Pr( U \leq \Tilde{X}_{2:\tilde{n}-1} )  Pr( U > \Tilde{X}_{1:\tilde{n}-1} )  (\mathbb{E}[\Tilde{Y}_1] + \mathbb{E}[\Tilde{Y}_2] ) \\
                    % + \ldots +& Pr(U \! \leq\! \Tilde{X}_{i:n-1}) \prod_{i=1}^{\tilde{k}-2} Pr(U \!>\! \Tilde{X}_{i:n-1} )\!  \left(\sum_{i=1}^{\tilde{k}-1}  \mathbb{E}[\Tilde{Y}_i]\right) \\[-0.8em]
                    &+ \ldots + \prod_{i=1}^{\tilde{k}-1} Pr( U > \Tilde{X}_{i:\tilde{n}-1} ) \left( \sum_{i=1}^{\tilde{k}} \mathbb{E}[\Tilde{Y}_i]\right).
\end{align*}
If we rearrange the sum above, we can rewrite $\mathbb{E}[Y^{\tilde{k}}]$ as:
\begin{align}\label{eqn:exp_y_k}
\mathbb{E}[Y^{\tilde{k}}] = \mathbb{E}[\Tilde{Y}_1] + \sum_{j=2}^{\tilde{k}} \left(  \mathbb{E}[\Tilde{Y}_j]  \prod_{i=1}^{j-1} Pr(U > \Tilde{X}_{i:\tilde{n}-1})  \right).
\end{align}
If we plug $\mathbb{E}[\Tilde{Y}_i]$ in~\eqref{eqn:exp_tilde_y_i} and $Pr(U^\ell\!>\!\Tilde{X}_{i:n-1})$ in~\eqref{eqn:prob_ul} for $i\geq1$ into~\eqref{eqn:exp_y_k}, we have:
\begin{align}\label{eqn:y^k}
  \mathbb{E}[Y^{\tilde{k}}] =  \frac{1}{\lambda_e \!+\! \lambda_s} \!+\! \sum_{j=2}^{\tilde{k}}\left( \frac{1}{\lambda_e\frac{(\tilde{n}-j)}{(m-1)}+\lambda_s} \prod_{i=1}^{j-1}\frac{\lambda_e\frac{(\tilde{n}-i)}{(m-1)}}{\lambda_e\frac{(\tilde{n}-i)}{(m-1)}+\lambda_s }  \right)\!. \\[-3em]
\end{align}
By definition, $\mathbb{E}[Y^{\tilde{k}}]=\mathbb{E}[\min\{\mathcal{X}_{(k:n-1)},U^\ell\}]$. It completes the first part of the proof. 
% We plug~\eqref{eqn:y^k} into~\eqref{eqn:subcriber_wo_memory}. Then, we have the following:
% \begin{equation}\label{eqn:cor_subcriber_wo_memory}
%     \bar{\Delta}(\begin{smallmatrix}   k  \\    n   \end{smallmatrix}) = \frac{ \!\frac{1}{\lambda_e\!+\!\lambda_s}\!+\!\!\sum_{j=2}^k\!\left(\!\frac{1}{\lambda_e\!\frac{(n-j)}{(n-1)}+\lambda_s}\!\prod_{i=1}^{j-1}\frac{\lambda_e\!\frac{(n-i)}{(n-1)}}{\lambda_e\!\frac{(n-i)}{(n-1)}\!+\!\lambda_s }  \right)}{Pr(\mathcal{X}_{(k:n-1)}\leq U^\ell)}
% \end{equation}
\begin{figure}[t]
    \centering
    \includegraphics[width=1\linewidth]{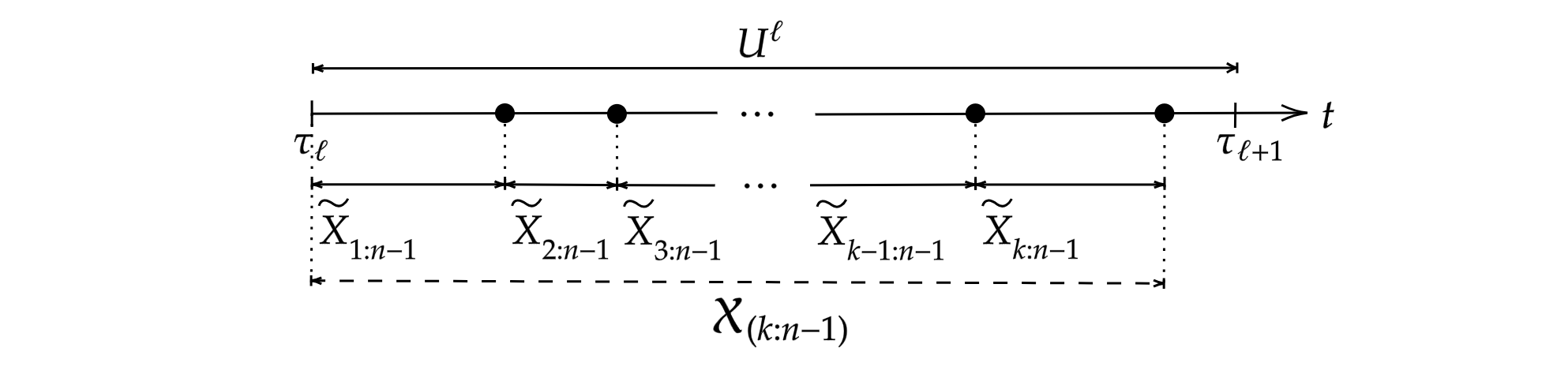}
    \vspace{-0.85cm}
     \caption{Sample timeline of the source update and the activation of edges that are all connected to a unique node. The time of activation of each edge is marked by \!($\bullet).$}\label{fig:x_tilde_timeline}
     \vspace{-0.5cm}
\end{figure}
From the conditional independence of random variables $\tilde{X}_{i:\Tilde{n}-1}$ for $1 \leq i \leq \Tilde{n}-1$ and the memoryless property of random variable $U$ (see Figure~\ref{fig:x_tilde_timeline}) and~\eqref{eqn:prob_ul}, we have the following:
\begin{align*}
    \\[-47pt]
\end{align*}
\begin{align}\label{eqn:subscriber_wo_miss_prob}
        Pr\!\left(\!\mathcal{X}_{(\Tilde{k}:\Tilde{n}-1)}\!\leq \!U \!\right)\! &=\! Pr\!\left(\!  \sum_{j=1}^k \Tilde{X}_{j:\Tilde{n}\!-\!1} \!\leq\! U \!\right) \!\!= \!\prod\limits_{i=1}^{\Tilde{k}}\frac{\lambda_e \frac{(\Tilde{n}-i)}{(m-1)}}{\lambda_e \frac{(\Tilde{n}-i)}{(m\!-\!1)} \!+\!\lambda_s },    
\end{align}
\begin{align*}
    \\[-50pt]
\end{align*}
% We let $\mathcal{B}(n,i)=\frac{(n-i)}{(n-1)}$ and $\mathcal{A}(n,j)= \prod\limits_{i=1}^{j}\frac{\lambda_e \mathcal{B}(n,i)}{\lambda_e \mathcal{B}(n,i) +\lambda_s }$. 
% We combine~\eqref{eqn:cor_subcriber_wo_memory} and~\eqref{eqn:subscriber_wo_miss_prob}. 
which completes the proof.
\end{proof}

% \begin{rem}
% One can compute the closed-form expression for $\agenonsubswo{k}{n}{s}{m}$ and $\agesubswo{k}{n}{s}{m}$ by plugging~\eqref{eqn:prop_min} and~\eqref{eqn:prop_prob} in Proposition~\ref{prop:closed_from_for_order} into~\eqref{eqn:wo_memory_nonsubs} and ~\eqref{eqn:wo_memory_subset_subs} in Theorem~\ref{thm:wo_memory_subset}.
% \end{rem}

% \paragraph{Total Key Subscription and Partial Key Subscription $(s\leq n)$} 

One can compute the closed-form expression for $\agenonsubswo{k}{n}{s}{m}$ and $\agesubswo{k}{n}{s}{m}$ by plugging~\eqref{eqn:prop_min} and  $Pr(\mathcal{X}_{\tilde{k}:\Tilde{n}-1} \leq U)$ in Proposition~\ref{prop:closed_from_for_order} into~\eqref{eqn:wo_memory_nonsubs} and~\eqref{eqn:wo_memory_subset_subs} in Theorem~\ref{thm:wo_memory_subset}. Then, we have the following corollary:
% of Theorem~\ref{thm:wo_memory_subset} and Proposition~\ref{prop:closed_from_for_order} for the version age of $k$-keys for an individual node {\em without memory.} 

\begin{corollary}\label{cor:wo_memory_age}
For an SHN, the version age of $k$-keys for an individual node without memory is the following. For a subscriber node:
\begin{align}\label{eqn:cor_wo_memory_subs}
   \agesubswo{k}{n}{s}{m} = \frac{  \sum\limits_{j=1}^k\left(\frac{\lambda_s \mathcal{A}(n,m,j-1) }{\lambda_e \mathcal{B}(n,m,j) + \lambda_s}   \right) }{   \mathcal{A}(n,m,k) }    \quad \mbox{ w.p. } 1. 
\end{align}
For a nonsubscriber node:
\begin{align}\label{eqn:cor_wo_memory_nonsubs}
   \agenonsubswo{k}{n}{s}{m}  =& \frac{ (n-s) \sum\limits_{j=1}^k\left(\frac{\lambda_s \mathcal{A}(n,m,j-1) }{\lambda_e \mathcal{B}(n,m,j) + \lambda_s}   \right) } {(n\!-\!s) \mathcal{A}(n,m,k)\!+\!(m\!-\!n) \mathcal{A}(n\!+\!1,m,k\!+\!1) }   \\
   &+\! \frac{  (m-n) \sum\limits_{j=1}^{k+1}\left(\frac{\lambda_s   \mathcal{A}(n\!+\!1,m,j\!-\!1)  }{\lambda_e \mathcal{B}(n+1,m,j) + \lambda_s}\right) } {  (n\!-\!s) \mathcal{A}(n,m,k)\!+\!(m\!-\!n) \mathcal{A}(n\!+\!1,m,k\!+\!1) }  \\[-3em]
\end{align}
with probability $1$ where $\mathcal{A}(n,m,j)= \prod\limits_{i=1}^{j}\frac{\lambda_e \mathcal{B}(n,m,i)}{\lambda_e \mathcal{B}(n,m,i)+\lambda_s }$ for $j\geq 1$, $\mathcal{A}(n,m,0)=1,$ and $\mathcal{B}(n,m,i)=\frac{n-i}{m-1}$ {as defined earlier in Proposition~\ref{prop:closed_from_for_order}}.     

% \begin{equation}\label{eqn:prop_min}
%   \agesubswo{k}{n}{s}{m}  =    \sum\limits_{j=1}^{\tilde{k}} \lambda_s \frac{\mathcal{A}(\tilde{n},m,j-1)}{\lambda_e \mathcal{B}(\tilde{n},m,j) + \lambda_s}  
% \end{equation}
% and
% \begin{equation}\label{eqn:prop_prob}
%   Pr(\mathcal{X}_{\tilde{k}:\Tilde{n}-1} \leq U)= {   \mathcal{A}(\tilde{n},m,\Tilde{k}) }  
% \end{equation}
% If the node is not subscribers
% \begin{equation}\label{eqn:prop_min}
%   \mathbb{E}[\min(\mathcal{X}_{\tilde{k}:\Tilde{n}-1},U)]=  \sum\limits_{j=1}^{\tilde{k}} \frac{\mathcal{A}(\tilde{n},m,j-1)}{\lambda_e \mathcal{B}(\tilde{n},m,j) + \lambda_s}  
% \end{equation}
% and
% \begin{equation}\label{eqn:prop_prob}
%   Pr(\mathcal{X}_{\tilde{k}:\Tilde{n}-1} \leq U)= {   \mathcal{A}(\tilde{n},m,\Tilde{k}) }  
% \end{equation}

% \begin{equation}\label{eqn:cor_wo_memory}
%    \agesubswo{k}{n}{n}{m} = \frac{ \frac{\lambda_s}{\lambda_e\!+\!\lambda_s} + \sum\limits_{j=2}^k\left(\frac{\lambda_s}{\lambda_e \mathcal{B}(n,j)\!+\!\lambda_s}\!  \mathcal{A}(n,j\!-\!1)  \right) }{   \mathcal{A}(n,k) }    \mbox{ w.p. } 1 
% \end{equation}
% where $\mathcal{A}(n,j)= \prod\limits_{i=1}^{j}\frac{\lambda_e \mathcal{B}(n,i)}{\lambda_e \mathcal{B}(n,i)\!+\!\lambda_s }$ and $\mathcal{B}(n,i)=\frac{(n-i)}{(n-1)}$. 
\end{corollary}

\begin{rem}
Corollary~\ref{cor:wo_memory_age} is applicable to both total key $(s=n)$ and partial key  $(s<n)$ subscription cases.    
\end{rem}

% One can find the version age of $k$-keys for an individual without memory in a scalable homogeneous network by pluging the corresponding $s$ value into~\eqref{eqn:cor_wo_memory_subs} and~\eqref{eqn:cor_wo_memory_nonsubs}. 
% If one selected $(s=n)$, then,  we have that for the nonscubriber node \begin{equation}\label{eqn:cor_wo_memory_nonsubs_total}
%    \agenonsubswo{k}{n}{n}{m} = \frac{ \sum\limits_{j=1}^{k+1}\left(\frac{\lambda_s   \mathcal{A}(n+1,m,j-1)  }{\lambda_e \mathcal{B}(n+1,m,j) + \lambda_s}\right) } { \mathcal{A}(n+1,m,k+1) }    \mbox{ w.p. } 1 
% \end{equation}
% and the subscriber nodes age is independent from number of subscriber, thus, it is equal to~\eqref{eqn:cor_wo_memory_subs}. 

We have the following scalability results for total key subscription. One can easily see that our results can be extended to partial key subscription when 
$n >> s$. However, due to space limitations, we only provide the total key subscription results for scalability.

\begin{corollary}\label{cor:wo_memory_scale}
The average version age of $k$-keys over an SHN without memory scales as follows when the network size $m$ increases:
\begin{align*}\label{eqn:scalability_wo_mem}
    \lim_{m \to \infty} \ageongraph{k}{n}{n}{m}\!=\!\left(1 \!+\! \frac{\lambda_s}{\alpha\lambda_e} \right)^k \!\!\left( \frac{\lambda_s \!+\! \alpha( \lambda_e \!- \!\lambda_s )}{\alpha \lambda_e} \right) \!- \!1,
\end{align*}
where $n = \lfloor\alpha m\rfloor$ for $\alpha \in (0,1]$.
\end{corollary}

% One can easily see that~\eqref{eqn:scalability_wo_mem} holds if we plug $n=\lfloor \alpha m \rfloor(=|\mathcal{S}|)$ into the age expressions in Corollary~\ref{cor:wo_memory_age}.

% $\mathcal{B}(\lfloor \alpha m \rfloor,m,j)=\alpha$
\begin{proof} We have the following for any finite $j>0$:
\begin{align}
   {\lim_{m \to \infty}\!\! } \mathcal{A}(\lfloor \alpha m \rfloor,m,j) = \prod_{i=1}^j  \left( \frac{\alpha \lambda_e}{\alpha \lambda_e + \lambda_s} \right) =  \left( \frac{\alpha \lambda_e}{\alpha \lambda_e + \lambda_s} \right)^j. \\[-3em]
 \end{align}
If we plug it into~\eqref{eqn:cor_wo_memory_subs}, then we have the following as $m \to \infty$:
  \begin{align}\label{eqn:subs_scale}
    \agesubswo{k}{n}{n}{m} \!=\! \frac{\sum^{k}_{j=1}\!\! \left(\frac{\alpha \lambda_e}{\alpha \lambda_e + \lambda_s} \right)^{j-1}}{ \left(\frac{\alpha \lambda_e}{\alpha \lambda_e + \lambda_s} \right)^{k} } \frac{\lambda_s}{\alpha\lambda_e \!+\! \lambda_s} = \left( 1 \!+ \!\frac{\lambda_s}{\alpha \lambda_e} \right)^k\!\! - \!1. \\[-3em]
 \end{align}
 % &=  \frac{ ( \frac{\alpha \lambda_e}{\alpha \lambda_e + \lambda_s})^k -1 }{  ( \frac{\alpha \lambda_e}{\alpha \lambda_e + \lambda_s} )^k } = ( 1 + \frac{\lambda_s}{\alpha \lambda_e} )^k - 1 
For $\alpha=1$, that is, $n=m$, there is no nonsubscriber node by construction. Thus, we have $\agesubswo{k}{m}{m}{m}=\ageongraph{k}{m}{m}{m}$. This completes the proof for $\alpha=1$. 
For $\alpha \in (0,1)$, we consider the nonsubcriber nodes. Then, we have the following as $m \to \infty$, from Corollary~\ref{eqn:cor_wo_memory_nonsubs}: 
 \begin{align}\label{eqn:nonsubs_scale}
    \agenonsubswo{k}{n}{n}{m} &\!=\!\frac{\sum^{k+1}_{j=1} \left(\frac{\alpha \lambda_e}{\alpha \lambda_e + \lambda_s} \right)^{j-1}}{\frac{\alpha\lambda_e + \lambda_s}{\lambda_s} \left(\frac{\alpha \lambda_e}{\alpha \lambda_e + \lambda_s} \right)^{k+1} }  = \left( 1 + \frac{\lambda_s}{\alpha \lambda_e} \right)^{k+1} - 1, \\[-3em]
 \end{align}
If we plug the expression for $\agesubswo{k}{n}{n}{m}$ and $\agenonsubswo{k}{n}{n}{m}$ above into~\eqref{eqn:age_on_graph}, then we have the expression in Corollary~\ref{cor:wo_memory_scale}.\end{proof}

\vspace{-2mm}
\section{Discussion and Numerical Results}\label{sec:discuss}
In this section, we provide numerical results for the version age of $k$-keys under memory and memoryless schemes. In particular, we compare empirical results obtained from simulations with our analytical results.
\vspace{-1.5mm}
\subsection{Full Subscription}
\vspace{-1.5mm}
In the first set of results, we consider an SHN in full subscription ($s=n$ and $m=n$). Fig.~\ref{fig:theoric_experimental} depicts the simulation and the theoretical results for both $\ageongraph{k}{n}{n}{n}$ and $\ageongraphwo{k}{n}{n}{n}$ as a function of gossip rate $\lambda_e$. The simulation results for $\ageongraph{k}{n}{n}{n}$ and $\ageongraphwo{k}{n}{n}{n}$ align closely with the theoretical calculations provided in Theorems~\ref{thm:hetero_w_memory_age} and \ref{thm:hetero_wo_memory_age}, respectively. In both schemes, we observe that the version age of $k$-keys for a node decreases with the increase in gossip rate $\lambda_e$ for fixed $n$ and $\lambda_s$. Fig.~\ref{fig:theoric_experimental}(a) also depicts the version age of $k$-keys as a function of the pair $(k,n)$. We observe, in Fig.~\ref{fig:theoric_experimental}(a), that $\ageongraph{k}{n}{n}{n}$ increases with $k$ for a fixed $\lambda_e$ and $n$ and it decreases as $n$ increases for fixed $\lambda_e$ and $k$. On the one hand, Fig.~\ref{fig:theoric_experimental}(b) confirms that the earlier observation regarding $\ageongraph{k}{n}{n}{n}$ and its connection with network parameters $\lambda_e$ and $(k,n)$ also hold true for $\ageongraphwo{k}{n}{n}{n}$. %On the other hand,
{However,} we observe, in Fig.~\ref{fig:theoric_experimental}, that the version age of $k$-keys {\em with} memory scheme $\Delta^k$ is {\em less than} the version age of $k$-keys {\em without} memory scheme $\ageongraphwo{k}{n}{n}{n}$ for the same values of $\lambda_e,\lambda_s$ and $(k,n)$. 

\begin{figure}[t]
\centering
\vspace{-0.5cm}
\subfloat[]{\includegraphics[width=0.47\linewidth]{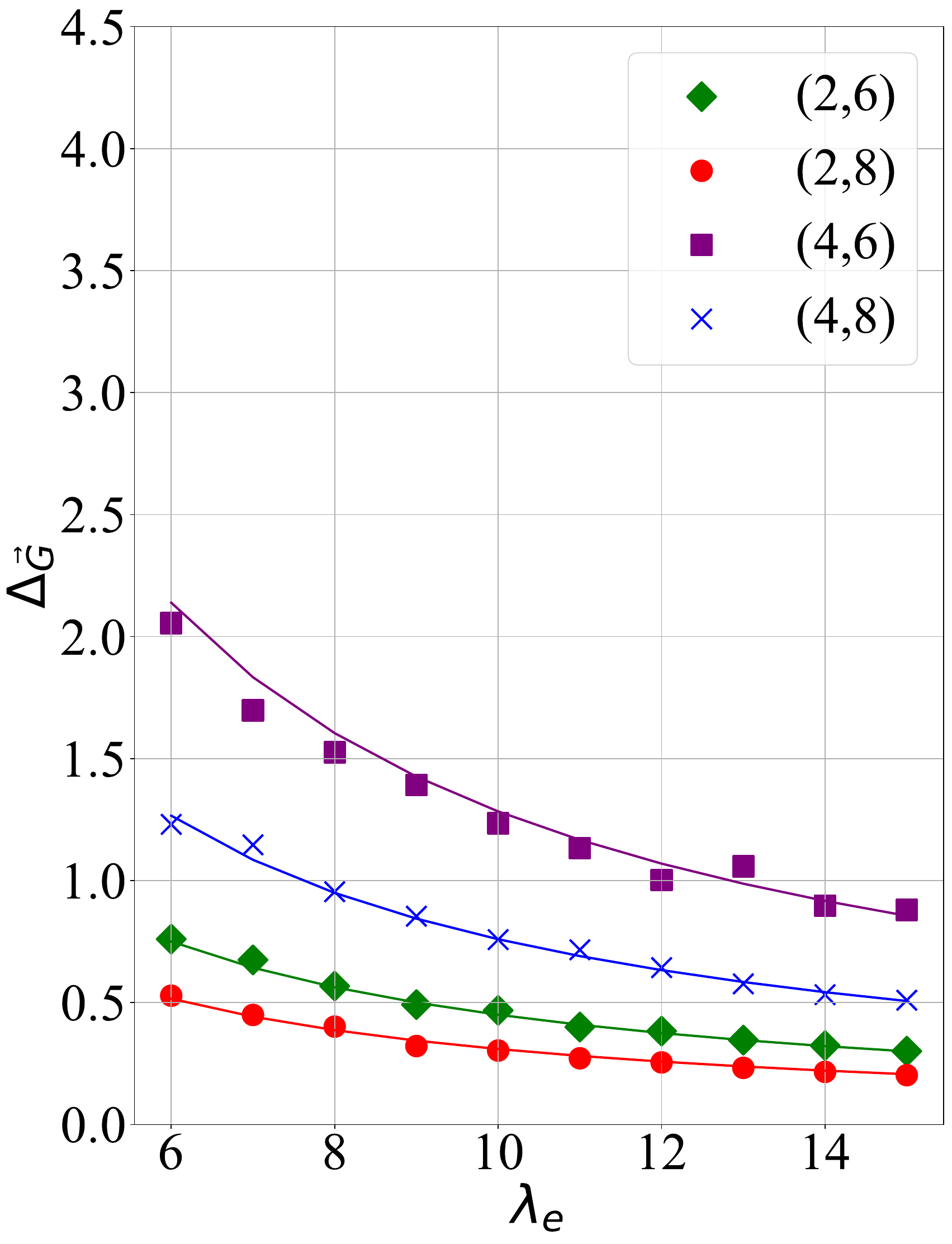}}
\hfil
\subfloat[]{\includegraphics[width=0.47\linewidth]{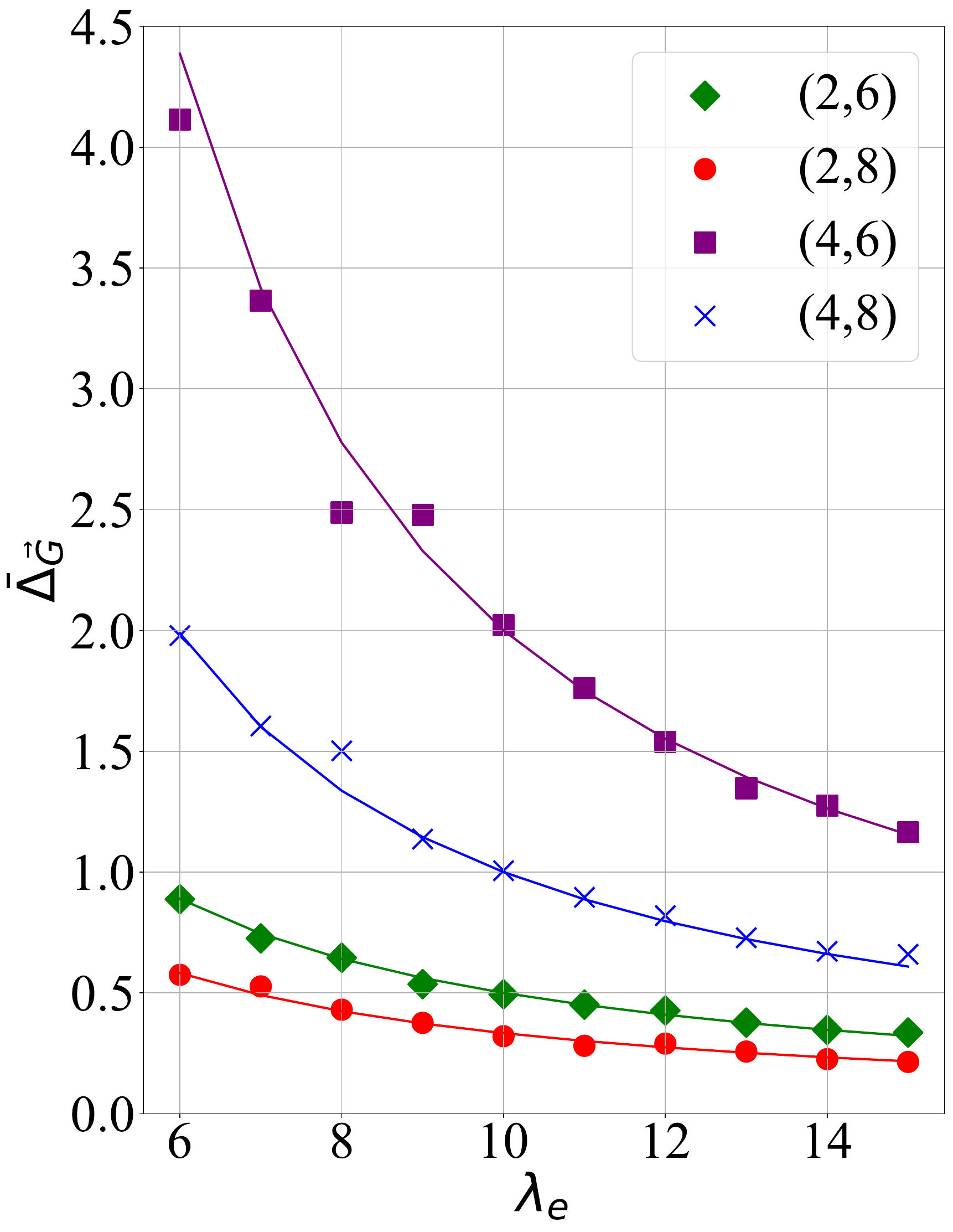}}
    \vspace{-0.2cm}
\caption{(a)$\ageongraph{k}{n}{n}{n}$ and (b) $\ageongraphwo{k}{n}{n}{n}$ as a function of $\lambda_e$ on a scalable homogeneous network when $\lambda_s=10$. Solid lines in Fig.~\ref{fig:theoric_experimental}(a) show theoretical $\ageongraph{k}{n}{n}{n}$ while solid lines in Fig.~\ref{fig:theoric_experimental}(b) show theoretical $\ageongraphwo{k}{n}{n}{n}$. Simulation results for $(2,6)$,$(2,8)$,$(4,6)$,$(4,8)$ TSS are marked by $\blacklozenge,\bullet,\blacksquare,\times$, respectively.}
	\label{fig:theoric_experimental}
\vspace{-0.5cm}
\end{figure}

Fig.~\ref{fig:asym_on_on} depicts $\ageongraph{k}{n}{n}{n}$ as a function of the number of the nodes $m(=n)$ for various gossip rates $\lambda_e$. We observe, in Fig.~\ref{fig:asym_on_on}, that $\ageongraph{k}{n}{n}{n}$ converges to $\frac{k\lambda_s}{\lambda_e}$ as $n$ grows. It aligns with the results obtained in Corollary~\ref{cor:w_memory_scale} for $\alpha=1$. 
\begin{figure}[t]
    \centering
    \includegraphics[width=0.9\linewidth]{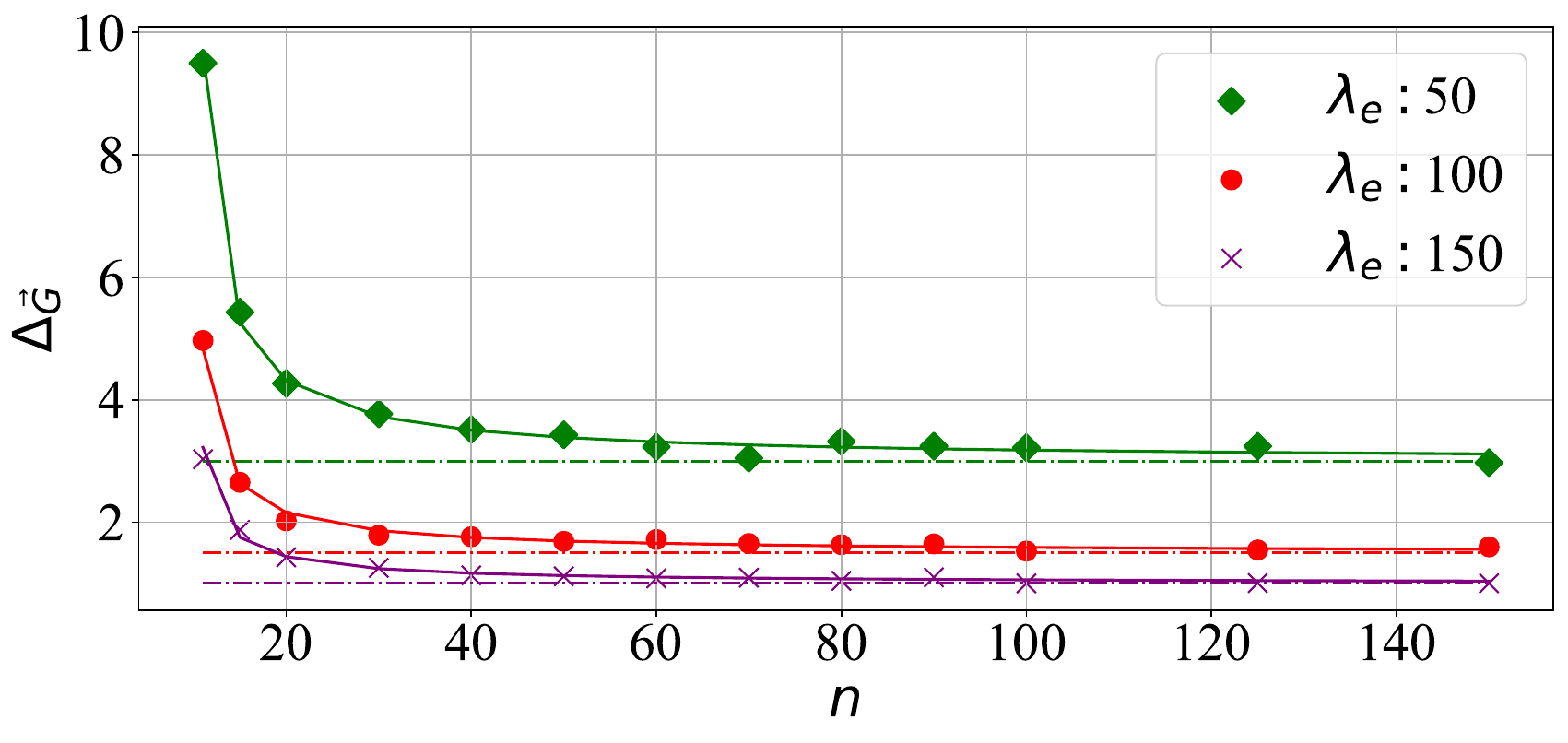}
    \caption{ $\ageongraph{k}{n}{n}{n}$ as a function of $m\!=\!n$ when $k\!=\!10$,$\lambda_s=15$ for $\alpha=1$. Solid lines show the theoretical results for $\ageongraph{k}{n}{n}{n}$ while simulation results for $\lambda_e\!=\!\{50,100,150\}$ selections are marked by $\blacklozenge,\bullet,\times$, respectively. Dashed lines show the theoretical asymptotic value of $\ageongraph{k}{n}{n}{n}$ on $n$.}
    \label{fig:asym_on_on}
    \vspace{-0.5cm}
\end{figure}
\begin{figure}[t]
    \centering
    \includegraphics[width=0.9\linewidth]{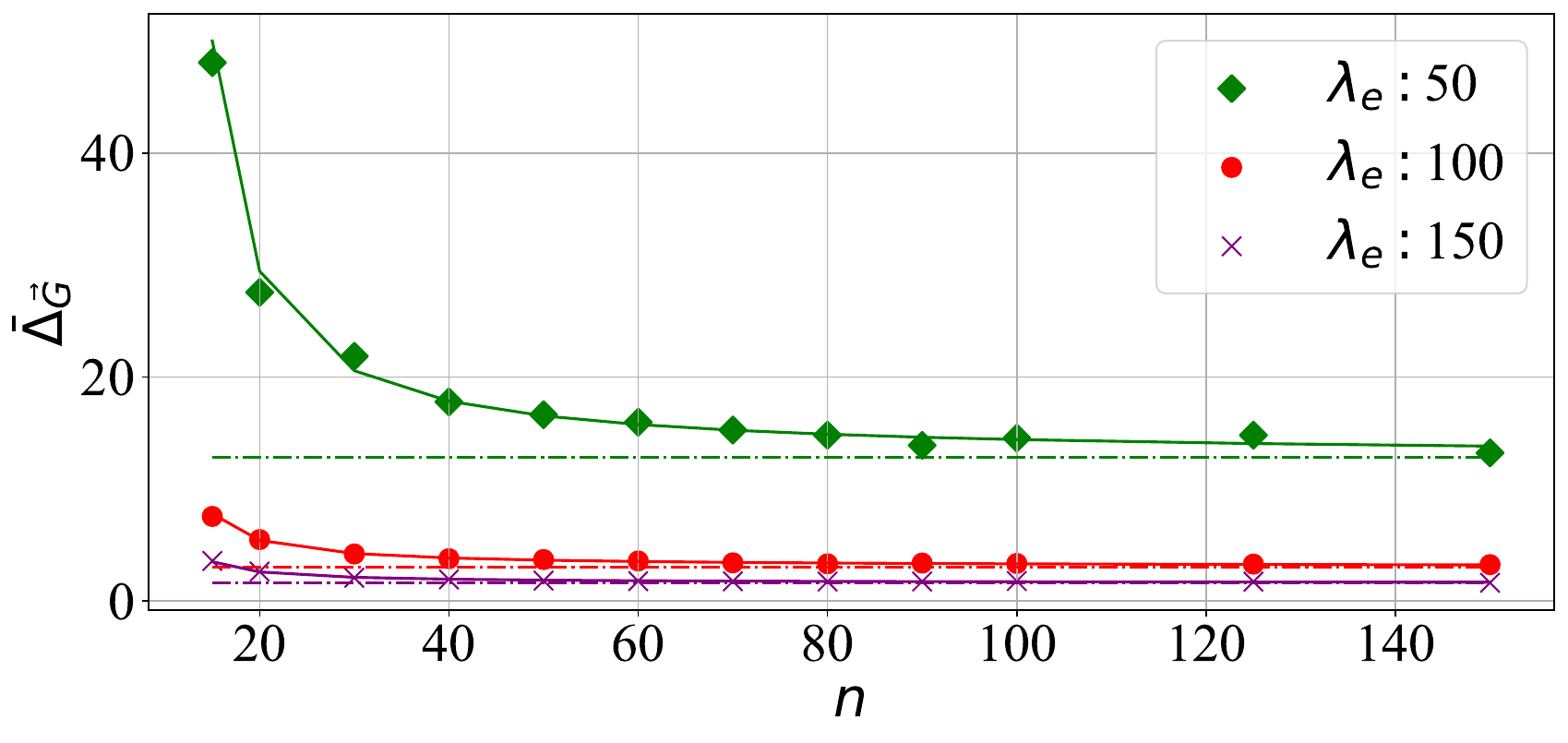}
    \caption{ $\ageongraphwo{k}{n}{n}{n}$ as a function of $m=n$ when $k=10$ and $\lambda_s=15$. Solid lines show the theoretical results for $\ageongraphwo{k}{n}{n}{n}$ while simulation results for $\lambda_e=\{50,100,150\}$ selections are marked by $\blacklozenge,\bullet,\times$, respectively. Dashed lines show the theoretical asymptotic value of $\ageongraphwo{k}{n}{n}{n}$ on $n$.}
    \label{fig:asym_on_wo}\vspace{-8mm}
\end{figure}
Fig.~\ref{fig:asym_on_wo} depicts $\ageongraphwo{k}{n}{n}{n}$ as a function of the number of the nodes $n$ for various gossip rates $\lambda_e$. We observe, in Fig.~\ref{fig:asym_on_wo}, that $\ageongraphwo{k}{n}{n}{n}$ converges to $\left(1 + \lambda_s/\lambda_e\right)^k-1$ as $n$ grows. It aligns with the results obtained in Corollary~\ref{cor:wo_memory_scale} for $\alpha=1$.
\vspace{-4mm}
\subsection{Partial Key Subscription}

% We consider a scalable homogeneous network in partial key subscription for various $s$ and $m$. Fig.~\ref{fig:bounds_on_k} depicts $\ageongraph{k}{8}{5}{12}$ and $\ageongraph{k}{8}{5}{18}$ as a functions of $\lambda_e$. The simulation results for $\ageongraph{k}{8}{5}{12}$ and $\ageongraph{k}{8}{5}{18}$ follows the upper and lower bound given in~\ref{thm:w_memory_r_subs}. We observe that $\ageongraph{k}{n}{s}{m}$ increases as $k$ increases for fixed $n,s,m$. Furthermore, for the same $k$, $\ageongraph{k}{n}{s}{m}$ increases as the number of total node $m$ increases. We also observe that, in Fig.~\ref{fig:bounds_on_k}, the lower and upper bound get closer to each other as the gossip rate $\lambda_e$ increases, which aligns with~\ref{thm:w_memory_r_subs}. 

% We first fix the number of subscribers $s$ and the total number of keys $n$, then we analyze {the} version age of $k$-keys for different $s$ and $m$ values.
In this set of results, we consider an SHN in partial key subscription $(s < n)$. Fig.~\ref{fig:bounds_on_k} depicts $\ageongraph{k}{8}{5}{12}$ and $\ageongraph{k}{8}{5}{18}$ as functions of $\lambda_e$ and $k$. The simulation results for $\ageongraph{k}{8}{5}{12}$ and $\ageongraph{k}{8}{5}{18}$ follow the upper and lower bounds given in Theorem~\ref{thm:w_memory_r_subs} for $k \in \{2,3,4\}$. We observe that $\ageongraph{k}{n}{s}{m}$ increases with $k$ for fixed values of $n$, $s$, and $m$. Furthermore, for the same $k$, $\ageongraph{k}{n}{s}{m}$ increases as the total number of nodes, $m$, increases. Additionally, in Fig.~\ref{fig:bounds_on_k}, the lower and upper bounds get closer to each other as the gossip rate $\lambda_e$ increases, which aligns with Theorem~\ref{thm:w_memory_r_subs}.

\begin{figure}[t]
%\vspace{-0.5cm}
\centering
\subfloat[]{\includegraphics[width=0.48\linewidth]{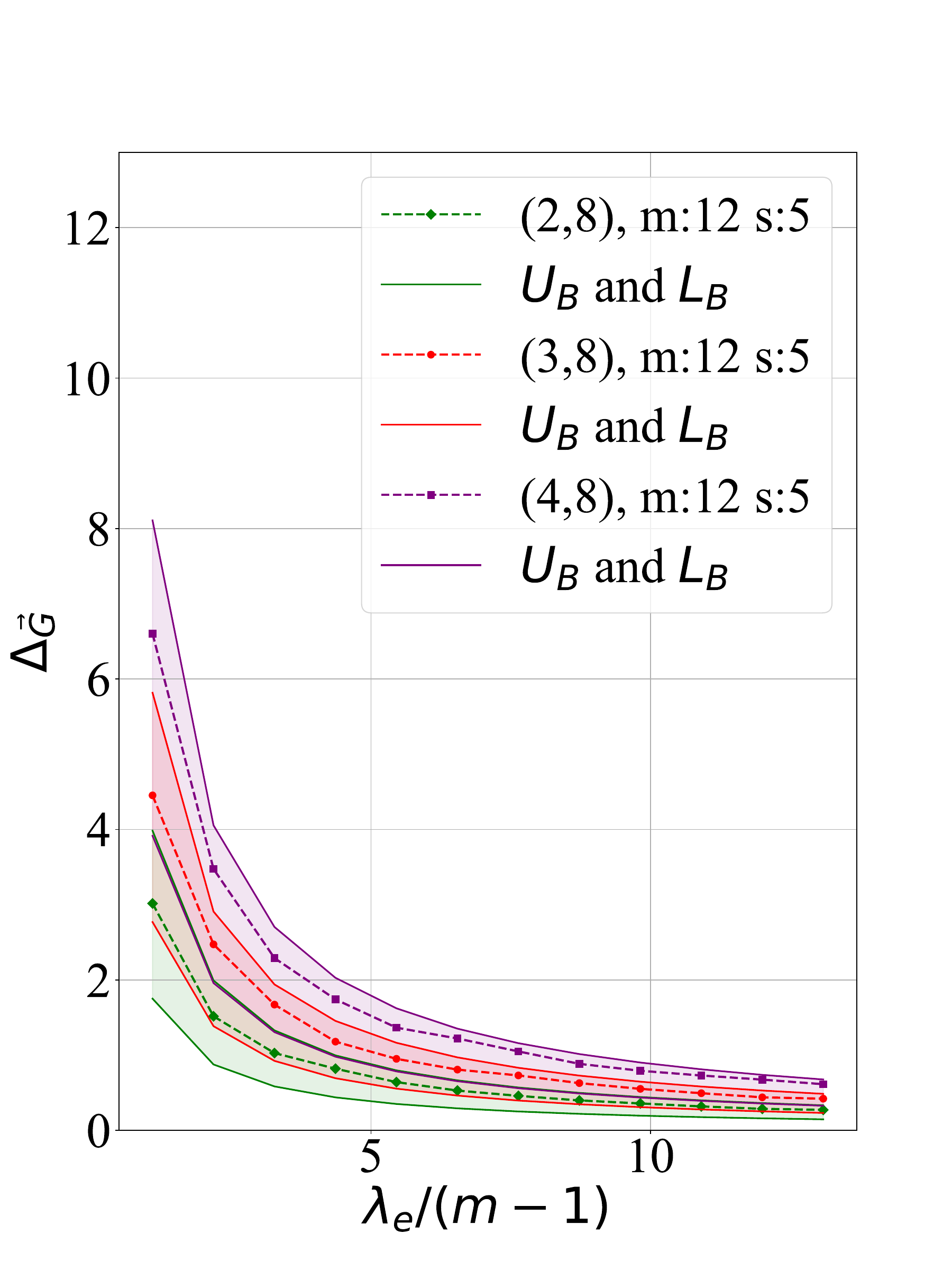}}
\hfil
\subfloat[]{\includegraphics[width=0.48\linewidth]{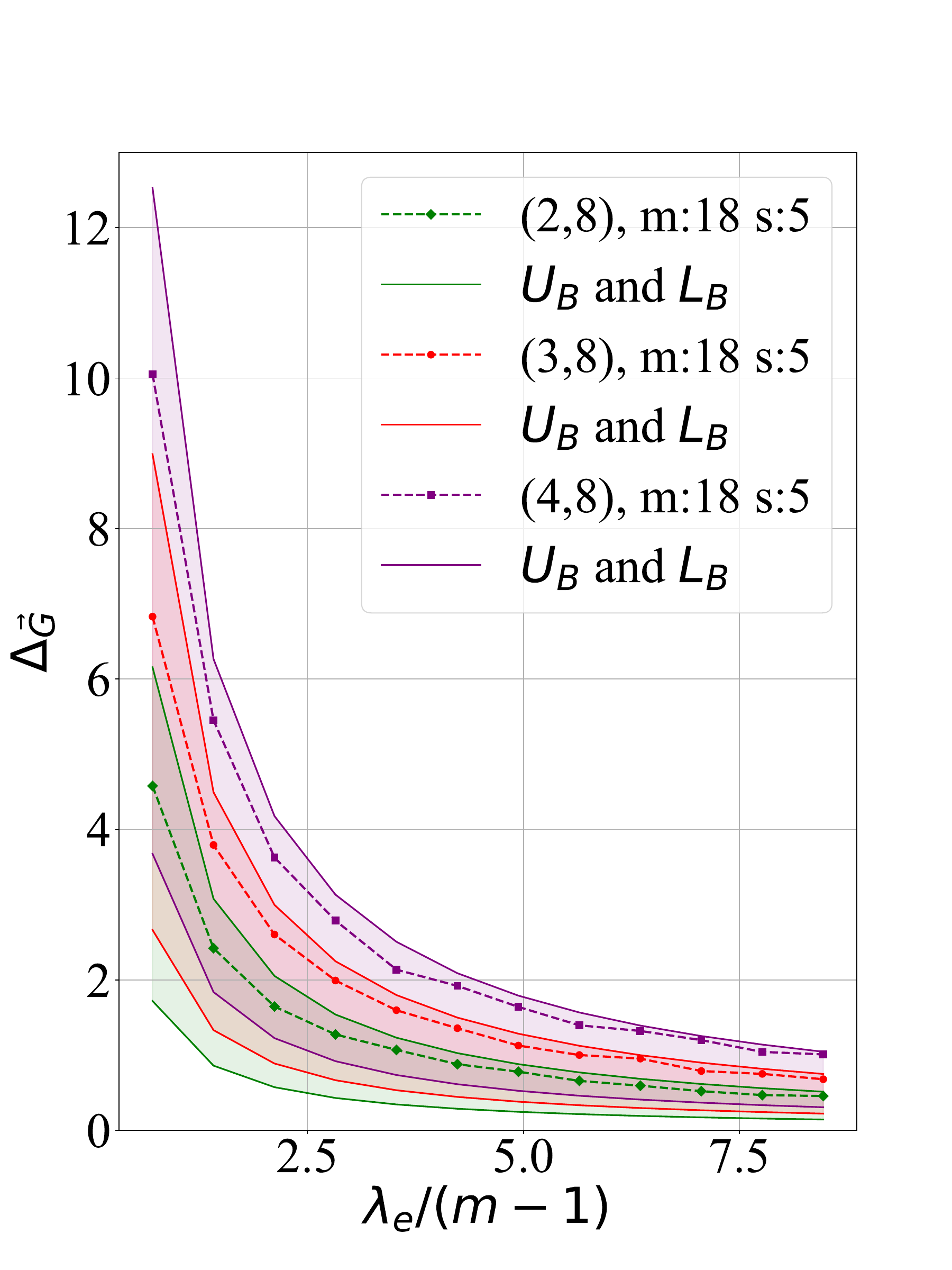}}
    \vspace{-0.2cm}
\caption{(a) $\ageongraph{k}{8}{3}{12}$ and (b) $\ageongraph{k}{8}{3}{18}$ as a function of $\lambda_e$ on an SHN when $\lambda_s=10$ for $k=\{2,3,4\}$. Solid lines in Fig.~\ref{fig:bounds_on_k} show upper and lower bounds for $\ageongraph{k}{n}{s}{m}$. Simulation results for $k=\{2,3,4\}$ are marked by $\blacklozenge,\bullet,\blacksquare$, respectively.}
	\label{fig:bounds_on_k}
\vspace{-0.5cm}
\end{figure}

Fig.~\ref{fig:bound_on_s} depicts $\ageongraph{4}{8}{s}{12}$ and $\ageongraph{4}{8}{s}{18}$ as functions of $\lambda_e$ and $s$. The simulation results align with the upper and lower bounds provided in Theorem~\ref{thm:w_memory_r_subs} for $s \in \{0,3,5,8\}$. In Fig.~\ref{fig:bound_on_s}, we observe that the average version age of $k$-keys over the graph tends to converge for different numbers of subscribers as the gossip rate increases, while a higher number of subscriber nodes leads to a higher average version age of $k$-keys over the graph. It is important to note that, while we do not currently have a formal proof for this claim, it is supported by a strong intuitive reasoning. Specifically, as the gossip rate $\lambda_e$ increases, the probability of a node collecting enough keys to decode the version update before the arrival of a new update also increases, that is, $Pr(\tau^*_\ell \leq \tau_{\ell+1})$ increases for any $\ell$. Consequently, the service time decreases.
% regardless of whether the previous keys are distributed among specific sets of nodes (subscribers) or are randomly distributed.
In other words, the probability that a version is "early stopped" by a newer version (i.e., $Pr(\tau_{\ell_1} \leq \tau_{\ell_2})$ for $\ell_1 > \ell_2$) decreases as the probability of decoding a status update before a new one is generated increases. This makes the number of subscribers insignificant for the average version age of $k$-keys over the graph. Therefore, as the gossip rate increases, the average version age of $k$-keys over the graph for different numbers of subscribers tends to converge.
\begin{figure}[!t]
\centering
\vspace{-0.4cm}
    \subfloat[]{
        \includegraphics[width=0.48\linewidth]{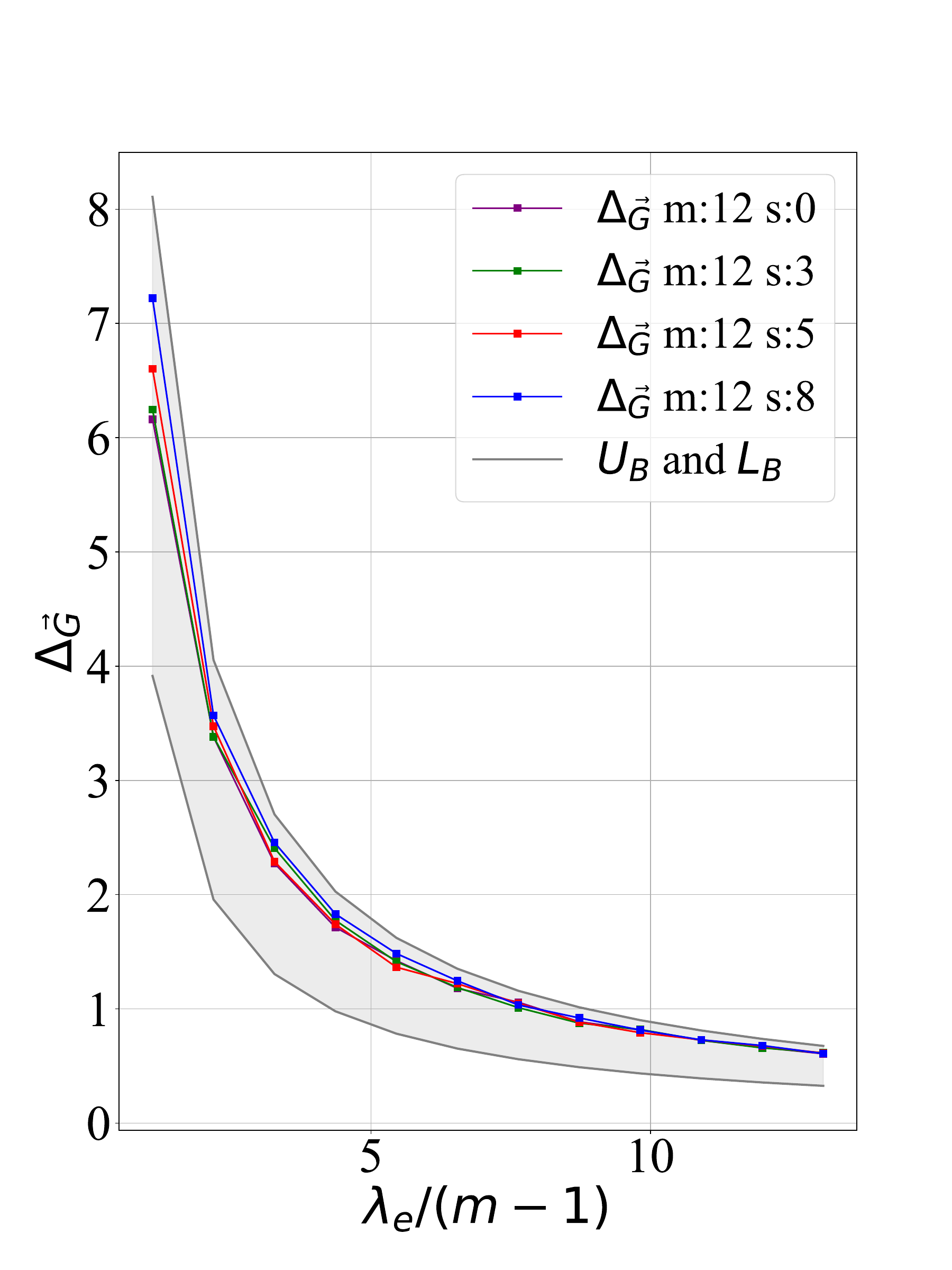}}
     \hfil %
    \subfloat[]{
        \includegraphics[width=0.48\linewidth]{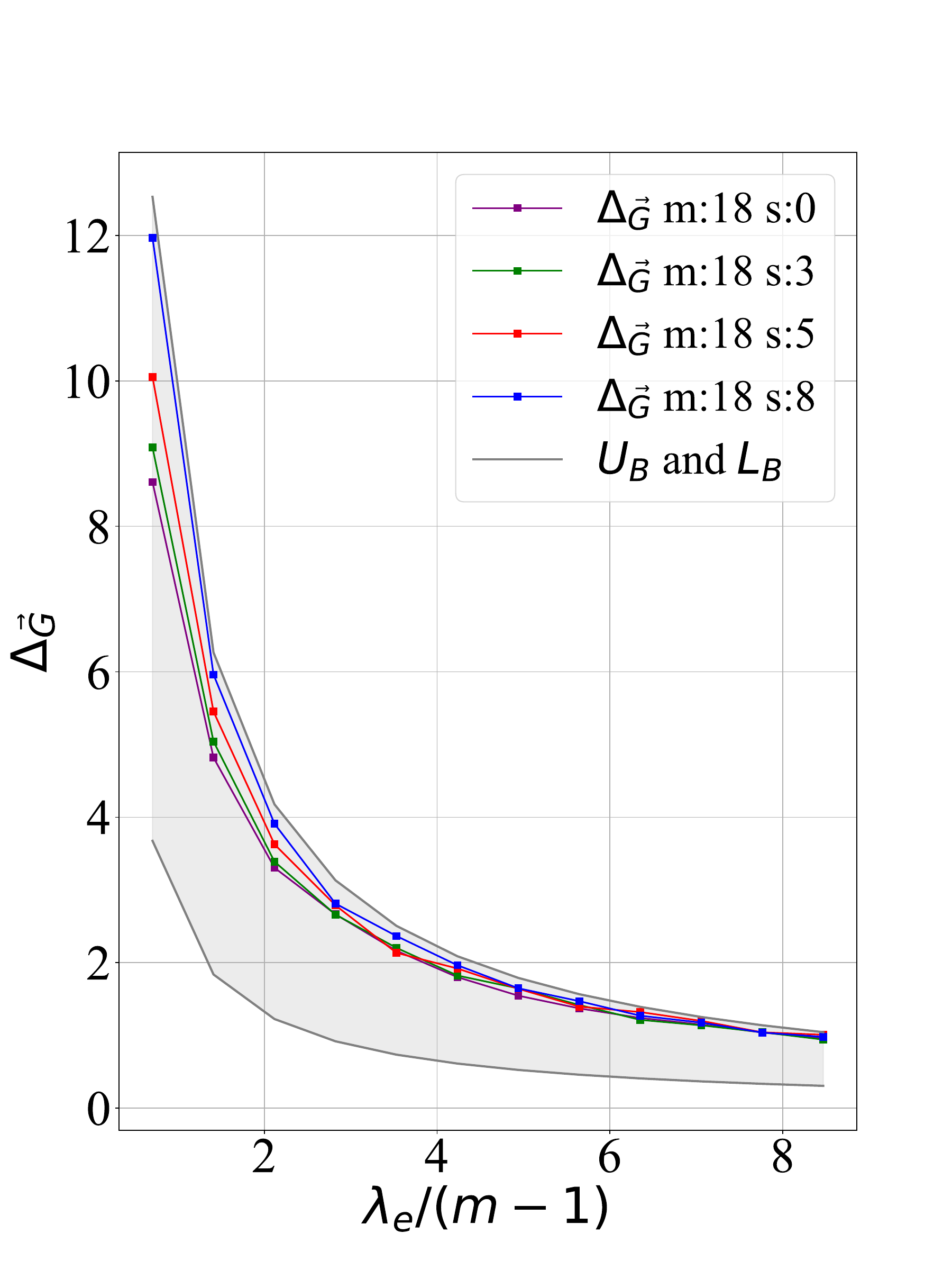}}
         \vspace{-0.2cm}
	\caption{(a) $\ageongraph{4}{8}{s}{12}$, (b) $\ageongraph{4}{8}{s}{18}$ as a function of $\lambda_e$ on an SHN when $\lambda_s=10$ for $s=\{0,3,5,8\}$. Solid lines in Fig.~\ref{fig:bound_on_s} show upper and lower bounds for $\ageongraph{k}{n}{s}{m}$. Simulation results for $s=\{0,3,5,8\}$ are marked by purple, green, red and blue lines, respectively.}
	\label{fig:bound_on_s}
 \vspace{-0.3cm}
\end{figure}
\begin{figure}[!t]
\centering
\vspace{-0.6cm}
    \subfloat[]{
        \includegraphics[width=0.48\linewidth]{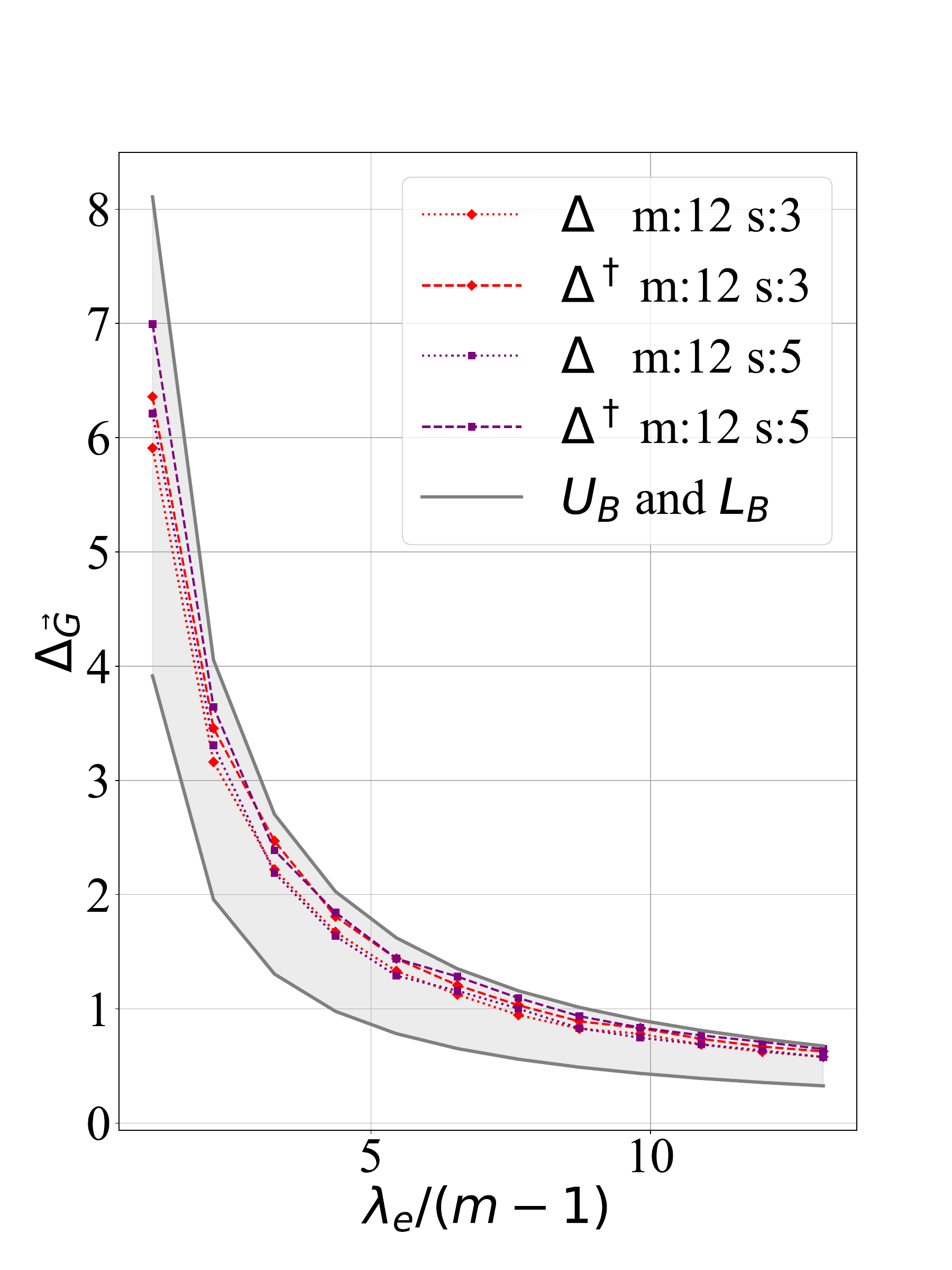}}
        \hfil %
    \subfloat[]{   
        \includegraphics[width=0.48\linewidth]{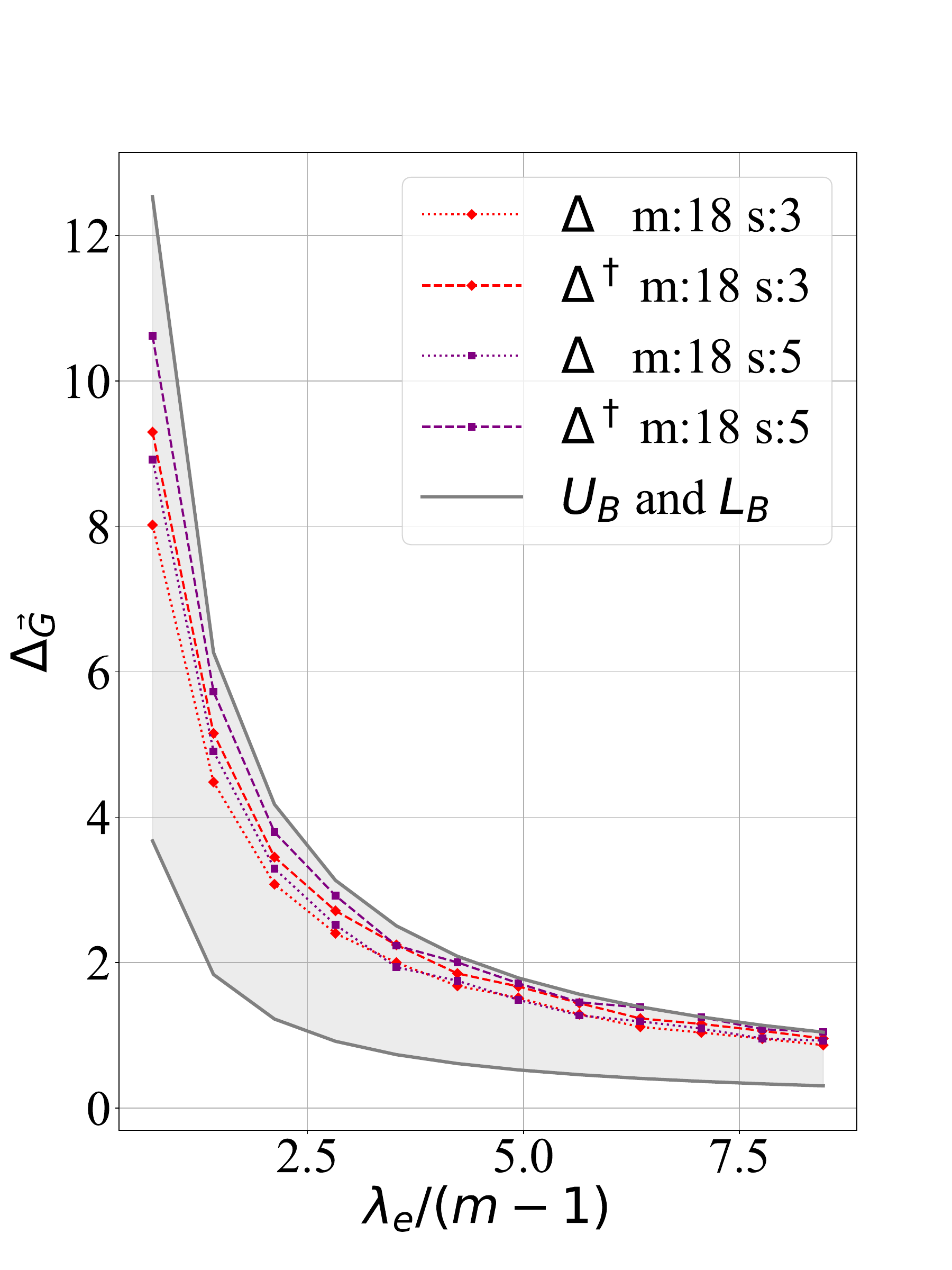}}
        \hfil\vspace{-0.4cm}
    \subfloat[]{
        \includegraphics[width=0.48\linewidth]{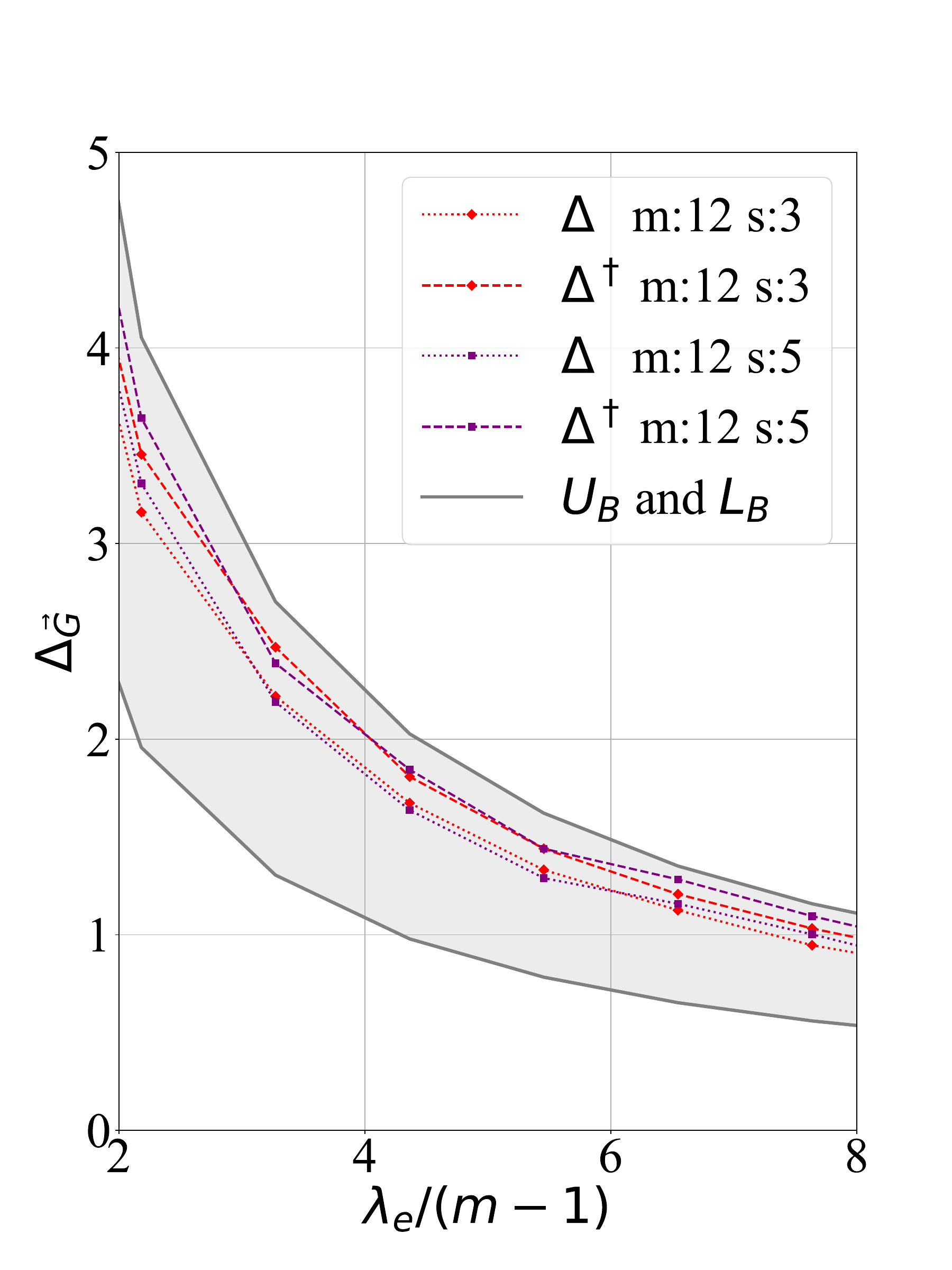}}
       \hfil 
    \subfloat[]{  
        \includegraphics[width=0.48\linewidth]{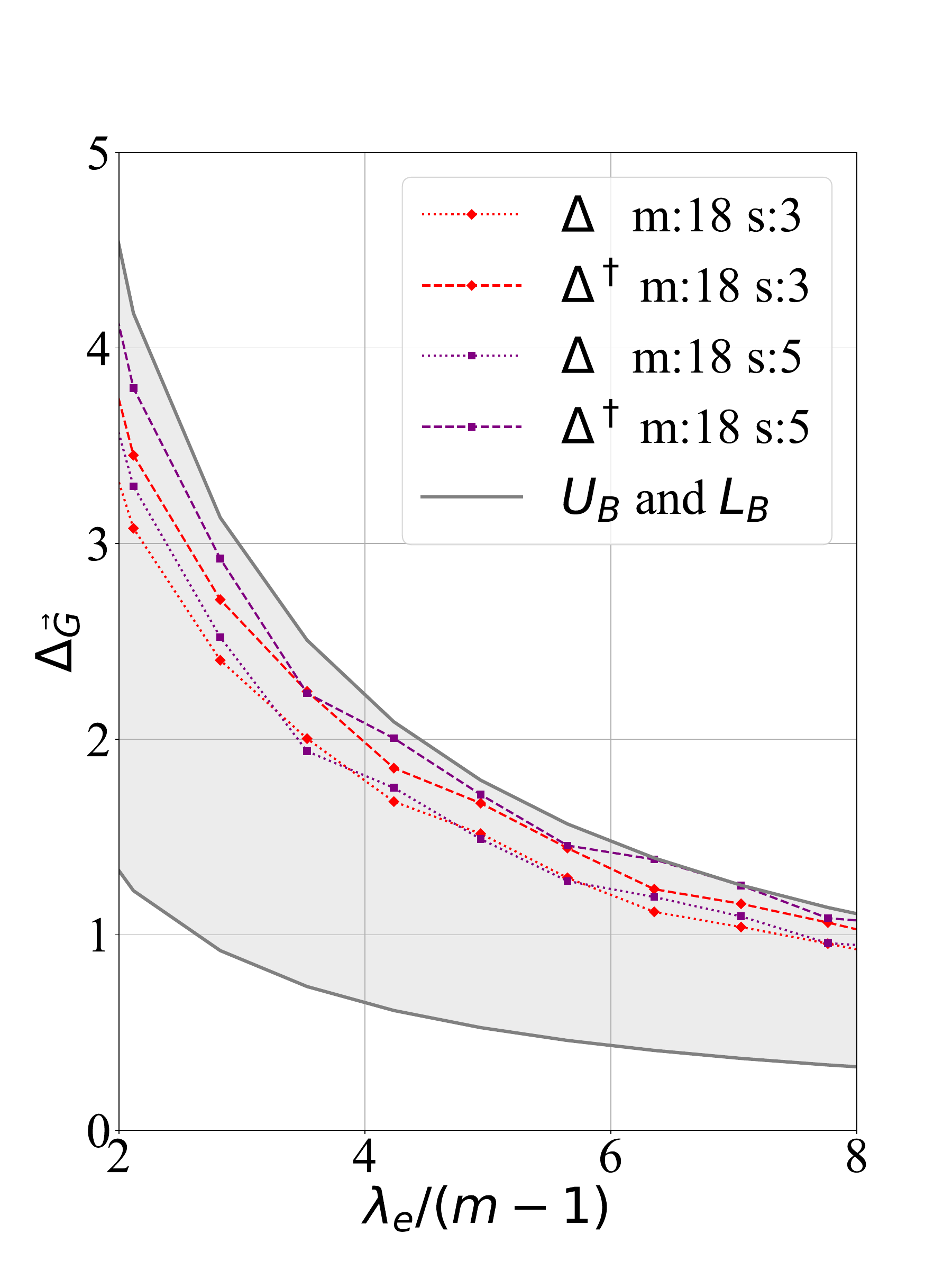}}
\vspace{-0.2cm} 
	\caption{(a) $\agesubs{4}{8}{s}{12}$ and $\agenonsubs{4}{8}{s}{12}$, (b) $\agesubs{4}{8}{s}{18}$ and $\agenonsubs{4}{8}{s}{18}$ as a function of $\lambda_e$ on an SHN when $\lambda_s=10$ for $s=\{3,5\}$. (c) and (d) offer a closer examination of the same plots presented in (a) and (b), respectively. Solid lines show upper and lower bounds for $\ageongraph{k}{n}{s}{m}$. Dashed lines and dotted lines show $\agesubs{4}{8}{s}{m}$ and  $\agenonsubs{4}{8}{s}{m}$, respectively. Simulation results for $s=\{3,5\}$ are marked by $\blacklozenge,\blacksquare$, respectively.}
	\label{fig:bounds_subs_nonsubs}
 \vspace{-0.5cm}
\end{figure}

Fig.~\ref{fig:bounds_subs_nonsubs} depicts $\agesubs{4}{8}{s}{12}$ and $\agenonsubs{4}{8}{s}{18}$ as functions of $\lambda_e$ and $s$. We observe that, in Fig.~\ref{fig:bounds_subs_nonsubs}, $\agesubs{4}{8}{s}{12}$ is lower than $\agenonsubs{4}{8}{s}{18}$ on the same network. Additionally, $\agesubs{4}{8}{s}{12}$ converges as the gossip rate increases independently from the number of subscribers $s$, which aligns with our earlier intuition. Furthermore, the simulation results align with the upper and lower bounds provided for $\ageongraph{4}{8}{s}{12}$. We observe that, in Fig.~\ref{fig:bounds_subs_nonsubs}, both the version age of $k$-keys for subscriber and nonsubscriber nodes adhere to upper and lower bounds in Theorem~\ref{thm:w_memory_r_subs} even if we proved them only for  $\ageongraph{4}{8}{s}{12}$. 
% , but the both the version age of $k$-keys for subscriber and nonsubscriber nodes adhere to them, as well.

In contrast to memory schemes, we have an explicit formula for $\ageongraphwo{k}{n}{s}{m}$ in memoryless scheme. Therefore, any given numerical results include the theoretical values for a given set of network parameters and the simulation results. Fig.~\ref{fig:wo_memory_subset} depicts both theoretical values and simulations results for $\agesubswo{4}{8}{3}{m}$ and $\agenonsubswo{4}{8}{3}{m}$ as functions of $\lambda_e$ and $m$. We observe that, in Fig.~\ref{fig:wo_memory_subset}, the simulation results are consistent with result of Theorem~\ref{thm:wo_memory_subset} (and Corollary~\ref{cor:wo_memory_age}) for both $\agesubswo{4}{8}{3}{m}$ and $\agenonsubswo{4}{8}{3}{m}$. Furthermore, $\agesubswo{4}{8}{3}{m}$ is less than $\agenonsubswo{4}{8}{3}{m}$ for the same gossip rate $\lambda_e$ and $m$ while both increases with the number of nodes $m$. 

% \mbnote{We discuss Fig.13 earlier than Fig.12. Shall we put this paragraph after the next one? Or we can change the order of the Figs.} 

Fig.~\ref{fig:cost_of_subs} depicts $\agesubswo{k}{n}{s}{m}$ and $\agenonsubswo{k}{n}{s}{m}$ as functions of $s$ for various $m$ values. The corresponding plots are obtained through Corollary~\ref{cor:wo_memory_age}. We observe that, in Fig.~\ref{fig:cost_of_subs}, $\agenonsubswo{k}{n}{s}{m}$ exponentially increases as the number of subscribers $s$ increases for fixed $n$ and $m$. However, $\agesubswo{k}{n}{s}{m}$ remains constant for any $s$. By definition, we also observe that the average version of $k$-keys over the network, $\ageongraphwo{k}{n}{s}{m}$, exponentially increases as the number of subscribers $s$ increases for fixed $n$ and $m$. These observations show that making a node subscriber to the source drastically decreases the version age of $k$-keys for the node, while it exponentially increases the version age of $k$-keys for the rest of the nonsubscriber nodes. Consequently, the average version age of $k$-keys over the network also increases exponentially, though at a slower rate compared to the increase in the age for nonsubscriber nodes.

\begin{figure}[t]
\centering
% \vspace{-0.6cm}
    \subfloat[]{
        \includegraphics[width=0.48\linewidth]{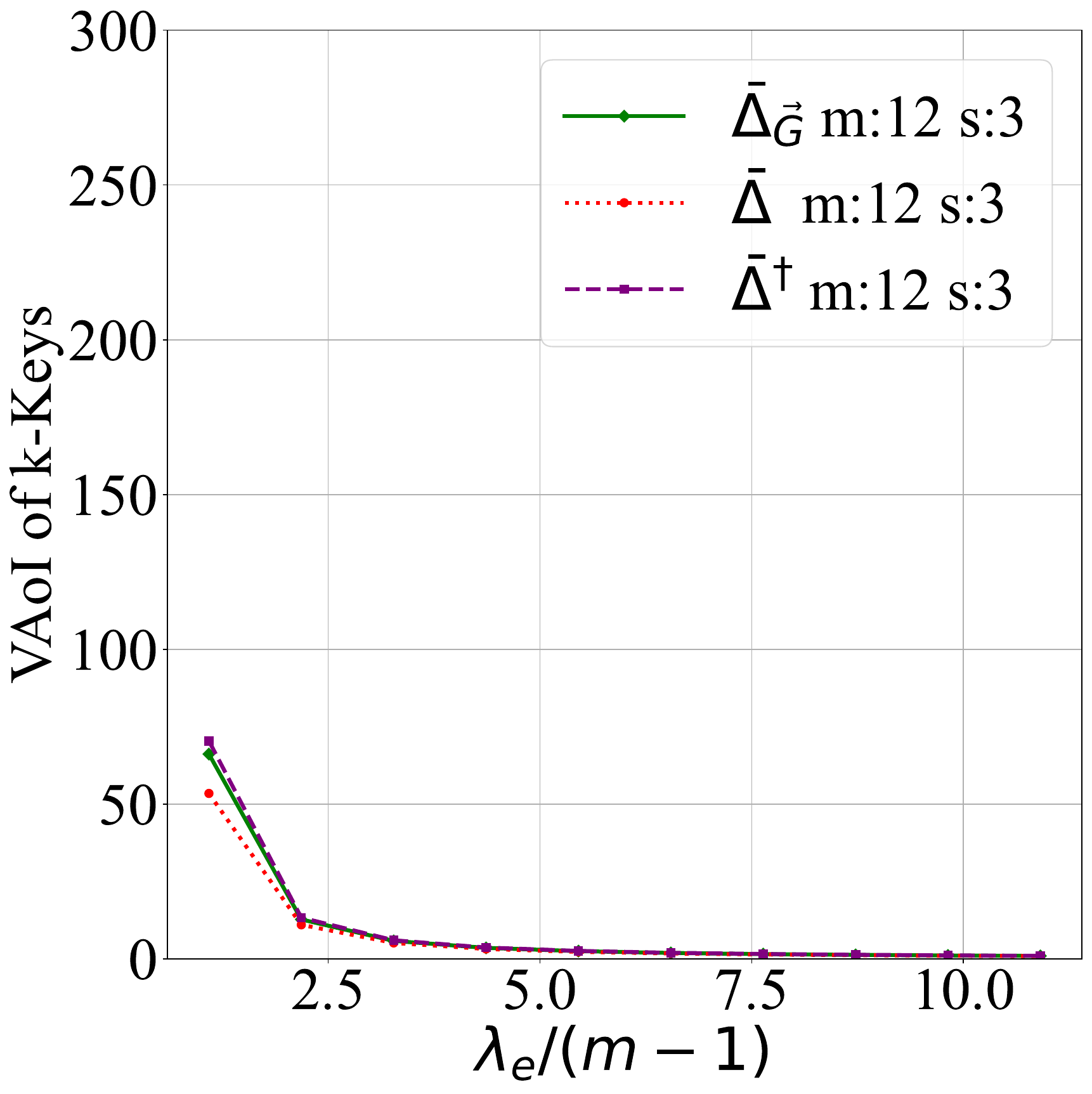}}
    \subfloat[]{   
        \includegraphics[width=0.48\linewidth]{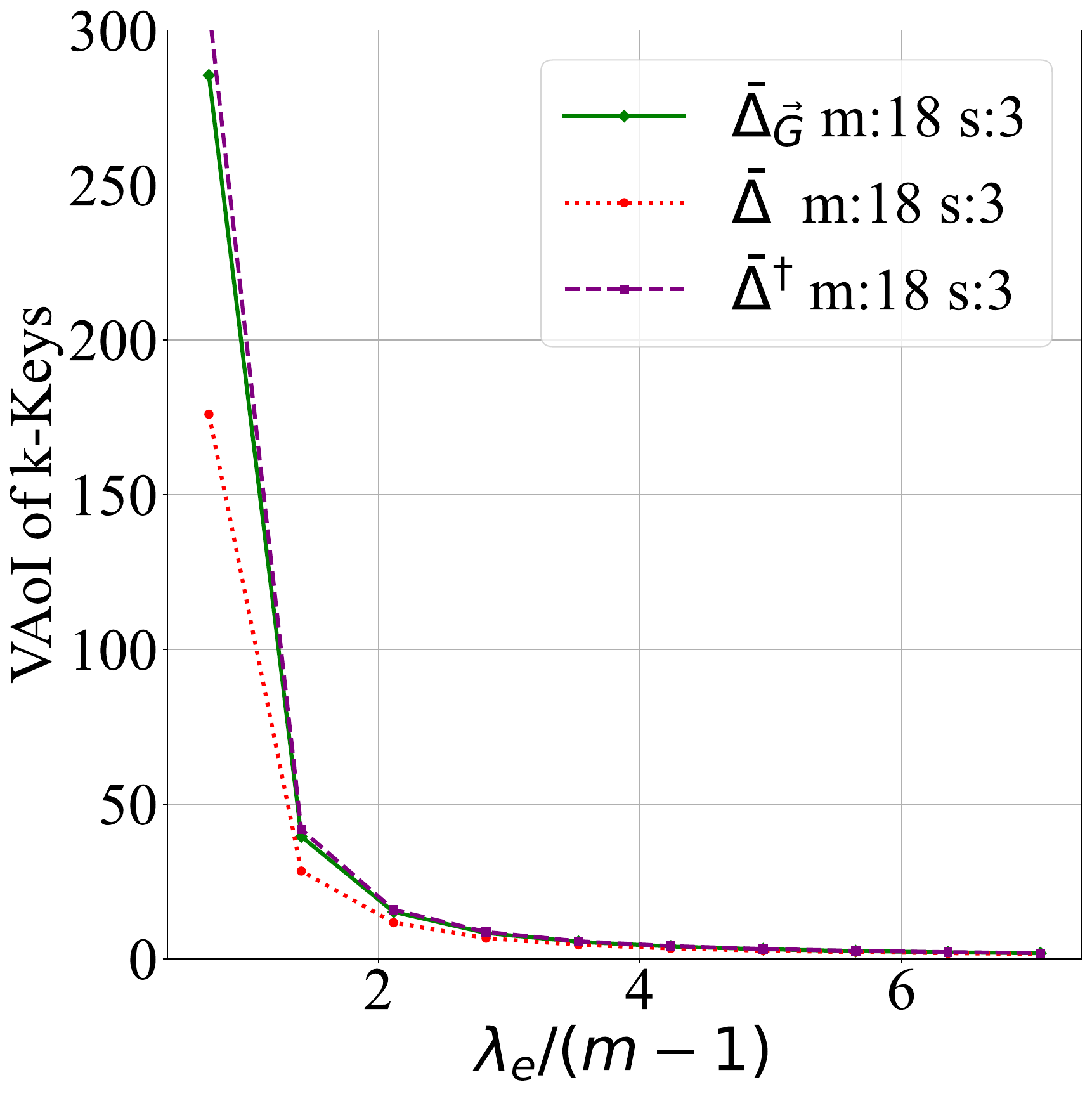}}
\vspace{-0.15cm}
    \subfloat[]{
        \includegraphics[width=0.48\linewidth]{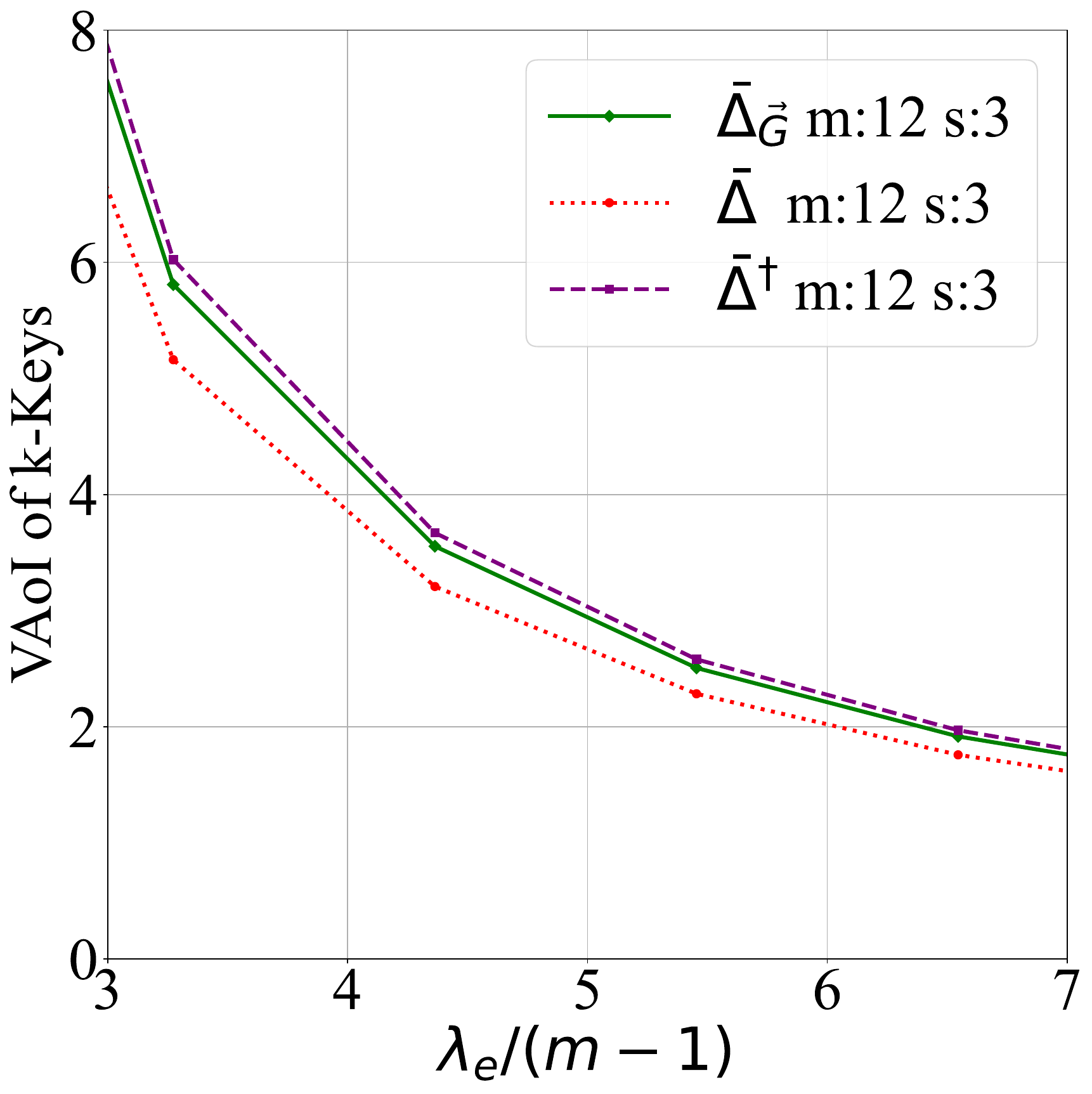}}
     \subfloat[]{   
        \includegraphics[width=0.48\linewidth]{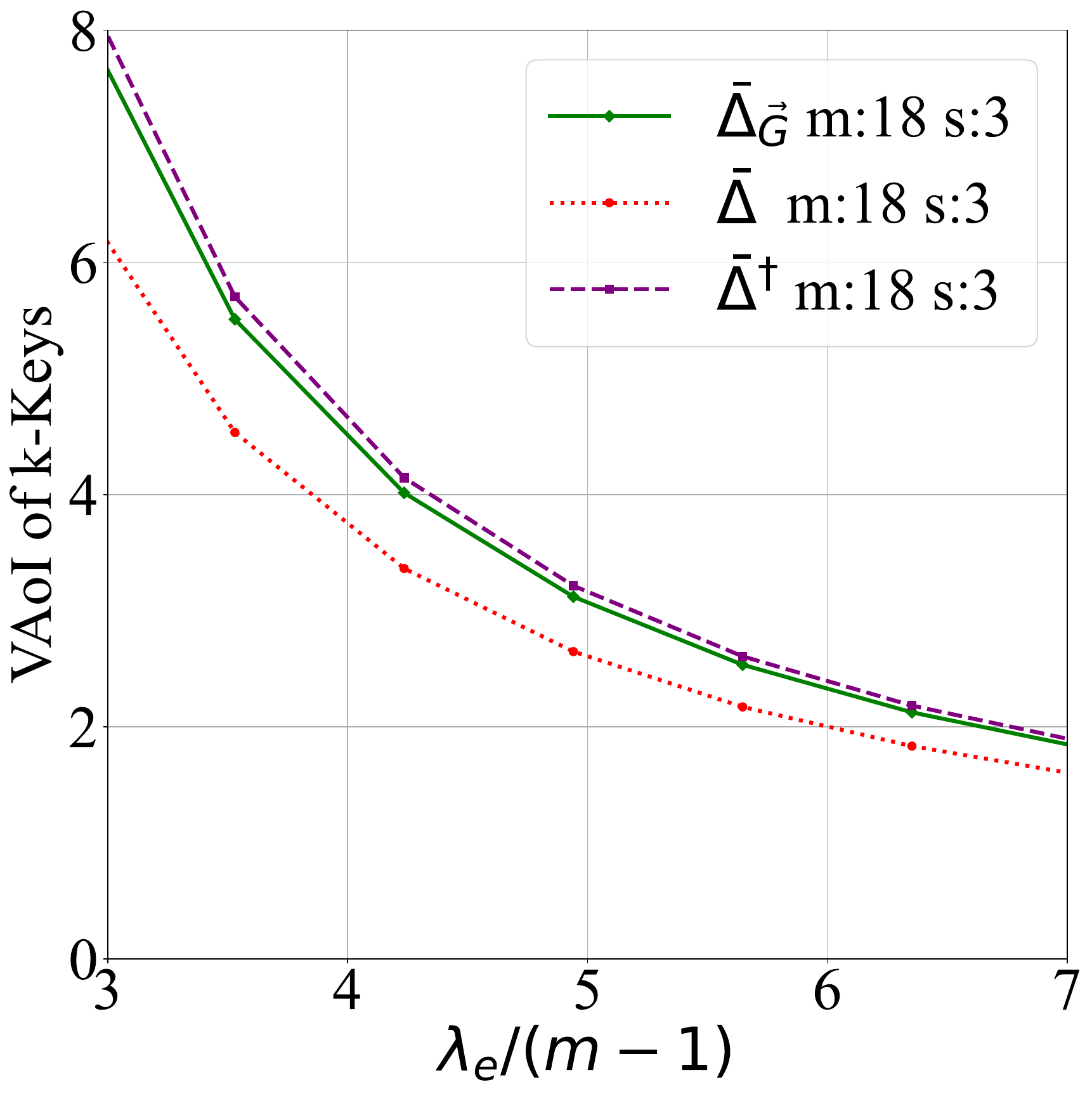}}
    % \vspace{-0.2cm}
	\caption{(a) and (c) show $\ageongraphwo{4}{8}{3}{12}$; (b) and (d) show $\ageongraphwo{4}{8}{3}{18}$ as a function of $\lambda_e$ on a SHN when $\lambda_s=10$, respectively. Simulation results for $\ageongraphwo{4}{8}{s}{m}$, $\agesubswo{4}{8}{s}{m}$, and $\agenonsubswo{4}{8}{s}{m}$ are marked by $\blacklozenge,\bullet,\blacksquare$, respectively while Solid lines dotted lines, and dashed lines show theoretical evaluation in the same order.}
	\label{fig:wo_memory_subset}
 \vspace{-0.1cm}
\end{figure}
\begin{figure}[t]
    \subfloat[]{
        \includegraphics[width=1\linewidth]{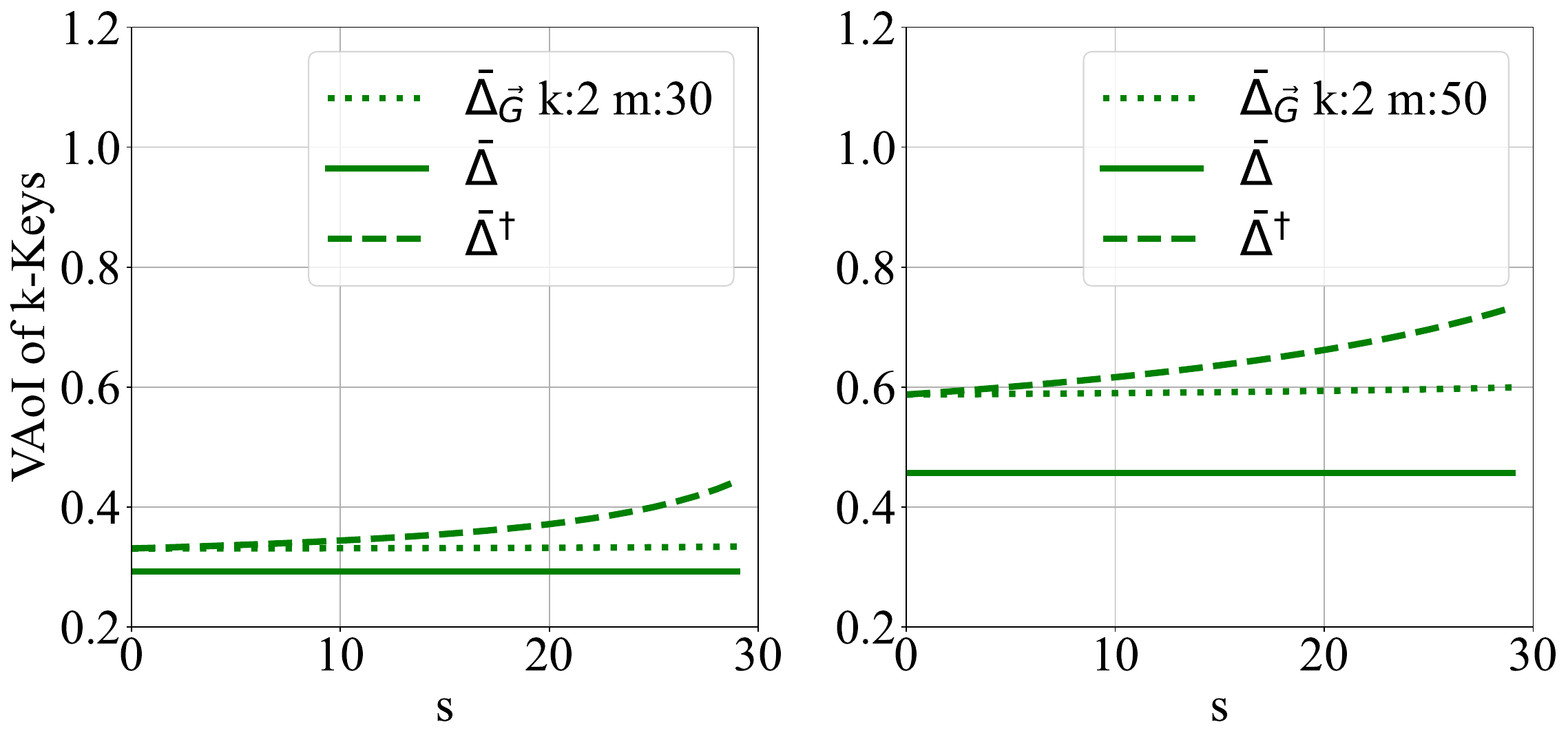}}
    \vspace{-0.1cm}
    \subfloat[]{   
        \includegraphics[width=1\linewidth]{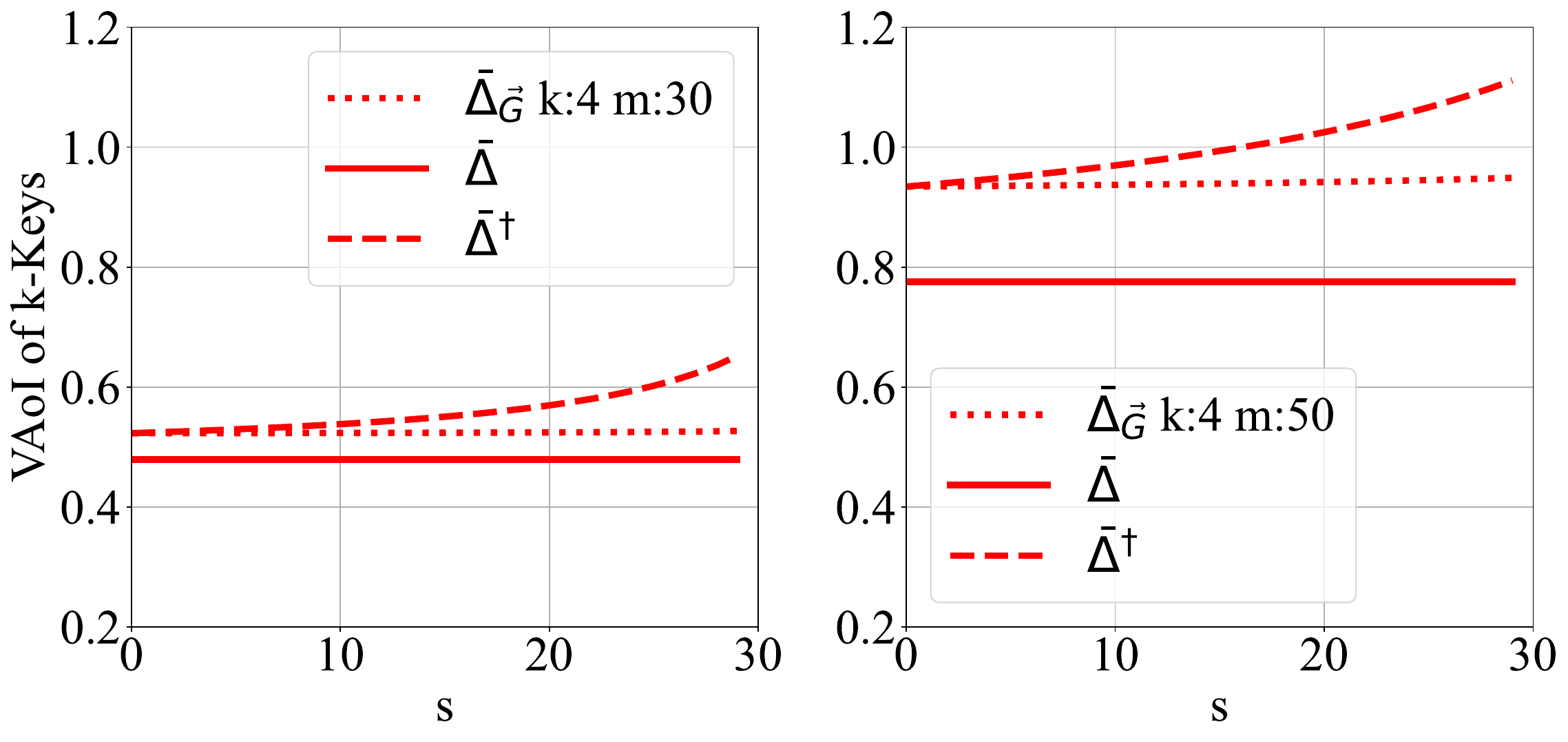}}
	\caption{(a) $\ageongraphwo{2}{30}{s}{m}$, (b) $\ageongraphwo{4}{30}{s}{m}$ as a function of $s$ and $m$ on an SHN when $\lambda_s\!=\!10$,$\lambda_e\!=\!100$ for $s\!=\!\{0,1,\!\cdots,30\}$. Solid lines, dashed lines, and dotted lines show $\agesubswo{k}{n}{s}{m}$, $\agenonsubswo{k}{n}{s}{m}$, and $\ageongraphwo{k}{n}{s}{m}$, respectively.}
	\label{fig:cost_of_subs}
 \vspace{-0.5cm}
\end{figure}

% We have the following corollary of Theorem~\ref{thm:hetero_w_memory_age} and Theorem~\ref{thm:hetero_wo_memory_age}.
% \begin{corollary}\label{cor:compare_scalable}
% For a fixed $\lambda_e,\lambda_s,k$, we have 
% \begin{align*}
%     \lim_{n \to \infty} \bar{\Delta}^k \leq \lim_{n \to \infty} {\Delta}^k

%     % \frac{\lambda_s}{\lambda_e} \sum\limits_{i=0}^{k-1} (\frac{\lambda_e + \lambda_s}{\lambda_e})^i   \mbox{ w.p. } 1.
% \end{align*}
% \end{corollary}
% Clearly it is lower bounded by 

% Then, as a corollary

% \begin{align*}
%     \lim_{n \to \infty}  { \bar{\Delta}(k)  - {\Delta}(k) } = \left( 1 + \frac{ \lambda_s}{\lambda_e}\right)^{k} - \left(  1 +  k\frac{ \lambda_s}{\lambda_e} \right) 
% \end{align*}

% \begin{equation}
%      \lambda^\varepsilon(k,n) := \inf\{\lambda_e \in \mathbb{R}^+ |  |\Delta^k-\bar{\Delta}^k|\leq \varepsilon \}
% \end{equation}

\section{Discussion on The Value of Memory}\label{sec:memory_value}
Finally, we can also compare memory and memoryless schemes. We observe that in the comparisons of Fig.~\ref{fig:bounds_subs_nonsubs} and Fig.~\ref{fig:wo_memory_subset}, $\agesubs{4}{8}{3}{m}$ is lower than $\agesubswo{4}{8}{3}{m}$ for the same gossip rate $\lambda_e$ and $m$. This observation holds true for nonsubscriber nodes as well. 

% To have a more detailed discussion on the value of memory, we consider a full subscription, in which case we have a closed-form expression for the version age of $k$-keys in the following section.

% On the other hand, as the gossip rate $\lambda_e$ increases, the average age in memory scheme in partial subscription and the average age in memoryless scheme in partial key subscription converge.

To quantify the value of memory in a network, we define {\em the memory critical gossip rate} of a $(k,n)$-TSS network for a margin $\varepsilon$, denoted by $\lambda^\varepsilon(k,n,s,m)$, as 
\begin{align}
     \lambda^\varepsilon(k,n,s,m) \!:= \!\inf\{\lambda_e\!\! \in \!\!\mathbb{R}^+ \!:  |\ageongraph{k}{n}{s}{m}\!-\!\ageongraphwo{k}{n}{s}{m}| \!\leq \!\varepsilon  \}. \\[-3em]
\end{align}
We have upper and lower bounds for $\ageongraph{k}{n}{s}{m}$ in~\eqref{thm:w_memory_r_subs}, thus, we can find an upper bound for $\lambda^\varepsilon(k,n,s,m)$ in partial key subscription case $(s<n)$ whereas we can find the exact value in total key subscription case $(s=n)$ by plugging the corresponding values in Corollaries~\ref{cor:w_memory_age} and \ref{cor:wo_memory_age}. 

For the sake of simplicity, we consider full subscription case and we denote $\lambda^\varepsilon(k,n,s,m)$ by $\lambda^\varepsilon(k,n)$ in the numerical results. Fig.~\ref{fig:lambda_epsilon} depicts the memory critical gossip rate of a $(k,30)$-TSS network for different margins as a function of $k$.

We first observe, in Fig.~\ref{fig:lambda_epsilon}, that the memory critical gossip rate $\lambda^\varepsilon(k,n)$ increases as the margin $\varepsilon$ decreases, and it exponentially increases with $k$ for fixed $n$. These observations show that $\ageongraphwo{k}{n}{n}{n}$ approaches $\ageongraph{k}{n}{n}{n}$ as $\lambda_e$ increases or $k$ decreases, consistent with Corollary~\ref{cor:w_memory_age}. Considering the event $E\!:=\!\{\mathcal{X}_{(k:n-1)}\leq U\}$, one can easily see that $Pr(E)$ converges to $1$ as $\lambda_e$ increases (frequent gossipping between nodes) or $k$ decreases for fixed $n$. In this case, the expectation $\mathbb{E}[\min(\mathcal{X}_{(k:n-1)},U)]\!\to\!\mathbb{E}[\mathcal{X}_{(k:n-1)}]$ in~\eqref{eqn:hetero_wo_memory_age}. This implies that $\ageongraphwo{k}{n}{n}{n}$ approach to $\ageongraph{k}{n}{n}{n}$ as $Pr(E)$ goes to $1$.

\begin{figure}
    \centering
    \includegraphics[width=0.9\linewidth]{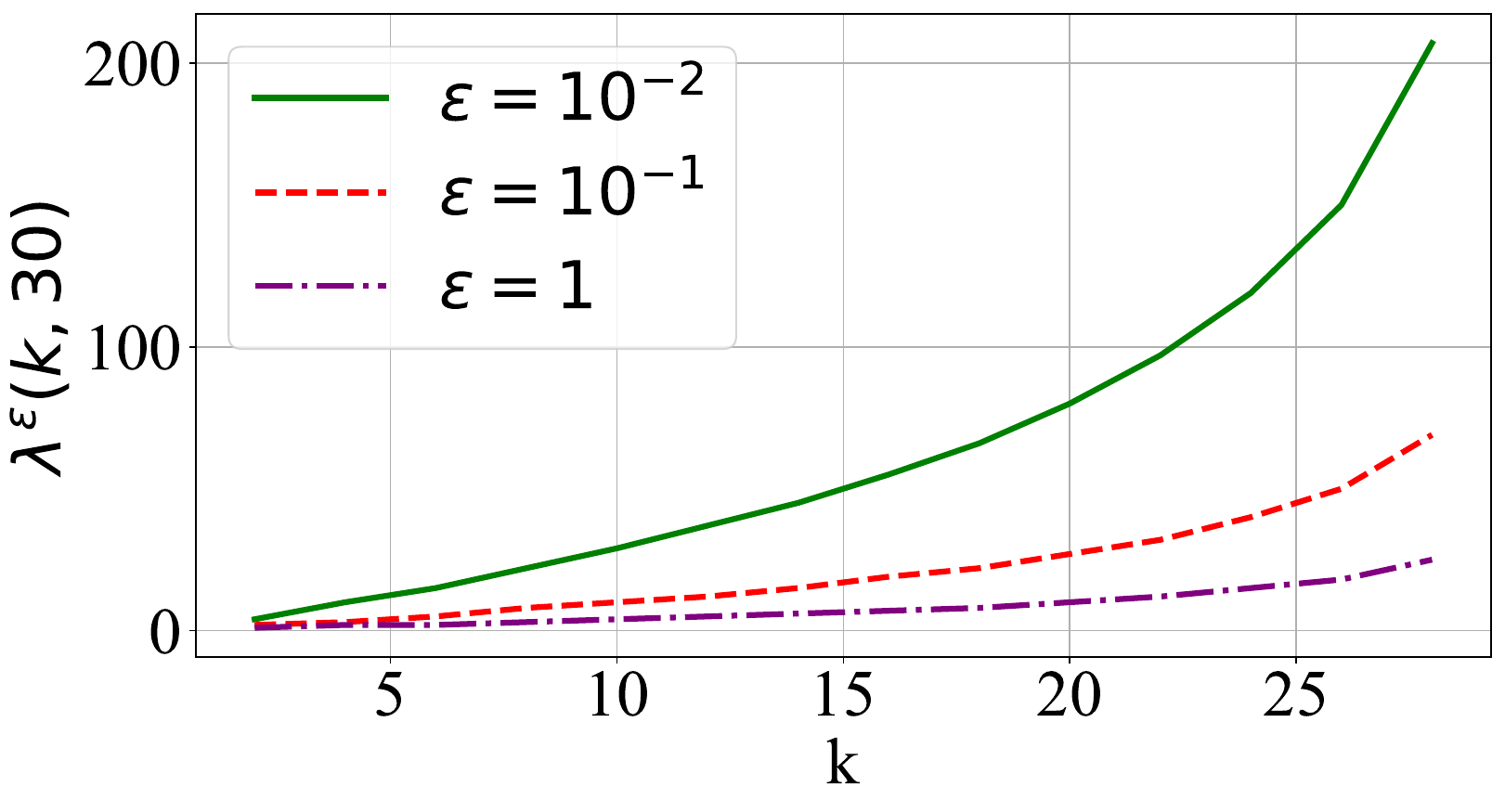}\vspace{-0.35cm}
    \caption{ The memory critical gossip rate $\lambda^\varepsilon(k,30)$ as a function of $k$ for fixed $n=30$ and $\lambda_s=15$. Solid line, dashed line and dotted-dashed line show $\lambda^\varepsilon(k,30)$ for $\varepsilon=10^{-2},\varepsilon=10^{-1}$ and $\varepsilon=1$, respectively.}
    \label{fig:lambda_epsilon}
    \vspace{-0.5cm}
\end{figure}

\section{Conclusion}\label{sec:concl}

We have examined an information update system consisting of $m$-receiver nodes and a single source that encrypts the information by using any $k$-out-of-$n$ threshold system, e.g. $(k,n)$-TSS. For memory scheme, we have provided closed-form expressions for the version age of $k$-keys in full subscription and total key subscription while we have provided tight upper and lower bounds for the same
in the partial key subscription case. For the memoryless scheme, we have provided closed-form expressions for the version age of $k$-keys for any network type. The evaluations have shown that a memory scheme yields a lower version age of $k$-keys compared to a memoryless scheme. Also, we have demonstrated that assigning a node as a subscriber reduces its version age of $k$-keys while exponentially increasing it for the remaining nonsubscriber nodes under a memoryless scheme.

% Also, we have shown that how assigning a node as a subscriber diminishes the version age of $k$-keys for that node while exponentially increasing the version age for the remaining nonsubscriber nodes under memoryless scheme. 

In our work, nodes only send the keys that are received from the source node, ensuring that any set of messages on the channels is not sufficient to decrypt the message at any time. An alternative approach might be to consider the case where nodes can share keys that are received from other nodes. For such a scheme, in future work, we aim to compute the version age of $k$-keys as a function of the number of keys in a message and find the optimal message structure. Another future research direction would be having  more informative keys in the network, using nonlinear decryption schemes.

% to compute the version age of $k$-keys when there are informative keys in the network with nonlinear decrpiton schemes. 

% \newpage
\bibliographystyle{IEEEtran}
\bibliography{aoi.bib}

\vfill

\end{document}